%% file: QPIR-QSS19.tex
\documentclass[onecolumn,12pt]{IEEEtran}

\usepackage[T1]{fontenc}

\usepackage{graphicx}

\usepackage{algorithmB}
\usepackage{algorithm}
\usepackage{algorithmic}

\usepackage{amssymb,amsmath}
\usepackage{amsthm}
\usepackage{latexsym, bm, mathtools, braket, multirow,enumitem, cite}

\usepackage[cmintegrals]{newtxmath}
\usepackage{color}

\input{symbols.tex}
\usepackage{array}

\newcolumntype{M}[1]{>{\centering\arraybackslash}m{#1}}
\newcolumntype{N}{@{}m{0pt}@{}}

\usepackage{caption}
\DeclareCaptionLabelFormat{cont}{#1~#2\alph{ContinuedFloat}}
\usepackage[caption = false]{subfig}

\usepackage{framed}
\usepackage[framemethod=TikZ]{mdframed}

\def\Label#1{\label{#1}\ [\ \text{#1}\ ]\ }
\def\Label{\label}
\newcommand{\red}[1]{#1}

\begin{document}
\title{Unified Approach to Secret Sharing and 
Symmetric Private Information Retrieval with Colluding Servers
in Quantum Systems}


\author{%
Masahito~Hayashi,~\IEEEmembership{Fellow,~IEEE}, and~
Seunghoan Song,~\IEEEmembership{Member,~IEEE} 
\thanks{M. Hayashi is with 
Shenzhen Institute for Quantum Science and Engineering, Southern University of Science and Technology, Nanshan District,
Shenzhen, 518055, China,
International Quantum Academy (SIQA), Futian District, Shenzhen 518048, China,
Guangdong Provincial Key Laboratory of Quantum Science and Engineering,
Southern University of Science and Technology, Nanshan District, Shenzhen 518055, China,
and
Graduate School of Mathematics, Nagoya University, Nagoya, 464-8602, Japan.
(e-mail:hayashi@sustech.edu.cn).}
\thanks{S. Song is with Graduate school of Mathematics, Nagoya University, Nagoya, 464-8602, Japan.}
\thanks{MH is supported in part by the National Natural Science Foundation of China (Grant No. 62171212)
and Guangdong Provincial Key Laboratory (Grant No. 2019B121203002),
SS is supported by JSPS Grant-in-Aid for JSPS Fellows No. JP20J11484 and Lotte Foundation Scholarship.  
}
}

\maketitle

\begin{abstract}
This paper unifiedly addresses two kinds of key quantum secure tasks, i.e., quantum versions of secret sharing (SS) and symmetric private information retrieval (SPIR) by using multi-target monotone span program (MMSP), which characterizes the classical linear protocols of SS and SPIR.
SS has two quantum extensions; One is the classical-quantum (CQ) setting, 
in which the secret to be sent is classical information and the shares are quantum systems.
The other is the quantum-quantum (QQ) setting, in which the secret to be sent is a quantum state and the shares are quantum systems.
The relation between these quantum protocols and MMSP 
has not been studied sufficiently.
We newly introduce the third setting, i.e., the entanglement-assisted (EA) setting, which is defined by modifying the CQ setting with allowing prior entanglement between the dealer and the end-user who recovers the secret by collecting the shares.
Showing that the linear version of SS with the EA setting is directly linked to MMSP, we characterize linear quantum versions of SS with the CQ ad QQ settings via MMSP.
Further, we introduce the EA setting of SPIR, which is shown to link to MMSP.
In addition, we discuss the quantum version of maximum distance separable codes.
\end{abstract}


\section{Introduction}
Recently, quantum information processing technology attracts much attention as a future technology.
In particular, it is considered that quantum information processing technology
has a strong advantage for cryptographic protocols.
Therefore, it is desired to develop an efficient method for construct 
various cryptographic protocols in a unified viewpoint.
This paper focuses on the quantum versions of 
two fundamental cryptographic protocols, 
secret sharing (SS) and private information retrieval (PIR).
Since these are key tools for cryptographic tasks,
their quantum versions are expected to take crucial roles in future quantum technologies.

In SS \cite{Shamir79,Blakley79},
a dealer is required to encode a secret into ${\bar{\snn}}$ shares so that 
the end-user can reconstruct the secret by using some subsets of shares 
but nobody obtains any part of the secret from the other subsets. 
In PIR \cite{CGKS98}, a user is required to retrieve one of the multiple files from server(s) without revealing which file is retrieved.
Since PIR with one server has no efficient solution \cite{CGKS98}, it has been extensively studied with multiple non-communicating servers, and thus, in the following, we simply denote multi-server PIR by PIR.
When 
the user obtains no information other than the retrieved file, the PIR protocol is called {\em symmetric PIR} (SPIR), 
which is also called {\em oblivious transfer} \cite{Rabin81} in the one-server case.

SS and SPIR have a similar structure because the secrecy of both protocols is obtained by 
partitioning the confidential information.
	On the other hand, the two protocols have a different structure because
        in SS, the secret is both the confidential and targeted information 
        but in PIR, the targeted file is not confidential.
Using the similarity, several studies constructed PIR protocols from SS protocols \cite{GIKM00, BIKO12,LMD17,YSL18, DER18}.
Recently, the paper \cite{SH2022} derived an equivalence relation between 
\red{linear SS protocols and 
linear} SPIR protocols even with general access structure.
In this equivalence, \red{
all linear SS protocols and a special class of linear SPIR protocols
are algebraicly characterized by  
using multi-target monotone span program (MMSP) \cite{Beimel11,BI93,Dijk95}.}

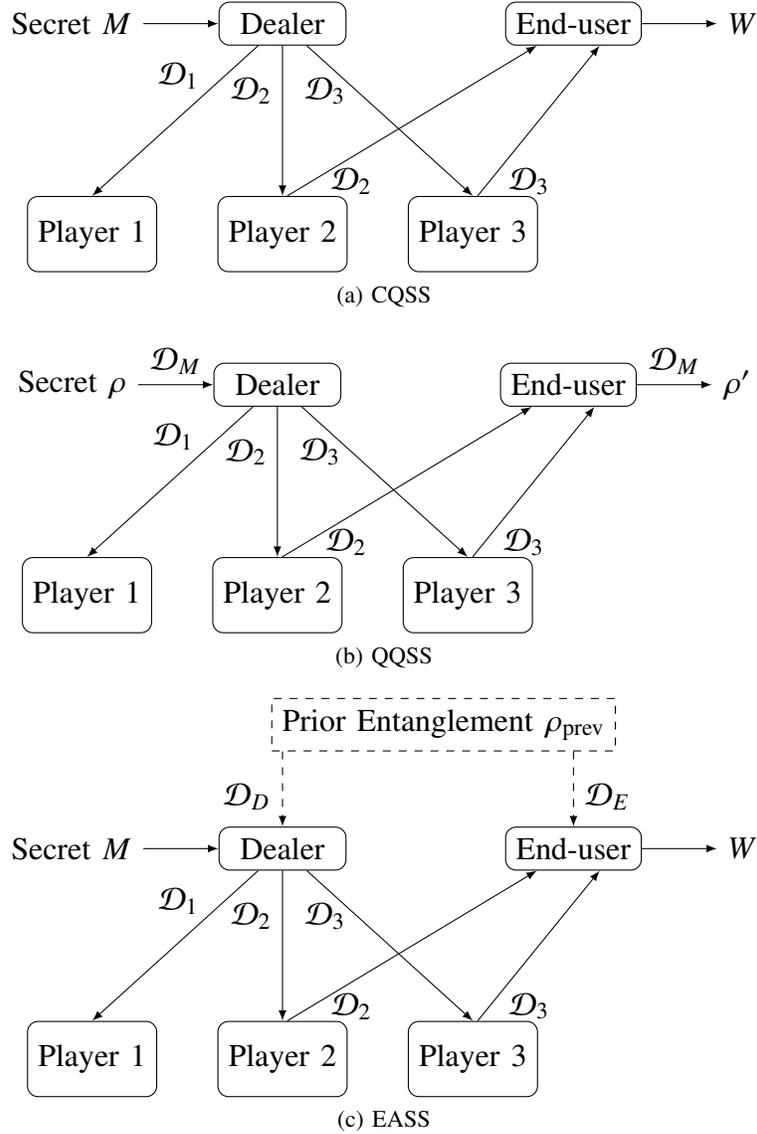
\begin{figure}
\begin{center}
\subfloat[CQSS]{
\begin{tikzpicture}[node distance = 3.3cm, every text node part/.style={align=center}, auto]

    \node [block] (dealer) {Dealer};

    \node [block, right=5em of dealer] (user) {End-user};
    
    \node [block,minimum height = 1cm, below=2cm of dealer] (serv2) {Player 2};
    \node [block,minimum height = 1cm, left=0.8cm of serv2] (serv1) {Player 1};
    \node [block,minimum height = 1cm, right=0.8cm of serv2] (serv3) {Player 3};
    
    \node [right=1cm of user] (receiv) {$W$};
    \node [left=1cm of dealer] (secret) {Secret $M$};

    
    \path [line] (dealer) --node[pos=0.2,left=2mm] {$\cD_1$} (serv1.north);
    \path [line] (dealer.south) --node[pos=0.3,left] {$\cD_{2}$} (serv2.north);
    \path [line] (dealer) --node[pos=0.3,left] {$\cD_{3}$} (serv3.north);

    \path [line] (serv2.83) --node[pos=0.1,right=1mm] {$\cD_{2}$} (user);
    \path [line] (serv3.83) --node[pos=0.1,right=1mm] {$\cD_{3}$} (user.320);
    
    \path [line] (user.east) -- (receiv);
    \path [line] (secret) -- (dealer);
    
%
%
%
    
\end{tikzpicture}
} \\

\subfloat[QQSS]{
        \begin{tikzpicture}[node distance = 3.3cm, every text node part/.style={align=center}, auto]
    \node [block] (dealer) {Dealer};

    \node [block, right=5em of dealer] (user) {End-user};
    
    \node [left=1cm of dealer] (secret) {Secret $\rho$};
    \node [right=1cm of user] (receiv) {$\rho'$};

    \node [block,minimum height = 1cm, below=2cm of dealer] (serv2) {Player 2};
    \node [block,minimum height = 1cm, left=0.8cm of serv2] (serv1) {Player 1};
    \node [block,minimum height = 1cm, right=0.8cm of serv2] (serv3) {Player 3};

    
    \path [line] (dealer) --node[pos=0.2,left=2mm] {$\cD_1$} (serv1.north);
    \path [line] (dealer.south) --node[pos=0.3,left] {$\cD_{2}$} (serv2.north);
    \path [line] (dealer) --node[pos=0.3,left] {$\cD_{3}$} (serv3.north);

    \path [line] (serv2.83) --node[pos=0.1,right=1mm] {$\cD_{2}$} (user);
    \path [line] (serv3.83) --node[pos=0.1,right=1mm] {$\cD_{3}$} (user.320);
    
    \path [line] (user.east) --node {$\cD_M$} (receiv);
    \path [line] (secret) -- node {$\cD_M$} (dealer);
\end{tikzpicture}
}\\
\subfloat[EASS]{
\begin{tikzpicture}[node distance = 3.3cm, every text node part/.style={align=center}, auto]
    \node [block] (dealer) {Dealer};

    \node [block, right=5em of dealer] (user) {End-user};
    
    \node [block,minimum height = 1cm, below=2cm of dealer] (serv2) {Player 2};
    \node [block,minimum height = 1cm, left=0.8cm of serv2] (serv1) {Player 1};
    \node [block,minimum height = 1cm, right=0.8cm of serv2] (serv3) {Player 3};
    
    \node [draw, dashed, above right=1cm and -1cm of dealer] (shared) {Prior Entanglement $\rho_{\mathrm{prev}}$};
    
    \node [right=1cm of user] (receiv) {$W$};
    \node [left=1cm of dealer] (secret) {Secret $M$};

    
    \path [line] (dealer) --node[pos=0.2,left=2mm] {$\cD_1$} (serv1.north);
    \path [line] (dealer.south) --node[pos=0.3,left] {$\cD_{2}$} (serv2.north);
    \path [line] (dealer) --node[pos=0.3,left] {$\cD_{3}$} (serv3.north);

    \path [line] (serv2.83) --node[pos=0.1,right=1mm] {$\cD_{2}$} (user);
    \path [line] (serv3.83) --node[pos=0.1,right=1mm] {$\cD_{3}$} (user.320);
    
    \path [line] (user.east) -- (receiv);
    \path [line] (secret) -- (dealer);

%
    
    \path [line,dashed] (shared.south -|dealer) -| node[pos=0.8,left=0mm] {$\cD_D$} (dealer);
    
    \path [line,dashed] (shared.south -|user) -| node[pos=0.8,right=0mm] {$\cD_E$} (user);

\end{tikzpicture}
}
\caption{Quantum SS protocols where the end-user receives the shares from Player 2 and Player 3.
Fig. (a), (b), and (c) show a CQSS protocol, a QQSS protocol, and an EASS protocol, respectively.
The notations in the above figures will be defined in Section \ref{S4}.
}   \label{fig:SS}
\end{center}
\end{figure}

As quantum versions of SS, 
existing studies investigated two problem settings.
One is classical-quantum SS (CQSS), in which the secret message to be sent is given as classical information 
\cite{HBB99, Gottesman00, MS08, KFMS10, Sarvepalli12, MM13,KKI99, KMMP09, WCY14, Matsumoto17, Matsumoto20}, which is illustrated as Fig. \ref{fig:SS} (a).
The other is quantum-quantum SS (QQSS), in which the secret message to be sent is given as a quantum state \cite{HBB99, Gottesman00, MS08, KFMS10, Sarvepalli12, MM13,CGL99, Smith00, GJMP15, ZM15, Matsumoto18,OSIY05,SR10}, which is illustrated as Fig. \ref{fig:SS} (b).
The studies \cite{Smith00, Sarvepalli12, WCY14, ZM15, Matsumoto17, Matsumoto18,  Matsumoto20} discussed the security of these problem setting by using general access structure.
Although the papers \cite{HBB99, Gottesman00, MS08, KFMS10, Sarvepalli12, MM13}
studied both settings, no exiting study a unified framework for both problem settings.
\red{That is, no preceding study clarified algebraic structure of 
CQSS protocols and QQSS protocols. 
In fact, since linear SS protocols are characterized algebraicly by MMSP completely,
we can expected that CQSS protocols and QQSS protocols 
can be characterized by a variant of MMSP. 
That is, it is expected that such characterization would be helpful
to understand what a type of CQSS and QQSS are possible.
{\bf However, such useful characterizations of 
CQSS and QQSS protocols with general access structure by MMSP has not been obtained.}
In addition, 
for classical SS protocol, ramp-type SS protocols 
have been actively studied in \cite{BM,IY,OKT,Stinson,Yamamoto}.
However, while its CQSS and QQSS versions were introduced,
their analysis is very limited and did not discuss general access structure
\cite{Matsumoto17,Matsumoto20,ZM15,Matsumoto18}.}

In this paper, \red{to resolve the above problems for 
CQSS protocols and QQSS protocols from a unified framework,}
as illustrated in Fig. \ref{fig:SS} (c),
we introduce
the third problem setting, entanglement-assisted SS (EASS), in which the secret message to be sent is given as a classical information, and prior entanglement is allowed between the dealer and the end-user who intends to decode the message while 
the above two problem settings allow no prior entanglement.
Analyzing two special cases of EASS, we derive our analyses of CQSS and QQSS.
Here, CQSS is simply given as a special case of EASS.
In contrast, we derive a notable conversion between 
QQSS and a special case of EASS by considering 
notable relations between dense coding and noiseless quantum state transmission.

To cover the security with general access structure, 
we study the relation between the security of EASS and the property of MMSP 
under linearity condition
while the paper \cite{Smith00} discussed this relation with a special case of access structure.
\red{For this analysis, 
we introduce a new concept, the symplectification for each access structure.
Through the symplectification for each access structure, 
linear EASS protocols are characterized by MMSP
because the symplectic structure plays a central role in this problem
although such a symplectification for an access structure
has not been considered by any existing study.
Then, using this concept, we characterize 
CQSS protocols.
To characterize QQSS protocols, we additionally invent new relations between
dense coding and noiseless quantum state transmission.
Then, we characterize QQSS protocols by combining our obtained characterization 
for EAQQ protocols and the above relations.}
In addition, we clarify the relation between 
QQSS protocols and quantum maximum distance separable (MDS) codes
while a special case of such relations was mentioned in \cite{Gottesman00}.

As quantum versions of SPIR, many existing papers \cite{KdW04,SH19,SH19-2, SH20,AHPH20, KL20, ASHPHH21,WKNL21-1, WKNL21-2} studied 
classical-quantum SPIR (CQSPIR), in which the file to be sent is given as a classical information.
However, no existing paper studied the relation between CQSPIR and quantum versions of SS
while such relation in the classical version was studied in \cite{SH2022}.
\red{In addition, 
several existing papers \cite{SH2022,TGKFHE17,ZG17,SJ18-col,YLK20,CLK20}
for the classical setting considered
the reconstruction of the message only from the answers from 
a part of servers, which is called a qualified set of servers.
Also, several existing papers \cite{SH2022,BU19-2,HFLH21,TGKFH19,WS18} considered
various cases for the set of colluding servers.
Therefore, the analysis with various qualified sets of servers and various sets of colluding servers can be considered as a hot topic in the area of SPIR.
However,
{\bf no existing paper studied the reconstruction of a CQSPIR protocol with a general qualified set of servers}
because all existing papers \cite{KdW04,SH19,SH19-2,SH20,AHPH20,ASHPHH21,KL20,WKNL21-1,WKNL21-2} of the quantum setting 
considered this task under the condition that all servers send the answer to the user.}
In this paper, to develop the above relation, as another quantum setting of SPIR,
we introduce entanglement-assisted SPIR (EASPIR), in which the file to be sent is given as a classical information, and prior entanglement is allowed between the servers and the user while 
CQSPIR does not allow such prior entanglement.
\red{Although the papers \cite{KLLGR16,ABCGLS19,SLH} 
considered use of prior entanglement between the server and the user
and the papers \cite{KLLGR16,SLH} 
showed its great 
advantages over the case without prior entanglement 
in the quantum non-symmetric PIR,
no existing paper discussed the use of this type of prior entanglement 
in the quantum SPIR.}
Using MMSP, we derive the conversion relation between EASS and EASPIR protocols under the linearity condition.
Due to this conversion, we address these two settings under general access structure.
In SPIR, general access structure characterizes 
what a set of servers is qualified to recover the file information
and what a set of colluded servers is disqualified to identify what file the user wants.
That is, under this problem setting, we can discuss the case when 
only a part of servers answers the query sent by the user.
In fact, no existing study addressed general access structure for CQSPIR
because all existing studies \cite{KdW04,SH19,SH19-2, SH20,AHPH20, KL20, ASHPHH21,WKNL21-1, WKNL21-2} considered only the case when the file information is recovered with answers from all severs.
This paper is the first paper to address the case when
a set of servers is qualified to recover the file information.

One may consider that these two problem settings with entanglement assistance 
are artificial. However, these two problem settings take key roles to discuss other three problem settings,
CQSS, QQSS, and CQSPIR.
That is, use of these two settings enables us to derive the above characterizations by MMSP and their variants. 
\red{In other words, two problem settings with entanglement assistance work
as theoretical tools to analyzing 
three settings, CQSS, QQSS, and CQSPIR
because our analysis for CQSS, QQSS, and CQSPIR does not work
without considering entanglement assisted settings.
In this sense, considering entanglement assisted settings is a key part of this paper.}

The remainder of the paper is organized as follows.
Section~\ref{Preliminaries} prepares several notations and prior knowledges for our analysis on quantum versions of 
SS and SPIR. 
Section~\ref{S4} formally describes CQSS, QQSS, and  
CQSPIR protocols.
Section~\ref{S5} reviews the existing results under various classical settings as special cases of 
our quantum settings.
Section~\ref{SS-5} presents our results for CQSS, QQSS, and CQSPIR protocols
without their proofs.
Section~\ref{S6-1} introduces EASS protocols and presents our results for EASS protocols, which implies our results for CQSS protocols.
Section~\ref{SS-7} presents the proofs for our results of QQSS protocols stated in
Section~\ref{SS-5} by showing notable relations between dense coding and noiseless quantum state transmission.
Section~\ref{S6-3} introduces EASPIR protocols and presents our results for EASPIR protocols, which implies our results for CQSPIR protocols.
Section \ref{ZMTA} presents a conversion theorem from EASPIR protocols to EASS protocols as a part of our results.
Section~\ref{Sec-Ex} addresses examples for problem settings.
Section~\ref{sec:conclusion} makes conclusions and discussion.
Appendix is devoted for the proofs of Theorems \ref{TH3CQ}, \ref{TH6}, and \ref{TH3},
which are stated in Sections \ref{SS-5} and \ref{S6-1} and guarantees
the existences of several types of MMSPs.

\section{Preliminaries} \Label{Preliminaries}
\subsection{Vector space and matrix over finite field} \Label{S3-1}
In this paper, our information is described as an element of a vector space on a finite field
$\mathbb{F}_q$
whose order is a prime power $q = p^r$
because we address linear protocols.
For preliminary for our analysis, we prepare several notations over 
the finite field $\mathbb{F}_q$.
First, we define $\tr z \coloneqq \Tr T_z\in\mathbb{F}_p$ for $z\in\mathbb{F}_q$, where $T_z\in\mathbb{F}_p^{r\times r}$ denotes the matrix representation of the linear map $y\in\mathbb{F}_q \mapsto zy \in\mathbb{F}_q$ 
by identifying the finite field $\mathbb{F}_q$ with the vector space 
$\mathbb{F}_p^{r}$.

\red{\begin{exam}
Consider the algebraic extension $\in\mathbb{F}_q$ of $\in\mathbb{F}_p$
with an irreducible polynomial 
$f(x)=-x^r+a_{r-1} x^{r-1}+\cdots + a_1 x +a_0 $, i.e., 
$q=p^r$.
The finite field $\mathbb{F}_q$ is written as a vector space $\mathbb{F}_p^r$
with basis $1, x, \ldots, x^{r-1}$.
Then, for $\alpha \in \{0,\ldots, r-1\}$ we have  
\begin{align}
x^\alpha x
=
\left\{
\begin{array}{ll}
 x^{\alpha+1} & \hbox{ when } \alpha \le r-2 \\
a_{r-1} x^{r-1}+\cdots + a_1 x +a_0  & \hbox{ when } \alpha= r-1 .
\end{array}
\right.
\end{align}
Therefore, $\Tr T_x= a_{r-1}$, i.e., $\tr x=a_{r-1} $.
\end{exam}}

A linear map from a vector space $\FF_q^{\sy}$ to 
$\FF_q^{{\bar{\snn}}}$ is written as a ${\bar{\snn}} \times \sy$ matrix $G$.
We say that $(G,F)$ is a ${\bar{\snn}} \times (\sy+\sx)$ matrix
when $G$ is an ${\bar{\snn}} \times \sy$ matrix and $F$ is an ${\bar{\snn}} \times \sx$ matrix.
The column vectors of $G$ are written as $g^1, \ldots, g^{\sy} 
\in \FF_q^{\red{\bar{\snn}}}$.
The image $\im G$ of $G$ is given as the vector space spanned by
$g^1, \ldots, g^{\sy} $.
Once the linear map $G$ is given,
we have an subspace $\im G \subset \FF_q^{{\bar{\snn}}}$.
When we identify vectors $h^1,h^2 \in \FF_q^{{\bar{\snn}}}$ satisfying $h^1- h^2 \in \im G$,
we define the quotient vector space $ \FF_q^{{\bar{\snn}}}/ \im G$.
Considering this identification, we define the natural map from 
 $ \FF_q^{{\bar{\snn}}}$ to  $ \FF_q^{{\bar{\snn}}}/ \im G$.
 This map is written as $\pi[\im G]$.
 
For a given positive integer ${\bar{\snn}}$, we denote the set 
$\{1, \ldots, {\bar{\snn}}\}$ by $[{\bar{\snn}}]$. 
For a subset $\red{\cC}$ of $[{\bar{\snn}}] $, we define the linear map $P_{\red{\cC}}$,
from $\FF_q^{{\bar{\snn}}}$ to $\FF_q^{|\red{\cC}|}$ as follows.
Given a vector $g=(g_1, \ldots, g_{{\bar{\snn}}})^T\in \FF_q^{{\bar{\snn}}}$,
the vector $P_{\cA} g\in \FF_q^{|\red{\cC}|}$ is defined as
$P_{\cC}g=(g_s)_{s \in \red{\cC}}\in \FF_q^{|\red{\cC}|}$.

Next, we consider the case with ${\bar{\snn}}=2\snn$.
For $x, y \in \mathbb{F}_q^{\snn}$,
	we denote $\langle x, y \rangle \coloneqq \tr \sum_{i=1}^{\snn} 
	x_iy_i\in\mathbb{F}_p$.
Then, we define the symplectic inner product
$( (x,y),(x',y') ):=  \langle x, y' \rangle-\langle x', y \rangle \in \FF_p$.
We say that
a vector $ g \in \FF_q^{{\bar{\snn}}}$
is orthogonal to another vector $ h \in \FF_q^{{\bar{\snn}}}$
in the sense of the symplectic inner product
when $(g,h)=0$.
We say that a matrix $G=(g^1, \ldots, g^{\sy}) \in \FF_q^{{\bar{\snn}} \times \sy}$ is self-column-orthogonal
when vectors $g^1, \ldots, g^{\sy}$ are orthogonal to each other in the sense of the symplectic inner product.
In addition, 
we say that an ${\bar{\snn}} \times \sx$ matrix $F$ is column-orthogonal to an ${\bar{\snn}} \times \sy$ matrix $G$
when vectors $f^1, \ldots, f^{\sx}$ are orthogonal to 
vectors $g^1, \ldots, g^{\sy}$ in the sense of the symplectic inner product.

\subsection{Fundamentals of Quantum Information Theory} \Label{sec:quantum-basic}

In this subsection, we briefly introduces the fundamentals of quantum information theory.
More detailed introduction can be found at \cite{NC00, Hay17}.
A quantum system is a Hilbert space $\cH$.
	Throughout this paper, we only consider finite dimensional Hilbert spaces.
A quantum state is defined by a {\em density matrix}, which is a Hermitian matrix $\rho$ on $\cH$ such that 
\begin{align}
\rho \geq 0, \quad \Tr \rho  = 1.
\end{align}
The set of states on $\cH$ is written as $\mathcal{S}(\cH)$.
A state $\rho$ is called a {\em pure state} if $\rank \rho = 1$, which can also be described by a unit vector of $\cH$.
If a state $\rho$ is not a pure state, it is called a {\em mixed state}.
The state $\rho_{\mathrm{mix}} \coloneqq  I_{\cH} / \dim \cH$ is called the completely mixed state.
The composite system of two quantum systems $\cA$ and $\cB$ is given as the tensor product of the systems $\cA\otimes \cB$.
For a state $\rho\in \mathcal{S}(\cA\otimes \cB)$, 
	the {\em reduced state} on $\cA$ is written as 
	\begin{align}
	\rho_{\cA} = \Tr_{\cB} \rho,
	\end{align}
where $\Tr_{\cB}$ is the partial trace with respect to the system $\cB$.

A state $\rho \in \mathcal{S}(\cA\otimes \cB)$ is called a {\em separable state} if 
	$\rho$ is written as 
	\begin{align}
	\rho = \sum_i p_i \rho_{\cA,i} \otimes \rho_{\cB,i},
	\end{align}
	for some distribution $p = \{p_i\}$ and states $\rho_{\cA,i} \in \mathcal{S}(\cA)$, $\rho_{\cB,i}\in \mathcal{S}(\cB)$.
A state $\rho \in \mathcal{S}(\cA\otimes \cB)$ is called an {\em entangled state} if it is not separable.

A quantum measurement is defined by a {\em positive-operator valued measure (POVM)}, which is the set of Hermitian matrices $\{\Pi_{x}\}_{x\in\mathcal{X}}$ on $\cH$ such that 
\begin{align}
\Pi_x \geq 0, \quad \sum_{x \in \mathcal{X}}  \Pi_x = I_{\cH}.
\end{align}
When the elements of POVM are orthogonal projections, i.e., $\Pi_x^2 = \Pi_x$ and $\Pi_x^\dagger = \Pi_x$,
	we call the POVM a {\em projection-valued measure (PVM)}.
A quantum operation is described by a {\em trace-preserving completely positive (TP-CP) map} $\kappa$, which is a linear map such that 
	\begin{align}
	\Tr \kappa(\rho) &= 1 \quad \forall \rho \in \mathcal{S}(\cH), \\
	\kappa \otimes \iota_{\mathbb{C}^{d}} (\rho) &\geq 0 \quad \forall \rho \in \mathcal{S}(\cH\otimes \mathbb{C}^d), \ \forall d \geq 1,
	\end{align}
	where $\iota_{\mathbb{C}^{d}}$ is the identity map on $\mathbb{C}^d$.
We often omit the identity map $\iota_{\mathbb{C}^{d}}$.
An example of TP-CP maps is the unitary map, which is defined by $\kappa_U (\rho) = U \rho U^*$ for a unitary matrix $U$. 

\subsection{Stabilizer formalism over finite fields} \Label{sec:stabilizer}

In this subsection, we introduce the stabilizer formalism for finite fields.
Stabilizer formalism gives an algebraic structure for quantum information processing.
We use this formalism for the construction of the QPIR protocol. 
Stabilizer formalism is often used for quantum error-correction.
More detailed introduction of the stabilizer formalism can be found at \cite{CRSS98,AK01, KKKS06, Haya2}.

In this paper, we denote 
the $q$-dimensional Hilbert space with a basis $\{ |j\rangle \mid  j\in \mathbb{F}_q \}$
by $\cH$.
For $a,b\in\mathbb{F}_q$, we define unitary matrices on $\cH$ 
\begin{align}
\mathsf{X}(a) &\coloneqq \sum_{j\in\mathbb{F}_q} |j+a\rangle \langle j |, \quad 
\mathsf{Z}(b) \coloneqq \sum_{j\in\mathbb{F}_q} \omega^{\tr bj} |j\rangle \langle j |,\\
\mathsf{W}(a,b) &\coloneqq \mathsf{X}(a)\mathsf{Z}(b),
\end{align}
where $\omega \coloneqq \exp({2\pi i/p})$.
For ${a} = (a_1,\ldots, a_{\snn}),~
b = (b_1,\ldots, b_{\snn}) \in\mathbb{F}_q^{\snn}$,
	and ${w} = ({a},{b}) \in \mathbb{F}_q^{2\snn}$, we define a unitary matrix on 
	$\cH^{\otimes {\snn}}$
\begin{align*}
    &\mathbf{W}_{[\snn]}(w) = \mathbf{W}_{[\snn]}(a,b) \\
    &\coloneqq 
	\mathsf{X}(a_1) \mathsf{Z}(b_1) \otimes \mathsf{X}(a_2) \mathsf{Z}(b_2)   \otimes \cdots \otimes \mathsf{X}(a_n) \mathsf{Z}(b_n) 
.
\end{align*}
Since $\mathsf{X}(a)\mathsf{Z}(b) = \omega^{-\tr ab} \mathsf{Z}(b)\mathsf{X}(a)$,
	for any $({a}, {b}), ({c},{d}) \in \mathbb{F}_q^{2 \snn}$,
	we have 
\begin{align}
\mathbf{W}_{[\snn]}(a,b)\mathbf{W}_{[\snn]}(c,d) &= 
            \omega^{((a,b),(c,d))} 
            \mathbf{W}_{[\snn]}(c,d) \mathbf{W}_{[\snn]}(a,b) \Label{eq:commutative}.
\end{align}
When the above operator is defined on a subset $\cA \subset [\snn]$,
it is written as $\mathbf{W}_{\cA}(w')$ with $w' \in \FF_q^{2 |\cA|}$.

\subsection{Information quantities}
To discuss information leakage, we often employ the mutual information.
To address the mutual information, we prepare the quantum relative entropy.
For two states $\rho$ and $\sigma$ on the quantum system $\cH$, 
	the quantum relative entropy is defined as
\begin{align*}
	D(\rho \| \sigma) &\coloneqq  
	\begin{cases}
	\Tr \rho (\log \rho - \log \sigma) &\text{if }  \supp(\rho) \subset \supp(\sigma)\\
	\infty 		&\text{otherwise},
	\end{cases}
\end{align*}
	where $\supp(\rho) \coloneqq \{ |x\rangle \in \cH \mid \rho |x\rangle \neq 0 \}$.

When the state on the joint system of two quantum systems $\cH_A$ and $\cH_B$ 
is given as $\rho_{AB}$, the mutual information
$I(A;B)[\rho_{AB}]$ is defined as
\begin{align}
I(A;B)[\rho_{AB}]:= D( \rho_{AB} \| \rho_A \otimes \rho_B),
\end{align}
where $\rho_A:=\Tr_B \rho_{AB}$ and $\rho_B:=\Tr_A \rho_{AB}$. 
The above quantity will be employed for the discussion on the relation between two types of SS protocols.

\subsection{Access structure}
In this paper, we discuss a general access structure when ${\bar{\snn}}$ players or ${\bar{\snn}}$ servers 
exists.
The family of subsets of $[{\bar{\snn}}]$ is identified with $\red{\{0,1\}^{[{\bar{\snn}}]}}$.
We call $\fA \subset \red{\{0,1\}^{[{\bar{\snn}}]}}$ a {\em monotone increasing collection} 
	if $\cA \in \fA$ implies $\cC\in\fA$ for any $\cA \subset \cC \subset[{\bar{\snn}}]$.
In contrast, we call $\fB \subset \red{\{0,1\}^{[{\bar{\snn}}]}}$ a {\em monotone decreasing collection} 
	if $\cB \in \fB$ implies $\cC\in\fB$ for any $\cC \subset \cB$.
In addition, we call $\fC \subset \red{\{0,1\}^{[{\bar{\snn}}]}}$ a $\sff$-{\em collection} 
	if $\cC \in \red{\fC}$ implies $|\cC|=\sff$.
An {\em access structure} on $[{\bar{\snn}}]$ is defined as a pair 
 $(\ACC, \REJ)$
of monotone increasing and 
decreasing collections $\REC$ and $\COL \subset \red{\{0,1\}^{[{\bar{\snn}}]}}$ such that $\REC \cap \COL = \emptyset$.

When ${\bar{\snn}}$ is an even number $2\snn$,
for a subset $\cA \in 2^{\snn}$, we denote its cardinality by $|\cA|$.
\red{Then, we define the symplectifications of 
a subset $\overline{\cA} \in 2^{{\bar{\snn}}}$ 
and an access structure $\fA$ as follows.}
When $\cA=\{ a_1, \ldots, a_l\}\subset [\snn] $, 
$\overline{\cA}$ is defined as $\{ a_1, \ldots, a_l, a_1+\snn, \ldots, a_l+\snn\}\subset 
[{\bar{\snn}}] $. 
Then, given a monotone increasing collection $\fA \subset 2^{[\snn]}$,
we define 
a monotone increasing collection $\overline{\fA} \subset \red{\{0,1\}^{[{\bar{\snn}}]}}$ as follows.
When $\fA=\{ \cA_1, \ldots, \cA_l\}\subset 2^{[\snn]}$, 
$\overline{\fA}$ is defined as $\{ \overline{\cA}_1, \ldots, \overline{\cA}_l\} \subset \red{\{0,1\}^{[{\bar{\snn}}]}}$, 
\red{which is called the symplectification of $\fA$.}
For a monotone decreasing collection $\fB \subset 2^{[\snn]}$,
we define
the monotone decreasing collection $\overline{\fB} \subset 2^{[\snn]}$
in the same way.

In fact, 
general access structure covers the case with several players have access to 
multiple systems as follows.
Assume that there are $\sz$ players and 
$s$-th player has access to the systems labeled by elements in the subset $A_s \subset [{\bar{\snn}}]$.
Here, we assume that $A_{s}\cap A_{s'} =\emptyset$ for $s\neq s'$
and $\cup_s A_s=[{\bar{\snn}}]$.
Any general access structure of this case is written by the pair of $\fA$ and $\fB$
to satisfy the condition that any elements of $\fA$ and $\fB$ are written as a form
$\cup_{s \in C} A_s$ with $C \subset [\sz]$.

\section{Our models} \Label{S4}

\subsection{Formulation of quantum versions of SS protocols
}\Label{S4-1}
First, we consider secret sharing, where 
the secret is a classical information and the dealer uses quantum states.
Hence, our problem setting is called classical-quantum secret sharing (CQSS).
In this problem setting, 
shares the secret is a classical information and is
given as a random variable $M \in \cM$ and $\smm \coloneqq |\cM|$.
We denote the quantum system to be sent to the $j$-th player by $\cD_j$,
and define the system $\cD[\cA]$ as $\otimes_{j \in \cA} \cD_j$
for any subset $\cA \subset [\snn]$.
Then, the dealer generates shares as quantum states, and distributes
the shares to $\snn$ players.
Finally, the end-user intends to recover the secret by collecting shares from several players.
As illustrated in Fig. \ref{fig:SS} (a), a CQSS protocol with one dealer, $\snn$ players, and 
one end-user 
is defined as Protocol \ref{Flow1}.

\begin{Protocol}[H]
\caption{CQSS protocol}         
\Label{Flow1}      
\begin{algorithmic}
\STEPONE
\textbf{Share generation}:
Depending on the message $M \in \cM$,
the dealer prepares 
$\snn$ shares as a state $\rho[M]$ on the joint system $\cD_1\otimes \cdots \otimes \cD_{\snn}$,
and sends the $j$-th share system $\cD_j$ to the $j$-th player.
\STEPTWO
\textbf{Decoding}:
For a subset $\cA \in \fA$,
the end-user decodes the message from the received state from players $\cA$
by a decoder, which is defined as a POVM
$\mathsf{Dec}(\cA) \coloneqq \{ {Y}_{\cA}(w) \mid  w\in[\smm] \}$ on 
$\cD[\cA]$.
The end-user outputs the measurement outcome $W$ as the decoded message.
\end{algorithmic}
\end{Protocol}

Then, in Protocol \ref{Flow1},
the share cost and the rate of a CQSS protocol are defined by 
\begin{align}
D \coloneqq \dim \bigotimes_{j=1}^{\snn} \cD_j , \quad
R \coloneqq \frac{\log\vM}{\log D} . 
\Label{Rate1}
\end{align}

The security of CQSS protocols are defined as follows.

\begin{defi}[$(\ACC,\REJ)$-security] \Label{def:nss}
For an access structure $(\ACC, \REJ)$ on $[\snn]$,
	a CQSS protocol defined as Protocol \ref{Flow1} 
is called $\ACC$-correct
	if the following correctness condition is satisfied.
It is called $(\ACC,\REJ)$-secure
	if 
	the following both conditions are satisfied.
\begin{itemize}[leftmargin=1.5em]
\item \textbf{Correctness}:
The relation 
    \begin{align*}
\Tr \rho[m] (Y_{\cA}(m)\otimes I_{\cA^c} )=1
    \end{align*}
holds for \red{$\cA\in \fA$} and $m \in \cM$.
\item \textbf{Secrecy}:
The state $\Tr_{\cB^c} \rho[m]$ does not depend on $m \in \cM$ 
\red{for $\cB \in \fB$}.
\end{itemize}
In particular, 
we define the $(\srr,\stt,\snn)$-security as 
the $(\ACC,\REJ)$-security 
under the choice $\ACC = \{ \cA \subset [\snn] \mid |\cA| \geq \srr \}$
and $\REJ = \{ \cB \subset [\snn] \mid |\cB| \leq \stt \}$.
This type concept of the $(\srr,\stt,\snn)$-security will be applied to various type of
$(\ACC,\REJ)$-security in the latter parts.
For example, a $(\srr,\srr-1,\snn)$-secure CQSS protocol is called threshold type CQSS protocol \cite{CGL99},
    and a $(\srr,\stt,\snn)$-secure CQSS protocol with $\srr > \stt$
is called ramp type CQSS protocol \cite{ZM15, Matsumoto18, Matsumoto20}.
\end{defi}

The classical case of ramp type SS protocols
have been actively studied in \cite{BM,IY,OKT,Stinson,Yamamoto}.
When the secret is given as a quantum state,
we need a different problem setting.
Since this problem uses quantum systems to generate shares,
our protocol is called a quantum-quantum secret sharing (QQSS) protocol.
To consider this problem, we need the quantum system $\cD_{M}$
with $\dim \cD_{M}=\smm $ to describe our secret.
As illustrated in Fig. \ref{fig:SS} (b), a QQSS protocol with one dealer, $\snn$ players, and one end-user 
is defined by Protocol \ref{Flow2-5}.
\begin{Protocol}[H]                  
\caption{QQSS protocol}         
\Label{Flow2-5}      
\begin{algorithmic}
\STEPONE
\textbf{Share generation}:
Applying a TP-CP map $\Gamma$ from $\cD_M $ to $\cD_1\otimes \cdots \otimes \cD_{\snn}$,
the dealer prepares 
$\snn$ shares as the joint system $\cD_1\otimes \cdots \otimes \cD_{\snn}$,
and sends the $j$-th share system $\cD_j$ to the $j$-th player.
\STEPTWO
\textbf{Decoding}:
For a subset $\cA \in \fA$,
the end-user decodes the message from the received state from players $\cA$
by a decoder, which is defined as a TP-CP map
$\mathcal{DEC}[\cA]$ from 
$\cD[\cA]$ to $\cD_M$.
\end{algorithmic}
\end{Protocol}

\begin{defi}[$(\ACC,\REJ)$-QQSS] \Label{def:q-nss}
For an access structure $(\ACC, \REJ)$  on $[\snn]$,
a QQSS protocol defined as Protocol \ref{Flow2-5} 
is called $(\ACC,\REJ)$-secure
	when the following condition holds.
$\ACC$-correctness is defined in the same way as Definition \eqref{def:nss}.\begin{itemize}[leftmargin=1.5em]
\item \textbf{Correctness}:
The relation 
    \begin{align*}
\mathcal{DEC}[\cA]( \Tr_{\cA^c}  \Gamma (\rho) )=\rho
    \end{align*}
holds for any state $\rho$ on $\cD_M$.
\item \textbf{Secrecy}:
The state $\Tr_{\cB^c} \Gamma(\rho)$ does not depend on the state $\rho$ on $\cD_M$.
\end{itemize}
\end{defi}

\subsection{Formulation of quantum versions of SPIR protocol
}\Label{S4-2}
We consider the following type of SPIR.
The files are given as classical information.
The query is limited to classical information.
The servers can use quantum system and their answers are quantum states.
In addition, the servers are allowed to share prior entangled states.
This problem setting is called classical-quantum SPIR (CQSPIR).
In this setting,
the files $M_1,\ldots, M_{\sff}\in [\smm]$ are uniformly and independently distributed.
Each of $\vN$ servers $\mathtt{serv}_1$, \ldots , $\mathtt{serv}_{\vN}$ contains
	a copy of all files $\vec{M} \coloneqq (M_1,\ldots, M_{\sff})^T$.
The $\vN$ servers are assumed to share an entangled state.
A user chooses a file index $K \in\{1,\ldots, \sff\}$ uniformly and 
independently of $\vec{M}$ 
	in order to retrieve the file $M_K$. 
The requirement is to construct a protocol that allows 
the user to retrieve $M_K$ from the collection of the answers from servers $\cA \in \fA$
without revealing $K$ to the collection of servers $\cB \in\fB$.
\red{The user uses a random variables 
$Q^{(K)} =  (Q_1^{(K)},\ldots,Q_{\snn}^{(K)})^T  
\in \mathcal{Q}_1\times\cdots\times \mathcal{Q}_{\snn}$ depending on $K$ as a query,
where $\mathcal{Q}_1,\ldots, \mathcal{Q}_{\vN}$ are finite sets.}
As illustrated in Fig. \ref{fig:SPIR} (a), a CQSPIR protocol is defined as Protocol \ref{Flow3}.

That is,
given the numbers of servers $\snn$ and files $\sff$,
	a CQSPIR protocol of the file size $\vM$ is described by 
$$
(\rho_{\mathrm{prev}}, \mathsf{Enc}_{\mathrm{user}}, \mathsf{Enc}_{\mathrm{serv}}, \mathsf{Dec})$$
of the shared entangled state, user encoder, server encoder, and decoder,
where $\mathsf{Enc}_{\mathrm{serv}} \coloneqq (\mathsf{Enc}_{\mathrm{serv}_{1}}, \ldots,\mathsf{Enc}_{\mathrm{serv}_{\vN}})$.

\begin{figure}
\begin{center}
  \includegraphics[width=0.8\linewidth]{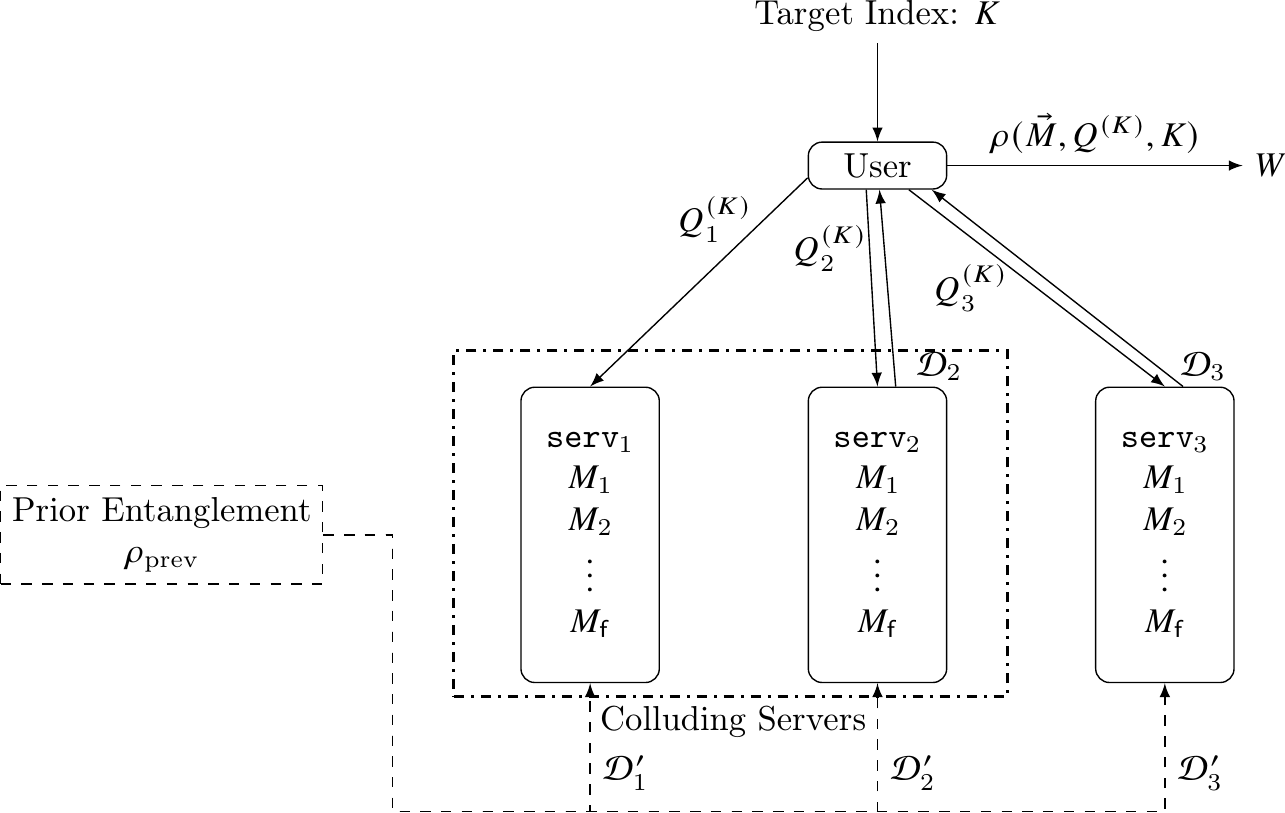}
  \end{center}
\caption{Classical-quantum (CQ) SPIR protocols where Sever 1 and Server 2 collude and Server 2 and Server 3 respond to the user.
}   \label{fig:SPIR}
\end{figure}

\begin{Protocol}[H]                  
\caption{CQSPIR protocol}         
\Label{Flow3}      
\begin{algorithmic}
\STEPONE
\textbf{Preparation}:
The state of the quantum system $\cD_1' \otimes \cdots \otimes \cD_{\snn}'$ is initialized as $\rho_{\mathrm{prev}}$ 
	and is distributed so that the $j$-th server $\mathtt{serv}_j$ contains $\cD_j'$. 
Let $U_S$ be random variable, called the {\em random seed} for servers, 
and the random seed $U_S$ is encoded as 
$\mathsf{Enc}_{\mathrm{SR}} (U_S) = R = (R_1,\ldots, R_{\snn})^T 
\in \cR = \cR_1\times \cdots \times \cR_{\snn}$ 
by the shared randomness encoder $\mathsf{Enc}_{\mathrm{SR}} $.
The randomness $R$ is distributed so that $j$-th server contains $R_j$.
\STEPTWO
\textbf{User's encoding}:
The user 
	randomly encodes the index $K$ to classical queries 
	$Q_1^{(K)},\ldots,Q_{\vN}^{(K)}$, i.e.,
\begin{align*}
\mathsf{Enc}_{\mathrm{user}} (K) = Q^{(K)} =  (Q_1^{(K)},\ldots,Q_{\snn}^{(K)})^T  
\in \mathcal{Q}_1\times\cdots\times \mathcal{Q}_{\snn},
\end{align*}
where $\mathcal{Q}_1,\ldots, \mathcal{Q}_{\vN}$ are finite sets.
Then, the user sends $Q_j$ to the $j$-th server $\mathtt{serv}_j$ ($j=1,\ldots, \snn$).
\STEPTHREE
\textbf{Servers' encoding}:
Let $\cD_1,\ldots, \cD_{\snn}$ be $\vD$-dimensional Hilbert spaces and 
$\cD[\cA]$.
After receiving the query $Q_j^{(K)}$, 
depending on the random variable $R_j$,
the $j$-th server $\mathtt{serv}_j$ constructs a TP-CP map $\Lambda_{j}$ from $\cD_j'$ to $\cD_j$ by the server encoder $\mathsf{Enc}_{\mathrm{serv}_j}$ as
\begin{align*}
\mathsf{Enc}_{\mathrm{serv}_j} (\vec{M},Q_j^{(K)}, R_j) =  \Lambda_j.
\end{align*}
Then, the $j$-th server $\mathtt{serv}_j$ applies $\Lambda_j$, and sends $\cD_j$ to the user.
The state on $\cD_1\otimes \cdots \otimes \cD_{\snn}$ is written as
\begin{align*}
\rho( \vec{M},Q^{(K)},K) 
\coloneqq \Lambda_1\otimes\cdots\otimes \Lambda_{\vN} (\rho_{\mathrm{prev}}) .
\end{align*}
\STEPFOUR
\textbf{Decoding}:
For a subset $\cA \in \fA$,
the user decodes the message from the received state from servers $\cA$
by a decoder, which is defined as a POVM
$\mathsf{Dec}(K,Q^{(K)},\cA) \coloneqq \{ {Y}_{K,Q^{(K)},\cA}(w) \mid  w\in [\smm] \}$ on 
$\cD[\cA]$
depending on the variables $K$ and $Q^{(K)}$. 
The user outputs the measurement outcome $W$ as the retrieval result.
\end{algorithmic}
\end{Protocol}

Then, in Protocol \ref{Flow3},
the upload cost, the download cost, and the rate of a CQPIR protocol $\qprot$ are defined by 
\begin{align}
U(\qprot) &\coloneqq \prod_{j=1}^{\snn} |\mathcal{Q}_j| , \quad
D(\qprot) \coloneqq \dim \bigotimes_{j=1}^{\snn} \cD_j , \\
R(\qprot) &\coloneqq \frac{\log\vM}{\log D(\qprot)} . 
\Label{Rate2}
\end{align}

The security of CQSPIR protocols are defined as follows.
\begin{defi}\Label{Def5}
For an access structure $(\ACC, \REJ)$  on $[\snn]$,
	a CQSPIR protocol defined as Protocol \ref{Flow3} 
	is called $(\ACC,\REJ)$-secure
	if 
	the following conditions are satisfied.
\begin{itemize}
\item \textbf{Correctness}:
For any $\cA \in \fA$, $k \in [\sff]$, and $\vec{m}=(m_1,\ldots, m_{\sff})^T \in [\smm]^\sff$,
the relation 
    \begin{align*}
\Tr \rho( \vec{m},q,k) (Y_{k,q,\cA}(m_k)\otimes I_{\cA^c} )=1
    \end{align*}
holds when $q$ is any possible query $Q^{(K)}$.

\item \textbf{User Secrecy}:
The distribution of $(Q_j^{(k)})_{j \in \cB}$ does not depend on $k \in [\sff]$ for any $\cB \in \fB$.

\item \textbf{Server Secrecy}:
We fix $K=k$, $M_k=m_k$, and $Q^{(K)}=q$. 
Then, the state $\rho( (m_1,\ldots, m_{\sff})^T,q,k)$ does not depend on $(m_j)_{j \neq k} \subset \cM^{\sff-1}$.
\end{itemize}
\end{defi}

\begin{table}[t]
\begin{center}
\caption{Symbols} \label{tab:symbols}
\begin{tabular}{|c|c|c|c|c|}
\hline
Symbol &	CQSPIR	&    CQSS  \\
\hline
\hline
$\snn$	&	Number of servers	&	Number of shares \\
\hline
$\sff$	&	Number of files & 	- 	\\
\hline
$\smm$	&	Size of one file 		& 	Size of secret	\\
\hline
\multirow{2}{*}{$\srr$}	&	Number of  &	Reconstruction 	\\
& responsive servers	&	threshold	\\
\hline
\multirow{2}{*}{$\stt$}	&	Number of 		&	Secrecy \\
&	colluding servers &	 threshold	\\
\hline
\end{tabular}
\end{center}
\end{table}

\section{Classical linear protocols}\Label{S5}
\subsection{Linear CSS}\Label{S5-1}
Section \ref{S4} formulates various quantum protocols.
This section reviews their classical version with the linearity condition.
As the first step, we formulate linear CSS as a special case of CQSS.

\begin{defi}[Linear CSS]
A CQSS protocol $\PNSS$ is called a {\em linear CSS} protocol with $(G, F)$ if 
the following conditions are satisfied. 
In this definition, the number of shares is written as ${\bar{\snn}}$ instead of $\snn$.
    \begin{description}
    \item[Vector representation of secret]
    The secret $M$ is written as a vector in $\FF_q^{\sx}$.
	\item[Vector representation of randomness]
	The dealer's private randomness $U_D$ is written as a uniform random vector in $\FF_q^{\sy}$.
	\item[Linearity of share generation]
	The $j$-th share is a random variable $Z_j \in \FF_q$. 
	The encoder is given as a linear map, i.e., an ${\bar{\snn}} \times (\sy+\sx) $ matrix $(G,F)$. That is,
	$(Z_1,\ldots, Z_{\bar{\snn}})^T = F M+ G U_D$.
	\end{description}

In the notation $(G,F)$, 
the first matrix $G$ identifies the direction of the randomization for secrecy,
and the second matrix $F$ identifies the direction of the message imbedding.
	Due to the above conditions, the rate of this linear CSS is $\sx/{\bar{\snn}}$.
\end{defi}

As a special case, we define MDS codes as follows.
\begin{defi}[$({\bar{\snn}},\sx)$-MDS code]
We consider the with $\sy=0$, i.e., we have only 
an ${\bar{\snn}} \times \sx$ matrix $(\emptyset,F)$. In this case,
the linear CSS protocol with $(\emptyset,F)$ is called an $({\bar{\snn}},\sx)$-maximum distance separable (MDS) code
when it is $\ACC$-correct with 
$\ACC = \{ \cA \subset [\snn] \mid |\cA| \geq \sx \}$.
\red{In addition,} when 
the linear CSS protocol with $(\emptyset,F)$ is an $({\bar{\snn}},\sx)$-MDS code,
we say that the matrix $F$
is an $({\bar{\snn}},\sx)$-MDS code.
\end{defi}

\red{In the following, we identify the matrix $F$ with
the linear CSS protocol with $(\emptyset,F)$.}

\begin{remark}
Usually, an MDS code is defined as a code whose minimum distance is ${\bar{\snn}}-\sx+1$ \cite{MDS}.
In fact, given a linear subspace $C \subset \FF_q^{{\bar{\snn}}}$,
the relation $\bar{n}- \min_{x \in C\setminus \{0\}} |x|= \sx-1$ holds
if and only if
$\dim P_{\cA}C= |\cA| $ for any $\cA \subset \red{\{0,1\}^{[{\bar{\snn}}]}}$ with $|\cA|=\sx$.
Hence, our definition for an MDS code is equivalent with the above conventional definition of an MDS code.
\end{remark}

\subsection{Multi-target monotone span program (MMSP)}\Label{S5-2}
A linear CSS is characterized by a multi-target monotone program (MMSP) \cite{Beimel11,BI93,Dijk95,SH2022}. 
Hence, to discuss its correctness and its secercy with a general access structure,
we focus on the following lemma.
To state the following lemma, we focus on the vector space $\Fq^{\sx+\sy}$, and
define the vector 
$\mathbf{e}_i\in\Fq^{\sx+\sy}$ as 
the row vector with $1$ in the $i$-th coordinate and $0$ in the others.
Also, we define the vector space $\cE$ spanned by 
$\{\mathbf{e}_{\sy+1},\ldots, \mathbf{e}_{\sy+\sx}\}$.

\begin{lemm} \label{lemm:mmsp_equiv_condition}
The following conditions are equivalent
for an ${\bar{\snn}}\times (\sy+\sx)$ matrix $(G,F)$ and a subset $\cA\subset [{\bar{\snn}}]$,
\begin{description}
\item[(A1)]
The column vectors $\pi [\im P_{\cA} G] P_{\cA} F $ are linearly independent.
\item[(A2)]
The vector space spanned by the row vectors $(P_{\cA} G,P_{\cA} F)$
contains $\cE$.
\end{description}

Also, the following conditions are equivalent
for a ${\bar{\snn}}\times (\sy+\sx)$ matrix $(G,F)$ and a subset $\cB\subset [{\bar{\snn}}]$.
\begin{description}
\item[(B1)]
Any column vector $P_{\cB} F $ is included in the linear span of 
column vectors $P_{\cB} G $.
\item[(B2)]
The vector space spanned by the row vectors $(P_{\cB} G,P_{\cB} F)$ does not contain
any non-zero element of $\cE$.
\end{description}
\end{lemm}

\begin{proof}
The condition (A1) is equivalent to the  following condition (A3).
\begin{description}
\item[(A3)]
There exists an $|\cA|\times |\cA| $ invertible matrix  $H$ such that
$H(P_{\cA} G,P_{\cA} F)=
\left(
\begin{array}{cc}
0 & I_{\sx} \\
*& *
\end{array}
\right)$, where $*$ means an arbitrary form.
\end{description}
Also, the condition (A2) is equivalent to the  following condition (A3).
Hence, we obtain the equivalence between (A1) and (A2).

The condition (B1) does not hold if and only if the following condition (B3*) holds.
\begin{description}
\item[(B3*)]
There exist an $|\cB|\times |\cB| $ invertible matrix  $H$ 
and a column vector $a$
such that
$H(P_{\cB} G,P_{\cB} F)=
\left(
\begin{array}{cc}
0 & a \\
*& *
\end{array}
\right)$.
\end{description}
Also, the condition (B2) does not hold if and only if the following condition (B3*) holds.
Hence, we obtain the equivalence between (B1) and (B2).
\end{proof}

Then, MMSP is defined as follows.

\begin{defi}[{Multi-target monotone span program (MMSP)}] \Label{defi:MMSP}
Given an ${\bar{\snn}}\times (\sy+\sx)$ matrix $(G,F)$, we say the following;
\begin{itemize}
\item (Acceptance) $(G,F)$ accepts $\ACC$ if the condition (A1) or (A2) holds
 for any $\cA\in\fA$.
 \item (Rejection) $(G,F)$ rejects $\REJ$ if the condition (B1) or (B2) holds
for any $\cB\in\fB$.
\end{itemize}
Then, the matrix $(G,F)$ is called $(\ACC,\REJ)$-MMSP if $(G,F)$ accepts $\ACC$ and rejects $\REJ$.
\end{defi}

A MMSP characterizes the security of CSS as follows.
\begin{prop}[\protect{\cite[Corollary 5]{SH2022}}] \Label{prop1}
A linear CSS protocol with $(G,F)$ is 
$(\ACC,\REJ)$-secure if and only if
$(G,F)$ is $(\ACC,\REJ)$-MMSP. 
\end{prop} 

\red{Therefore, 
a linear CSS protocol completely characterized by an MMSP.
That is, to consider the security of a given linear CSS protocol with $(G,F)$,
it is sufficient to consider its corresponding MMSP defined by $(G,F)$.}

\begin{proof}
For a linear CSS protocol with $(G,F)$,
the decodable information is given as $ (\im P_{\cA}F+ \im P_{\cA}G)/ \im P_{\cA}G$.
Thus, if and only if the map $x \in \FF_q^{\sx} \mapsto \pi[\im P_{\cA}G] (P_{\cA} F x)
\in (\im P_{\cA}F+ \im P_{\cA}G)/ \im P_{\cA}G$ is injetive, the correctness holds.
Also, if and only if \red{$P_{\cB} F x  \in \im P_{\cB}G$} for $\in \FF_q^{\sx}$, 
the secrecy holds.
Therefore,
the acceptance condition for MMSP guarantees correctness,
and 
the rejection condition for MMSP guarantees secrecy.
Hence, we have the following proposition.
\end{proof}

\begin{remark}
Our definition of MMSP is the same as the definition in \cite{SH2022}
while the paper \cite{SH2022} uses the conditions (A2) and (B2).
Hence, Proposition \ref{prop1} was shown in \cite{SH2022} by using the conditions (A2) and (B2).
This paper mainly uses the conditions (A1) and (B1) while 
other preceding studies also use the conditions (A2) and (B2) as well as the reference \cite{SH2022}.
The definition in \cite{SH2022} generalized the definition in \cite{Beimel11,BI93,Dijk95}.
The MMSP defined in \cite{Beimel11,BI93,Dijk95} corresponds to the definition in \cite{SH2022}
with 
    $\ACC \cup \REJ = \red{\{0,1\}^{[{\bar{\snn}}]}}$ and $\ACC \cap \REJ = \emptyset$, i.e., 
	every subset of $[{\bar{\snn}}]$ is either authorized or forbidden.
Our definition of MMSP also generalizes the monotone span programs \cite{KW93},
    which corresponds to the case $\sx = 1$ and $\ACC\cup\REJ = \red{\{0,1\}^{[{\bar{\snn}}]}}$ for our MMSP definition.
The papers \cite{Brickell89, KW93,Beimel_thesis} proved 
	the equivalence of linear CSS protocols with complete security and monotone span programs.
\end{remark}

As special cases, we define $(\ACC,\REJ)$-MMSPs with thresholds as follows \cite{SH2022}.

\begin{defi}[$(\srr,\stt,{\bar{\snn}})$-MMSP]\Label{LXP}
When $\ACC = \{ \cA \subset [{\bar{\snn}}] \mid |\cA| \geq \srr \}$
	and $\REJ = \{ \cB \subset [{\bar{\snn}}] \mid |\cB| \leq \stt \}$,
	an $(\ACC,\REJ)$-MMSP are called an $(\srr,\stt,{\bar{\snn}})$-MMSP.
\end{defi}

An $(\srr,\stt,{\bar{\snn}})$-MMSP is related with an MDS code.
A matrix $A\in\FF_q^{{\bar{\snn}}\times\sff}$ is an $({\bar{\snn}},\sff)$-MDS code
	if and only if any $\sff$ rows of $A$ are linearly independent because 
	the linear independence guarantees the correctness.
Then, we have the following proposition \cite{SH2022}.

\begin{prop} \Label{theo:MDS-MMSP}
An ${\bar{\snn}}\times (\stt+ (\srr-\stt))$ matrix
$(G,F)$ is an $(\srr,\stt,{\bar{\snn}})$-MMSP if and only if
the matrix $(G,F)$ is an $({\bar{\snn}},\srr)$-MDS code, and the matrix $G$ is 
an $({\bar{\snn}},\stt)$-MDS code.
\end{prop}

\subsection{Linear CSPIR}\Label{S5-3}
Next, we formulate linear CSPIR as a special case of CQSPIR as follows.

\begin{defi}[{Linear CSPIR}] \Label{defi:lin-spir}
A protocol 
is called a {\em linear CSPIR} protocol if 
    the following conditions are satisfied. 
In this definition, the number of servers is written as ${\bar{\snn}}$ instead of $\snn$.
\begin{description}
    
    \item[Vector representation of files]
    The files $M_i$ are written as a vector in $\Fq^{\sx}$. 
    The entire file is written by the concatenated vector 
    \red{$\vec{M} = (M_1, \ldots, M_\sff)^T \in \Fq^{\sff \sx }$}.
    \item[Linearity of shared randomness]
    The random seed $U_S$ is written by a uniform random vector in $\Fq^{\sy}$.
    The randomness encoder is written as a matrix 
    $G\in\Fq^{{\bar{\snn}}\times \sy}$ and 
    the shared randomness is written as $R = G U_S \in \Fq^{{\bar{\snn}}}$.
    The randomness of the $j$-th server is written as $R_j = G_j U_S  \in \Fq$, where $G=(G^{(1)}, \ldots, G_{{\bar{\snn}}})^T$, i.e., $G_j$ is a row vector of $G$.  
        \item[Linearity of servers]
The $j$-th server's system $\cD_j$ is classical system, i.e., is given as a random variable $D_j \in \Fq$.    
    The answer of the $j$-th server $D_j$ is written as 
    the sum of 
        the shared randomness $R_j$
        and
        the encoded output of the files $\vec{M}$ by \red{an $\sff \sx $-dimensional random column vector} $ Q_{j}^{(K)} $, which depends on the query, i.e., 
        \begin{align}
        D_j = Q_{j}^{(K)} \vec{M}  + R_j \in \red{\Fq}.
        \Label{eq:linearencc}
        \end{align}
Therefore, we can consider that the query to the $j$-th server is given as
the linear function, a \red{random} matrix, 
\red{$ Q^{(K)} =(Q_{j}^{(K)})_{j=1}^{\bar{\snn}} \in \Fq^{{\bar{\snn}}\times \sff\sx }$.}
\end{description}

The above protocol is called 
the linear CSPIR protocol with $G, Q^{(K)}$.

Due to the above conditions, the PIR rate and the shared randomness rate of a linear SPIR protocol are $\sx/\sz$ and $\sy/\sx$, respectively.
\end{defi}

In the case of linear CSPIR protocols, the server secrecy can be characterized as follows.

\begin{lemm}\Label{L7}
\red{Assume that $ Q^{(k)} =(Q_{j}^{(k)})_{j=1}^{\bar{\snn}}$ is an ${\bar{\snn}} \times \sff \sx $ random matrix for $k=1, \ldots, \sff$.}
A linear CSPIR protocol with $G,Q^{(K)}$ satisfies 
the server secrecy if and only if
any column vector of
$\red{Q^{(k)}_j }\in \FF_q^{{\bar{\snn}} \times \sx}$ belongs to the linear span of column vectors of $G$ for $j\neq k$.
\end{lemm}

\begin{proof}
The user cannot distinguish elements in $\im G$.
Hence, the server secrecy holds if and only if the following holds.
Let $k$ be an arbitrary element in $[\sff]$ and $m_k$.
The element $\pi[\im G] Q^{(k)} \vec{m}$ does not depend on $(m_j)_{j\neq k}$.
Since the above condition is equivalent to the condition stated in this lemma. 
Hence, the desired statement is obtained.
\end{proof}

For the choice of the query $Q^{(k)}$, we consider the following construction
in a similar way to
\cite{SJ17-2, WS17-2, WS18, TGKFHE17, TGKFH19, SH2022}.

\begin{defi}[Standard form] \Label{def:projection1}
Let $F$ be an ${\bar{\snn}} \times \sx$ matrix taking values in $\FF_q$. 
A query $Q^{(K)}$ is called the standard form with matrix $F$
when it is given as follows.
The user prepares uniform random variable $U_{Q} \in \FF_q^{\sy \times \sx \sff }$ independently of $K$. Then,
\begin{align}
Q^{(k)}:= F E_k +G U_Q,\Label{CAP1}
\end{align}
where $E_k:=(\delta_{1,k} I_{\sx}, \ldots, \delta_{\sff,k} I_{\sx}) $.
A linear CSPIR protocol with $G,Q^{(K)}$ is called
a standard linear CSPIR protocol with $(G, F)$ when the query $Q^{(K)}$ is the standard form with matrix $F$.
\end{defi}


The security of a standard linear CSPIR protocol is characterized as follows.

\begin{prop}\Label{P2}
The standard linear CSPIR protocol with $(G, F)$ is 
$(\ACC,\REJ)$-secure
if and only if
the matrix $(G,F)$ is $(\ACC,\REJ)$-MMSP.
\end{prop}

Although Proposition \ref{P2} was shown shown in \cite{SH2022} by using the conditions (A2) and (B2), we show it by using the conditions (A1) and (B1)
because we use the conditions (A1) and (B1) in the latter discussion.

\begin{proof}
The choice \eqref{CAP1} of $Q^{(k)}$ satisfies the condition of Lemma \ref{L7}.
Hence, it is sufficient to discuss the user secrecy and the correctness.
The user secrecy holds for $\cB \in \fB$ if and only if
$P_{\cB} (F+ G U)$ and $P_{\cB} GU$ cannot be distinguished
when $U \in \FF_q^{\sy \times \sx}$ is subject to the uniform distribution.
This condition is equivalent to the rejection condition for $\fB$ and $(G,F)$.
The correctness holds for $\cA \in \fA$ if and only if
the map $m \in \FF_q^{\sx} \mapsto \pi[\im P_{\cA}G] (P_{\cA} F x)
\in  (\im P_{\cA}F+ \im P_{\cA}G)/ \im P_{\cA}G$ is injetive.
This condition is equivalent to the acceptance condition for $\fA$ and $(G,F)$.
Therefore, the desired statement is obtained.
\end{proof}

\section{Linear quantum protocols without preshared entanglement with user}\Label{SS-5}
\subsection{Linear CQSS protocol}\Label{S6-1-1}
We assume that 
an ${\bar{\snn}}\times \snn$ matrix $G^{(1)}$,
an ${\bar{\snn}}\times \sy_2$ matrix $G^{(2)}$, and
an ${\bar{\snn}}\times \sx$ matrix $F$
on the finite field $\FF_q$ 
with ${\bar{\snn}}=2 \snn$ satisfy the following conditions.
All column vectors of $(G^{(1)},G^{(2)},F)$ are linearly independent, and
all column vectors 
of $G^{(1)}$ are commutative with each other, which are equivalent to the self-column-orthogonal  condition.
Then, we define
a CQSS protocol as follows.
We choose the message set $\cM$ as $\FF_q^{\sx}$.
We choose the Hilbert space ${\cal H}$ as the space spanned by 
$\{ |x\rangle \}_{x \in \FF_q}$.
We define a normalized vector $|\psi[G^{(1)}]\rangle \in {\cal D}_D:= {\cal H}^{\otimes n}$ as
the common eigenvector with eigenvalue $1$ of 
$\mathbf{W}_{[\snn]}(g^1), \ldots, \mathbf{W}_{[\snn]}(g^\snn) $, i.e., 
\begin{align}
\mathbf{W}_{[\snn]}(G^{(1)} y )|\psi[G^{(1)}]\rangle=|\psi[G^{(1)}]\rangle
\hbox{ for } y \in \FF_q^{\snn}.
\end{align}
 Then, we define the linear CQSS protocol with $(G^{(1)},G^{(2)},F)$ 
 as Protocol \ref{protocol1CQ}.
\begin{Protocol}[H]                  
\caption{Linear CQSS protocol with $(G^{(1)},G^{(2)},F)$}         
\Label{protocol1CQ}
\begin{algorithmic}
\STEPONE
\textbf{Preparation}:
We set the initial state $\rho_{D}$ on ${\cal D}_D$
to be $|\psi[G^{(1)}]\rangle$.
\STEPTWO
\textbf{Share generation}:
The dealer prepares a uniform random variable $U_D \in \FF_q^{\sy_2} $.
For $m \in \cM$, the dealer applies 
$\mathbf{W}_{[\snn]}( Fm+G^{(2)}U_D)$ on ${\cal D}_D$. That is,
the encoding operation $\Gamma[m]$ \red{on ${\cal D}_D$} is defined as
\begin{align}
\Gamma[m](\rho):= \sum_{u_D \in \FF_q^{\sy}} 
\frac{1}{q^{\sy_2}}
\mathbf{W}_{[\snn]}( Fm+G^{(2)}u_D) \rho 
\mathbf{W}_{[\snn]}^\dagger( Fm+G^{(2)} u_D).
\end{align}
The shares are given as parts of the state $ \Gamma[m]
(|\psi[G^{(1)}]\rangle \langle \psi[G^{(1)}]|)$.
\STEPTHREE
\textbf{Decoding}:
For a subset $\cA \in \fA$,
the end-user makes measurement on the basis
$\{\mathbf{W}_{\cA}(y) \Tr_{\cA^c}|\psi[G^{(1)}]\rangle \langle \psi[G^{(1)}]|
\mathbf{W}_{\cA}^\dagger(y)
\}_{y \in \FF_q^{2|\cA|}}$.
Based on the obtained outcome, the end-user recovers $m$.
\end{algorithmic}
\end{Protocol}

A usual linear CQSS protocol  does not have randomization 
$U_D \in \FF_q^{\sy_2} $ \cite{HBB99, KKI99, Gottesman00, MS08, KMMP09, KFMS10, Sarvepalli12, MM13,WCY14, Matsumoto17, Matsumoto20}.
That is, $\sy_2=0$ and it does not have the matrix $G^{(2)}$.
Such a protocol is called the randomless linear CQSS protocol with $(G^{(1)},F)$.
In this case, the state with message $M=0$ is determined as
the stabilizer of the group generated by $G^{(1)}$.
That is, any CQSS protocol given as the application of $\mathbf{W}_{[\snn]}$
to the stabilizer state is written as the above way.
Then, we have the following theorem.

\begin{table}[t]
\caption{Comparison for analysis for CQSS protocols}
\Label{hikaku1}
\begin{center}
\begin{tabular}{|c|c|c|c|c|}
\hline
 & general & relation   &  \multirow{2}{*}{ramp}   \\
 & access &to  &    \multirow{2}{*}{scheme}   \\
 & structure &MMSP &       \\
\hline
\cite{WCY14,Sarvepalli12} &Yes& No & No      \\
\hline
\cite{Smith00} &Yes& special cases & No   \\
\hline
\cite{Matsumoto17,Matsumoto20} &No& No & special cases    \\
\hline
\multirow{2}{*}{This paper} & \multirow{2}{*}{Yes} & general case  & general case   \\
 &  & (Theorem \ref{Cor1})  & (Corollary \ref{Coro6CQ}) \\
\hline
\end{tabular}
\end{center}
\end{table}

\begin{theo}\Label{Cor1}
Given a $2\snn \times \snn$ self-column-orthogonal matrix $G^{(1)}$,
a $2\snn\times \sy_2$ matrix $G^{(2)}$,
and
a $2\snn \times \sx$ matrix $F$, 
the following conditions for $G^{(1)},G^{(2)},F$ are equivalent.
\begin{description}
\item[(C1)]
The linear CQSS protocol with $(G^{(1)},G^{(2)},F)$
is $(\ACC,\REJ)$-secure.
\item[(C2)]
The linear CSS protocol with $((G^{(1)},G^{(2)}),F)$ 
is $(\bar{\ACC},\bar{\REJ})$-secure.
\item[(C3)]
The matrix $((G^{(1)},G^{(2)}),F)$ is an $(\bar{\ACC},\bar{\REJ})$-MMSP.
\end{description}
\end{theo}

Also, we have the following proposition.
\begin{prop}\Label{ZNOCQ}
In Protocol \ref{protocol1CQ},
even when STEP 3 is replaced by another decoder,
the decoder can be simulated by the decoder given in STEP 3.
That is, once STEPs 1 and 2 are given in Protocol \ref{protocol1CQ},
without loss of generality, we can assume that our decoder is given as STEP 3.
\end{prop}

We will prove the above theorem and proposition
after we introduce linear EASS protocols in the next section.
The papers \cite{Smith00,WCY14,Sarvepalli12} studied CQSS protocols 
with general access structure.
However, the paper \cite{WCY14,Sarvepalli12} did not discuss its relation with MMSP,
and discussed only the case $\ACC = \REJ^c$,
and the paper \cite{Smith00} considered it only 
when $\sx=1$ and the state $|\psi[G^{(1)}]\rangle$ is restricted to the following form
with an $\snn \times \sy'$ matrix $F'$
\begin{align}
\sum_{a \in \bF_q^{\sy'-1}}
\left|F'
\left(
\begin{array}{c}
0 \\
a
\end{array}
\right)
\right\rangle.
\end{align}
That is, our result, Theorem \ref{Cor1} covers 
the relation with MMSP and general access structure
under a large class of CQSS protocols.
The comparison with the existing results is summarized in Table \ref{hikaku1}.

To characterize CQSS protocols with two thresholds, i.e., with the ramp scheme,
we define the following special case of $(\ACC,\REJ)$-MMSPs.

\begin{defi}[$(\srr,\stt,\snn)$-CQMMSP]
We choose $\ACC = \{ \cA \subset [\snn] \mid |\cA| \geq \srr \}$
and $\REJ = \{ \cB \subset [\snn] \mid |\cB| \leq \stt \}$.
Given a $2\snn \times \snn$ self-column-orthogonal matrix $G^{(1)}$,
a $2\snn\times \sy_2$ matrix $G^{(2)}$,
and a $2\snn\times \sx$ matrix $F$, 
the matrix $(G^{(1)},G^{(2)},F)$ is called an $(\srr,\stt,\snn)$-CQMMSP
when the matrix $((G^{(1)},G^{(2)}),F)$ is an $(\bar{\ACC},\bar{\REJ})$-MMSP.
\end{defi}

\begin{theo}\Label{TH3CQ}
When $\snn \ge \srr > \stt >0$,
and $\srr > \snn/2 $,
there exist a positive integer $s$,
a $2\snn \times \snn$ self-column-orthogonal matrix $G^{(1)}$,
a $2\snn\times [2\stt-\snn]_+$ matrix $G^{(2)}$,
and a $2\snn\times (2\srr-\max(2\stt,\snn))$ matrix $F$ on $\FF_{q}$ with $q=p^s$
such that 
the matrix $(G^{(1)},G^{(2)},F)$ is an $(\srr,\stt,\snn)$-CQMMSP.
Here, we use the notation $[x]_+:= \max(0,x)$.
\end{theo}

Theorem \ref{TH3CQ} is shown in Appendix \ref{A2B}.
Combining Theorems \ref{Cor1} and \ref{TH3CQ}, we obtain the following corollary.

\begin{coro}\Label{Coro6CQ}
When $\snn \ge \srr > \stt >0 $  and $\srr \ge \snn/2>0$,
there exists an $(\srr,\stt,\snn)$-secure CQSS protocol of rate 
$(2\srr-\max(2\stt,\snn))/\snn$.
In particular,
when $\snn \ge \srr > \snn/2\ge \stt >0$,
there exists an $(\srr,\stt,\snn)$-secure randomless
CQSS protocol of rate $(2\srr-\snn)/\snn$.
\end{coro}

In the classical case with the ramp scheme, 
the optimal rate of $(\srr,\stt,\snn)$-secure SS protocol
is $(\srr-\stt)/\snn$ 
\cite{Ogata,Okada,Paillier}.
Hence, the rate of Corollary \ref{Coro6CQ} is twice of the classical case.
In fact, the existence of the above rate of the ramp case
was not shown in the general case.
Only a limited case of the ramp case was discussed in \cite{Matsumoto17,Matsumoto20}.
That is, the achievability of the rate $(2\srr-\max(2\stt,\snn))/\snn$.
was not shown in existing studies.

\subsection{Linear QQSS}\Label{S8}
Next, we discuss linear QQSS protocols.
For this aim, we focus on
a $2\snn\times (\snn-\sx)$ self-column-orthogonal matrix $G^{(1)}$
and a $2\snn\times \sy_2$ matrix $G^{(2)}$.
Also, we focus on a $2\snn\times 2 \sx$ matrix $F$ 
that is column-orthogonal to $G^{(1)}$.
In this construction, 
we use the notation given in Section \ref{S6-1-1}, and use the matrix $G^{(1)}$
as the stabilizer.
That is, we define the subset $\cD_D[y,G^{(1)}] \subset \cD_D$ in the same way as $\cD_E[y,G^{(1)}]$ for $y \in \FF_q^{\snn-\sx}$.
Then, we define a QQSS protocol as Protocol \ref{protocol6}.
This protocol is called the linear QQSS protocol with $(G^{(1)},G^{(2)},F)$.

\begin{Protocol}[H]                  
\caption{Linear QQSS protocol with $(G^{(1)},G^{(2)},F)$}         
\Label{protocol6}
\begin{algorithmic}
\STEPONE
\textbf{Share generation}:
The dealer encodes the system $\cH^{\otimes |\sx|}$
into the subspace $\cD_{D}[0,G^{(1)}]$. Hence, the message system $\cD_M$
is identified with $\cD_{D}[0,G^{(1)}]$.
The dealer prepares a uniform random variable $U_{D,2} \in \FF_q^{\sy_2} $,
and applies 
$\mathbf{W}_{[\snn]}(G^{(2)} U_{D,2})$ on ${\cal D}_D$. 
\STEPTWO
\textbf{Decoding}:
For a subset $\cA \in \fA$,
the end-user applies a suitable TP-CP map $\overline{\Gamma}$
to recover the original state.
\end{algorithmic}
\end{Protocol}

\begin{theo}\Label{TH4QQ}
Given a $2\snn\times (\snn-\sx)$ self-column-orthogonal matrix $G^{(1)}$
and
a $2\snn\times \sy_2$ matrix $G^{(2)}$,
we choose a $2\snn\times 2 \sx$ matrix $F$ column-orthogonal to $G^{(1)}$.
Then, the following conditions for $G^{(1)},G^{(2)},F$
are equivalent.
\begin{description}
\item[(D1)]
The linear QQSS protocol with $(G^{(1)},G^{(2)},F)$
is $(\ACC,\REJ)$-secure.
That is, there exists a suitable TP-CP map $\overline{\Gamma}$ to recover the original state
in STEP 2.
\item[(D2)]
The linear CSS protocol with $((G^{(1)},G^{(2)}),F)$
is $(\bar{\ACC},\bar{\REJ})$-secure.
\item[(D3)]
The matrix $((G^{(1)},G^{(2)}),F)$ is an $(\bar{\ACC},\bar{\REJ})$-MMSP.
\end{description}
\end{theo}

This theorem will be shown by Theorem \ref{TH4} including the construction of the decoder in Section \ref{S8-2}.
The papers \cite{Sarvepalli12} studied QQSS protocols 
with general access structure.
However, the paper \cite{Sarvepalli12} did not discuss its relation with MMSP.

To characterize QQSS protocols with the threshold case, i.e., the ramp scheme,
we define the following special case of $(\ACC,\REJ)$-MMSPs.

\begin{defi}[$(\srr,\stt,\snn)$-QQMMSP]
We choose $\ACC = \{ \cA \subset [\snn] \mid |\cA| \geq \srr \}$
and $\REJ = \{ \cB \subset [\snn] \mid |\cB| \leq \stt \}$.
Given a $2\snn \times (\snn-\sx)$ self-column-orthogonal matrix $G^{(1)}$,
a $2\snn\times \sy_2$ matrix $G^{(2)}$,
and a $2\snn\times 2\sx$ matrix $F$ 
column-orthogonal to $G^{(1)}$,
the matrix $(G^{(1)},G^{(2)},F)$ is called an $(\srr,\stt,\snn)$-QQMMSP
when the matrix $((G^{(1)},G^{(2)}),F)$ is an $(\bar{\ACC},\bar{\REJ})$-MMSP.
\end{defi}

When the matrix $(G^{(1)},G^{(2)},F)$ is called an $(\srr,\stt,\snn)$-QQMMSP,
the linear QQSS protocol with $(G^{(1)},G^{(2)},F)$
has the perfect secrecy for $\stt$ colluded players.
In this case, the secret can be recovered by the end user when 
the end user collect the shares from $\srr$ players.
Such a linear QQSS protocol is called 
an $(\srr,\stt,\snn)$-linear QQSS protocol.

\begin{theo}\Label{TH6}
Let $p$ be a prime number.
Then, the following conditions for positive integers $\srr,\stt,\snn$ with $\snn \ge \srr > \stt >0$ 
are equivalent.
\begin{description}
\item[(E1)]
The condition  $\srr \ge (\snn+1)/2$ holds.
\item[(E2)]
There exist positive integers $s$, $\sy_2$, and $\sx$,
a $2\snn \times (\snn-\sx)$ self-column-orthogonal matrix $G^{(1)}$,
a $2\snn\times \sy_2$ matrix $G^{(2)}$,
and a $2\snn\times 2 \sx$ matrix $F$ column-orthogonal 
to the matrix $G^{(1)}$
on $\FF_{q}$ with $q=p^s$
such that 
the matrix $(G^{(1)},G^{(2)},F)$ is an $(\srr,\stt,\snn)$-QQMMSP.
\item[(E3)]
Choose $\stt':= \max(\stt, \snn-\srr)$.
There exist a positive integer $s$, 
a $2\snn \times (\snn-\srr+\stt')$ self-column-orthogonal matrix $G^{(1)}$,
a $2\snn\times (\stt'+\srr-\snn)$ matrix $G^{(2)}$,
and a $2\snn\times 2 (\srr-\stt')$ matrix $F$ column-orthogonal 
to the matrix $G^{(1)}$
on $\FF_{q}$ with $q=p^s$
such that 
the matrix $(G^{(1)},G^{(2)},F)$ is an $(\srr,\stt,\snn)$-QQMMSP.
\end{description}
\end{theo}
Theorem \ref{TH6} is shown in Appendix \ref{A3}.
Therefore, the threshold scheme, i.e., a $(\srr,\srr-1,\snn)$-secure QQSS protocol
exists when $\snn \ge \srr  \ge (\snn+1)/2$ \cite{CGL99}.
Combining Theorem \ref{TH4QQ} and Theorem \ref{TH6}, we obtain the following corollary.

\begin{coro}\Label{Coro9}
When the condition  $\srr \ge (\snn+1)/2$ holds,
there exists an $(\srr,\stt,\snn)$-secure QQSS protocol with rate 
$(\srr-\max(\stt, \snn-\srr))/\snn$.
\end{coro}

Only a limited case of the ramp case for QQSS protocols
was discussed in existing studies \cite{ZM15,Matsumoto18}.
That is, the achievability of the rate $(\srr-\max(\stt, \snn-\srr))/\snn$
was not shown in existing studies.
In addition, as a special case of QQSS protocol, we define the QQ version of MDS codes as follows.

\begin{defi}[$(\snn, \srr)$-QQMDS code]\Label{AXP}
We consider the case with $\stt= \snn-\srr$.
Assume that $ G^{(1)}$ is a $2\snn\times 2(\snn-\srr)$ self-column-orthogonal matrix and 
$F$ is a $2\snn\times 2(2\srr-\snn)$ matrix column-orthogonal to $G^{(1)}$.
We say that the randomless linear QQSS protocol with $(G^{(1)}, F)$ 
is an $(\snn, \srr)$-QQMDS code
when it is $\fA$-correct with $\fA= \{ \cA \subset [\snn] \mid |\cA| \geq \srr \}$.
\end{defi}

Usually, a QQMDS code is called a quantum MDS code \cite{KL,Rains,HG}.
Hence, any $(\srr,\snn-\srr,\snn)$-secure QQSS protocol
is an $(\snn,\srr)$-QQMDS code while 
its special case was mentioned in \cite{RST05}.
We have the following characterization for QQMDS codes.

\begin{lemm}\Label{Cor73QQ}
Given a $2\snn \times 2(\snn-\srr)$ self-column-orthogonal matrix $G^{(1)}$
and $2\snn\times 2(2\srr-\snn)$ matrix $F$ column-orthogonal to $G^{(1)}$,
the following conditions are equivalent.
\begin{description}
\item[(F1)]
The randomless linear QQSS protocol with $(G^{(1)},F)$ 
is a $(\snn, \srr)$-QQMDS code.
\item[(F2)]
The matrix $(G^{(1)},F)$ accepts
$\bar{\ACC}$ with $\ACC= \{ \cA \subset [\snn] \mid |\cA| \geq \srr \}$
in the sense of Definition \ref{defi:MMSP}.
\end{description}
\end{lemm}

This Lemma will be shown as Corollary \ref{Cor73} in Section \ref{S8-2}.
Considering the special case of Theorem \ref{TH6} with $\stt=\snn-\srr$,
we have the following corollary because 
an $(\srr,\snn-\srr,\snn)$-QQMMSP $(G^{(1)},\emptyset,F)$ accepts
$\bar{\ACC}$ with $\ACC= \{ \cA \subset [\snn] \mid |\cA| \geq \srr \}$.

\begin{coro}\Label{Cor74}
For any positive integers $\srr,\snn$ with $\snn \ge \srr > 0$
and any prime $p$,
there exist a positive integer $s$,
a $2\snn \times 2(\snn-\srr)$ self-column-orthogonal matrix $G^{(1)}$,
and a $2\snn\times 2(2\srr-\snn)$ matrix $F$ column-orthogonal 
to the matrix $G^{(1)}$
on $\FF_{q}$ with $q=p^s$
such that 
the matrix $(G^{(1)},F)$ accepts
$\bar{\ACC}$ with $\ACC= \{ \cA \subset [\snn] \mid |\cA| \geq \srr \}$.
\end{coro}

Therefore, due to Lemma \ref{Cor73QQ} and Corollary \ref{Cor74},
there exists 
an $(\srr,\snn-\srr,\snn)$-secure QQSS protocol
that gives an $(\snn,\srr)$-QQMDS code
with a sufficiently large prime power $q$ of any prime number $p$.

\begin{remark}
Usually, a quantum MDS code is defined as a code minimum distance $d$ 
whose dimension is $\snn +2-2d$ \cite{KL,Rains,HG}.
The reference \cite{GBP} showed that the following conditions for stabilizer codes
are equivalent to 
the condition that 
the minimum distance is $d$. 
\begin{description}
\item[(G1)]
It is possible to detect the existence of error 
only with $d-1$ systems.
\item[(G2)]
It is possible to recover the original state when 
errors occur only in $\lfloor (d-1)/2\rfloor$ systems.
\item[(G3)]
It is possible to recover the original state even when 
$\lfloor d-1\rfloor$ systems are lost at most. 
\end{description}
Our definition (Definition \ref{AXP})
is based on the condition (G3).
\end{remark}


\begin{remark}\Label{R5}
We remark the relation with stabilizer codes.
When $G^{(2)}$ is empty, 
the randomless linear QQSS protocol with $(G^{(1)},F)$
with a $2\snn\times (\snn-\sx)$ matrix $G^{(1)}$ and a $2\snn\times \sx$ matrix $F$
has a relation with a stabilizer code.
Since $G^{(1)}$ is self-column-orthogonal,
the group
$N:=\{ G^{(1)} y \}_{y \in \FF_q^{(\snn-\sx)}}$
satisfies the self-orthogonal condition 
$N\subset N^{\perp}:=
\{  v\in \FF_q^{2\snn}| (v, v')=0 ,\quad \forall v' \in N  \}$.
Since $F$ is column-orthogonal to $G^{(1)}$,
$N^\perp=\{ G^{(1)} y+F x \}_{y \in \FF_q^{\sy_1}, x \in \FF_q^{2\sx}}$.
Therefore, 
the randomless linear QQSS protocol with $(G^{(1)},F)$ can be considered as a 
stabilizer code with the stabilizer $N$.

Conversely, we consider a stabilizer code with 
a stabilizer $N \subset \FF_q^{2\snn}$ of dimension $\snn-\sx $.
Then, depending on $N$, we choose a $2\snn\times (\snn-\sx)$ matrix $G^{(1)}$ and a $2\snn\times \sx$ matrix $F$ to satisfy the conditions.
\begin{align}
N=\{ G^{(1)} y \}_{y \in \FF_q^{(\snn-\sx)}},\quad 
N^\perp=\{ G^{(1)} y+F x \}_{y \in \FF_q^{\sy_1}, x \in \FF_q^{2\sx}}.
\end{align}
\end{remark}

\begin{table}[t]
\caption{Comparison for analysis for QQSS protocols}
\Label{hikaku3}
\begin{center}
\begin{tabular}{|c|c|c|c|c|}
\hline
 & general & \multirow{2}{*}{relation}   &  \multirow{2}{*}{ramp} & relation between    \\
 & access &\multirow{2}{*}{to MMSP}  &   \multirow{2}{*}{scheme}  &  QQMDS code and \\
 & structure & &     & QQSS  protocol \\
\hline
\cite{ZM15,Matsumoto18} &No& No & special case    & No \\
\hline
\cite{KL,Rains} &\multirow{2}{*}{No} & \multirow{2}{*}{No} & \multirow{2}{*}{No} &No (They studied   \\
\cite{HG,GBP}&&&& only QQMDS code.)\\
\hline
\multirow{2}{*}{This paper} & \multirow{2}{*}{Yes} & general case  & general case  & Yes \\
& &(Theorem \ref{TH4QQ})  & (Corollary \ref{Cor74}) & (Lemma \ref{Cor73QQ})\\
\hline
\end{tabular}
\end{center}
\end{table}

\subsection{Linear CQSPIR protocol}
We choose
an ${\bar{\snn}}\times \snn$ matrix $G^{(1)}$,
an ${\bar{\snn}}\times \sy_2$ matrix $G^{(2)}$, and
an ${\bar{\snn}}\times \sx$ matrix $F$
on the finite field $\FF_q$ with ${\bar{\snn}}=2 \snn$
 in the same way as 
Section \ref{S6-1-1}.
Then, Similar to Section \ref{S6-1-1},
we define 
a normalized vector $|\psi[G^{(1)}]\rangle \in
\cD_1 \otimes \cdots \otimes \cD_{\snn}$ 
as the stabilizer of the column vectors of $G^{(1)}$.
Then, using the state  $|\psi[G^{(1)}]\rangle$,
we define the linear CQSPIR protocol with $G^{(1)},G^{(2)},Q^{(K)}$
 as Protocol \ref{protocol5CQ}.

\begin{Protocol}[H]                  
\caption{Linear CQSPIR protocol with $G^{(1)},G^{(2)},Q^{(K)}$}        
\Label{protocol5CQ}      
\begin{algorithmic}
\STEPONE
\textbf{Preparation}:
We set the initial state $\rho_{D}$ on 
$\cD_1 \otimes \cdots \otimes \cD_{\snn} $
to be $|\psi[G^{(1)}]\rangle$.
Let $U_{S,2}$ be a random variable subject to the uniform distribution on $\FF_q^{\sy_2}$.
The shared randomness $R = (R_1,\ldots, R_{2\snn})^T $ is generated as
$R_j :=G_j^{(2)} U_{S,2}$ for $j=1, \ldots, 2\snn$.
The randomness $R$ is distributed so that $j$-th server contains $R_j$ and $R_{\snn +j}$
for $j=1, \ldots, \snn$.
\STEPTWO
\textbf{User's encoding}:
The user 
	randomly encodes the index $K$ to classical queries 
	$Q^{(K)}:=(Q_1^{(K)},\ldots,Q_{\red{2}\snn}^{(K)})^T$, \red{which 
is an ${\bar{\snn}} \times \sff \sx $ random matrix for $k=1, \ldots, \sff$.
The user sends $Q_j^{(K)},Q_{\snn+j}^{(K)}$ to the $j$-th server $\mathtt{serv}_j$.}
\STEPTHREE
\textbf{Servers' encoding}:
The $j$-th server $\mathtt{serv}_j$ applies unitary 
$\mathsf{W}( Q_j^{(K)} \vec{m}+R_j , Q_{\snn+j}^{(K)} \vec{m}+R_{\snn+j})$ on ${\cal D}_j$. 
\STEPFOUR
\textbf{Decoding}:
For a subset $\cA \in \fA$,
the user makes the measurement  given by the POVM
$\{\mathbf{W}_{\cA}(y) (\Tr_{\cA^c}
|\psi[G^{(1)}]\rangle \langle \psi[G^{(1)}]|)
\mathbf{W}_{\cA}^\dagger(y)\}_{y \in \FF_q^{2|\cA|}}$.
Based on the obtained outcome, 
the user outputs the measurement outcome $m$ as the retrieval result.
\end{algorithmic}
\end{Protocol}

The linear CQSPIR protocol with $G^{(1)},G^{(2)},Q^{(K)}$ 
is called
the standard linear CQSPIR protocol with $(G^{(1)},G^{(2)},F)$
when the query $Q^{(K)}$ is given as \eqref{CAP1} by using $F$.
In addition, 
linear CQSPIR protocols discussed in 
\cite{KdW04,SH19,SH19-2}
do not have randomization $U_S \in \FF_q^{\sy_2} $.
That is, $\sy_2=0$ and it does not have the matrix $G^{(2)}$.
Such a protocol is called the randomless standard linear CQSPIR protocol with $(G^{(1)},F)$.
Then, we have the following theorem.

\begin{theo}\Label{Cor3CQ}
Given a $2\snn \times \snn$ self-column-orthogonal matrix $G^{(1)}$,
a $2\snn\times \sy_2$ matrix $G^{(2)}$,
and
a query $Q^{(K)}$,
the following conditions are equivalent.
\begin{description}
\item[(H1)]
The linear CQSPIR protocol with $G^{(1)},G^{(2)},Q^{(K)}$ 
is $(\ACC,\REJ)$-secure.
\item[(H2)]
The linear CSPIR protocol with $(G^{(1)},G^{(2)}),Q^{(K)}$
is $(\bar{\ACC},\bar{\REJ})$-secure.
\end{description}

In addition, we assume that $F$ is a $2\snn\times \sx$ matrix.
the following conditions for $G^{(1)},G^{(2)}$, and $F$ are equivalent.
\begin{description}
\item[(H3)]
The standard linear CQSPIR protocol with $(G^{(1)},G^{(2)},F)$ 
is $(\ACC,\REJ)$-secure.
\item[(H4)]
The standard linear CSPIR protocol with $((G^{(1)},G^{(2)}),F)$
is $(\bar{\ACC},\bar{\REJ})$-secure.
\item[(H5)]
The matrix $((G^{(1)},G^{(2)}),F)$ is an $(\bar{\ACC},\bar{\REJ})$-MMSP.
\end{description}
\end{theo}

Also, we have the following proposition.
\begin{prop}\Label{ZNOSPIRCQ}
In Protocol \ref{protocol5CQ},
even when STEP 4 is replaced by another decoder,
the decoder can be simulated by the decoder given in STEP 3.
That is, once STEPs 1, 2, and 3 are given in Protocol \ref{protocol5CQ},
without loss of generality, we can assume that our decoder is given as STEP 4.
\end{prop}

The above theorem and proposition will be shown after Corollary \ref{Cor2} later.
Combining Theorems \ref{TH3CQ} and \ref{Cor3CQ}, we obtain the following corollary.

\begin{coro}\Label{Coro7CQ}
When 
$\snn \ge \srr > \stt \ge \snn/2>0$,
there exists an $(\srr,\stt,\snn)$-secure linear CQSPIR protocol of rate $2(\srr-\stt)/\snn$.
\end{coro}

\red{The preceding paper \cite{SH20} proposed a protocol 
with the rate $2(\snn-\stt)/\snn$ when $\srr=\snn$.
That is, no existing study considered CQSPIR protocols 
with general qualified sets $\ACC$ including the case with $\srr<\snn$. 
Hence, Corollary \ref{Coro7} can be considered as a generalization of the existing 
result \cite{SH20}.
The comparison with the existing CQSPIR results is summarized in Table \ref{hikaku2}.
In the classical case,} 
the optimal rate of $(\srr,\stt,\snn)$-secure SPIR protocol
is $(\srr-\stt)/\snn$ \cite[Corollary 4]{SH2022}.
Hence, the rate of Corollary \ref{Coro7} is twice of the classical case.
In addition, the rate of CQSPIR cannot exceed $1$ due to the condition $\stt \ge \snn/2$.

\begin{table}[t]
\caption{Comparison for analysis for CQSPIR protocols}
\Label{hikaku2}
\begin{center}
\begin{tabular}{|c|c|c|c|c|}
\hline
 & general & relation   &  threshold type   \\
 & qualified &to  &   qualified  \\
 & set &MMSP &   set  \\
\hline
\cite{KdW04,SH19,SH19-2} &No& No & case with $\srr=\snn$ and special $\stt$   \\
\hline
\cite{SH20} &No & No & case with $\srr=\snn$ and general $\stt$  \\
\hline
\multirow{2}{*}{This paper} & \multirow{2}{*}{Yes} & general case  & general case  \\
&&(Theorem \ref{Cor3CQ})& (Corollary \ref{Coro7CQ})\\
\hline
\end{tabular}
\end{center}
\end{table}

\begin{remark}\Label{RTTY}
One may consider that SPIR for quantum states can be discussed in this framework.
However, a simple application of this method to 
SPIR for quantum states does not work due to the following reasons.
To transmit quantum states, we need to perform a quantum operation across several subsystems.
In the case of SS, only the dealer makes encoding.
Hence, a quantum operation across several subsystems is possible.
However, in the case of SPIR,
several servers perform encoding operations individually.
Hence, it is impossible to perform a quantum operation across several subsystems.
This is reason why 
we cannot apply the same scenario to SPIR for quantum states.
\end{remark}

\section{Quantum SS protocols with preshared entanglement with end-user}\Label{S6-1}
This section introduces EASS protocols and presents our results for EASS protocols, which implies our results for CQSS protocols.

\subsection{Formulation}
Modifying the CQSS setting by allowing prior entanglement between 
the dealer and the end-user, we formulate
an SS protocol with preshared entanglement with user.
Since this problem setting employs entanglement assistance, 
this protocol is called an entanglement-assisted secret sharing (EASS) protocol.
As illustrated in Fig. \ref{fig:SS} (c), an EASS protocol with one dealer, $\snn$ players, and 
one end-user is defined as Protocol \ref{Flow2}.

\begin{Protocol}[H]                  
\caption{EASS protocol}         
\Label{Flow2}      
\begin{algorithmic}
\STEPONE
\textbf{Preparation}:
The dealer and the end-user 
have quantum systems $\cD_{D}$ and $\cD_{E}$, respectively,
and share a state $\rho_{DE}$ on the joint quantum system $\cD_{D}\otimes \cD_{E}$
before the protocol.
\STEPTWO
\textbf{Share generation}:
Depending on the message $m \in \cM$,
the dealer prepares 
$\snn$ shares as the joint system $\cD_1\otimes \cdots \otimes \cD_{\snn}$
by applying a TP-CP map $\Gamma[m]$ from $\cD_{D}$ to $\cD_1\otimes \cdots \otimes \cD_{\snn}$,
and sends the $j$-th share system $\cD_j$ to the $j$-th player.
\STEPTHREE
\textbf{Decoding}:
For a subset $\cA \in \fA$,
the end-user decodes the message from the received state from players $\cA$
by a decoder, which is defined as a POVM
$\mathsf{Dec}(\cA,E) \coloneqq \{ {Y}_{\cA,E}(w) \mid  w\in[\smm] \}$ on 
$\cD[\cA]\otimes \cD_E$.
The end-user outputs the measurement outcome $W$ as the decoded message.
\end{algorithmic}
\end{Protocol}

\begin{defi}[$(\ACC,\REJ)$-EASS] \Label{def:dense-nss}
For an access structure $(\ACC, \REJ)$  on $[\snn]$,
	an EASS protocol defined as Protocol \ref{Flow2} 
	is called $(\ACC,\REJ)$-secure
	if 
	the following conditions are satisfied. $\ACC$-correctness is defined in the same way as Definition \eqref{def:nss}.
\begin{itemize}[leftmargin=1.5em]
\item \textbf{Correctness}:
The relation 
    \begin{align*}
\Tr \Gamma[m](\rho_{DE}) (Y_{\cA,E}(m)\otimes I_{\cA^c} )=1
    \end{align*}
holds for $m \in \cM$.
\item \textbf{Secrecy}:
The state $\Tr_{(\cB,E)^c} \Gamma[m](\rho_{DE})$ does not depend on $m \in \cM$.
\end{itemize}
In particular, when the system $\cD_E$ has the same dimension as $\cD_D$
and the state $\rho_{DE}$ is a maximally entangled state, 
the EASS protocol is called fully EASS (FEASS) protocol. 
\end{defi}

\subsection{Linear FEASS protocol}
Next, we formulate linear protocols with preshared entanglement with user.
Given an ${\bar{\snn}}\times (\sy+\sx)$ matrix $(G,F)$ on the finite field $\FF_q$ 
with ${\bar{\snn}}=2 \snn$, we define
an EASS protocol as follows.
We choose the message set $\cM$ as $\FF_q^{\sx}$,
and choose the Hilbert space ${\cal H}$ as the space spanned by 
$\{ |x\rangle \}_{x \in \FF_q}$.

We define $\cD_j$ and $\cD_{E,j}$ as ${\cal H}$, and define
${\cal D}_D$, $\cD[\cA]$, ${\cal D}_E$, and
${\cal D}_E[\cA]$ as
$\otimes_{j=1}^{\snn}\cD_j$,
$\otimes_{j \in \cA}\cD_j$,
$\otimes_{j=1}^{\snn}\cD_{E,j}$, and 
$\otimes_{j \in \cA}\cD_{E,j}$, respectively, for any subset $\cA \subset [\snn]$.
Hence, ${\cal D}_D$ and ${\cal D}_E$ are ${\cal H}^{\otimes \snn}$.
In the following sections, we adopt the above definitions.
Also, we use the maximally entangled state $
|\phi\rangle:=\sum_{x \in \FF_q}\frac{1}{\sqrt{q}} |x\rangle |x\rangle$
 on ${\cal H}^{\otimes 2}$.
In particular, we define the state $|\Phi\rangle$ as
 the state on the composite 
${\cal D}_D\otimes {\cal D}_E$ whose reduced density on 
$\cD_{j} \otimes \cD_{E,j}$ is $|\phi\rangle$.
\begin{table}[t]
\caption{Symbols for FEASS and EASS protocols}
\Label{symbol1}
\begin{center}
\begin{tabular}{|c|c|c|}
\hline
symbol & meaning & definition   \\
\hline
$\cD_j$ & the j-th share system & $\cH$  \\
\hline
$\cD_{E,j}$ &the reference system of $\cD_j$ &  $\cH$  \\
\hline
${\cal D}_D$& dealer's system&  $\otimes_{j=1}^{\snn}\cD_{j}=\cH^{\otimes n}$ \\
\hline
$\cD[\cA]$ &share system of $\cA$ & $\otimes_{j \in \cA}\cD_j$ \\
\hline
 ${\cal D}_E$& end-user's system &  $\otimes_{j=1}^{\snn}\cD_{E,j}=\cH^{\otimes n}$ \\
\hline
${\cal D}_E[\cA]$ & the reference system of $\cD[\cA]$  & $\otimes_{j \in \cA}\cD_{E,j}$  \\
\hline
\end{tabular}
\end{center}
\end{table}
Then, we define the linear FEASS protocol with $(G,F)$ as Protocol \ref{protocol1}.
\begin{Protocol}[H]                  
\caption{Linear FEASS protocol with $(G,F)$}         
\Label{protocol1}      
\begin{algorithmic}
\STEPONE
\textbf{Preparation}:
We set the initial state $\rho_{DE}$ on ${\cal D}_D \otimes {\cal D}_E$
to be $|\Phi\rangle$.
\STEPTWO
\textbf{Share generation}:
The dealer prepares a uniform random variable $U_D \in \FF_q^{\sy} $.
For $m \in \cM$, the dealer applies 
$\mathbf{W}_{[\snn]}( Fm+GU_D)$ on ${\cal D}_D$. That is,
the encoding operation $\Gamma[m]$ \red{on ${\cal D}_D$} is defined as
\begin{align}
\Gamma[m](\rho):= \sum_{u_D \in \FF_q^{\sy}} 
\frac{1}{q^{\sy}}
\mathbf{W}_{[\snn]}( Fm+Gu_D) \rho 
\mathbf{W}_{[\snn]}^\dagger( Fm+Gu_D).
\end{align}
The shares are given as parts of the state $ \Gamma[m](|\Phi\rangle \langle \Phi|)$.
\STEPTHREE
\textbf{Decoding}:
For a subset $\cA \in \fA$,
the end-user takes partial trace on $\cD_{E}[\cA^c]$, and 
makes measurement on the basis
$\{\mathbf{W}_{\cA}(y) |\Phi\rangle^{|\cA|}\}_{y \in \FF_q^{2|\cA|}}$.
Based on the obtained outcome, the end-user recovers $m$.
\end{algorithmic}
\end{Protocol}

\begin{lemm}\Label{L2U}
The following conditions for $(G,F)$ are equivalent.
\begin{description}
\item[(I1)]
The linear FEASS protocol with $(G,F)$ 
is $(\ACC,\REJ)$-secure.
\item[(I2)]
The linear CSS protocol with $(G,F)$ 
is $(\bar{\ACC},\bar{\REJ})$-secure.
\item[(I3)]
The matrix $(G,F)$ is an $(\bar{\ACC},\bar{\REJ})$-MMSP.
\end{description}
\end{lemm}
\begin{proof}
Since the equivalence between (I2) and (I3) follows from Proposition \ref{prop1},
we show the equivalence between (I1) and (I2).

Given a subset $\cA \subset [\snn]$ and $x \in \FF_q^{{\bar{\snn}}}$, we consider
a linear FEASS protocol with $(G,F)$. 
In this protocol, we have
 \begin{align}
& \Tr_{(\cA,E)^c} \mathbf{W}_{[\snn]}( x) (|\Phi\rangle \langle \Phi|)
 \mathbf{W}_{[\snn]}^\dagger( x)  \nonumber \\
 =&
\mathbf{W}_{\cA}(P_{\bar{\cA}} x) ( |\phi\rangle \langle \phi|^{\otimes |\cA|})
 \mathbf{W}_{\cA}^\dagger(P_{\bar{\cA}} x) \otimes \rho_{mix}^{\otimes |\red{\cA^c}|}.\Label{XZP}
 \end{align}
Hence, the above state can be identified with the 
classical information $ P_{\bar{\cA}} x$.
That is, the analysis on 
the correctness and the secrecy 
in the linear FEASS protocol with $(G,F)$ 
is equivalent with 
the correctness and the secrecy 
in the linear CSS protocol with $(G,F)$. 
We obtain the equivalence between (I1) and (I2).
\end{proof}

\begin{prop}\Label{ZNO}
In Protocol \ref{protocol1},
even when STEP 3 is replaced by another decoder,
the decoder can be simulated by the decoder given in STEP 3.
That is, once STEPs 1 and 2 are given in Protocol \ref{protocol1},
without loss of generality, we can assume that our decoder is given as STEP 3.
\end{prop}

\begin{proof}
When the end-user can access all shares,
any possible state $\Gamma[m](\rho)$ is a diagonal state with respect to the basis 
$\{\bW_{[\snn]}(x) |\Phi\rangle\}_{x \in \FF_q^{2n}}$.
For any subset $\sA \subset [\snn]$, 
the state reduced density with respect to $(\cA E)^c$ is given as \eqref{XZP}.
Since the state $\rho_{mix}^{\otimes |\red{\cA^c}|}$ has no information,
without loss of generality, we can consider that the state is 
$\mathbf{W}_{\cA}(P_{\bar{\cA}} x)  |\phi\rangle \langle \phi|^{}
 \mathbf{W}_{\cA}^\dagger(P_{\bar{\cA}} x)$, which is 
 a diagonal state with respect to the basis 
$\{\bW_{\sA}(x) |\phi\rangle^{\otimes |\cA|}\}_{x \in \FF_q^{2|\cA|}}$.
Therefore, 
even when STEP 3 is replaced by another decoder,
the decoder can be simulated by the decoder given in STEP 3.
 \end{proof}

\subsection{Linear EASS protocol}\Label{S6-1-3}
Next, we focus on 
a $2\snn\times \sy_1$ self-column-orthogonal matrix $G^{(1)}$
and
a $2\snn\times \sy_2$ matrix $G^{(2)}$.
\red{We denote 
$\sy_1$ column vectors of $G^{(1)}$ by $g^{1}, \ldots, g^{\sy_1}$.}
In this case, we propose another EASS protocol.
For this aim, we choose $\snn-\sy_1$ column vectors 
$\bar{g}^1, \ldots, \bar{g}^{\snn-\sy_1}$
such that all the vectors $g^{1}, \ldots, g^{\sy_1}, 
\bar{g}^1, \ldots, \bar{g}^{\snn-\sy_1}$
are orthogonal to each other in the sense of symplectic inner product.
We also choose vectors $h^{1}, \ldots, h^{\sy_1} $
such that 
$\langle h^{j},{g}^{j'}\rangle =\delta_{j,j'}$ and 
$\langle h^{j},{h}^{j'}\rangle=0$.
We denote the matrix
$(h^{1}, \ldots, h^{\sy_1}) $ by $H^{(1)}$.

For $y \in \FF_q^{\sy_1}$ and $x \in \FF_q^{\snn-\sy_1}$,
we define the vector $|x,y\rangle$ as
\begin{align}
\mathbf{W}_{[\snn]}(g_j)|x,y\rangle&=\omega^{y_j} |x,y\rangle \hbox{ for }j=1, \ldots, \sy_1 \\
\mathbf{W}_{[\snn]}(\bar{g}_j)|x,y\rangle&=\omega^{x_j} |x,y\rangle \hbox{ for }j=1, \ldots, \snn-\sy_1\\
\mathbf{W}_{[\snn]}(H^{(1)} \bar{y})|x,y\rangle&= |x,y+\bar{y}\rangle. \Label{CMD}
\end{align}
We define the space 
$\cD_E[y,G^{(1)}]$ as the space spanned by $\{|x,y\rangle\}_{x \in \FF_q^{\snn-\sy_1}}$.
We define the entangled state
\begin{align}
|\Phi[y,G^{(1)}]\rangle:= \sum_{ x \in \FF_q^{\red{n-}\sy_1}}
\frac{1}{\sqrt{q^{\red{n-}\sy_1}}}|x,y\rangle|x,y\rangle.
\Label{MUR}
\end{align}

To consider the relation with CQSS protocols, 
we modify Protocol \ref{protocol1} as follows.
The initial state $|\Phi\rangle$ is replaced by $ |\Phi[0,G^{(1)}]\rangle$.
The random variable $U_D$ is written as $(U_{D,1},U_{D,2})$ with
$U_{D,1} \in \FF_q^{\by_1}$ and $U_{D,2} \in \FF_q^{\by_2}$ so that
$G U_D= G^{(1)} U_{D,1}+G^{(2)} U_{D,2}$.
The applied unitary $ \mathbf{W}_{[\snn]}(F m +G U_{D})$ is replaced by 
$ \mathbf{W}_{[\snn]}(Fm+G^{(2)} U_{D,2})$.
In this protocol, the end-user's space $\cD_E$ is given as $\cD_E[0,G^{(1)}]$.
This protocol is formally written as Protocol \ref{protocol2}, and is called 
the linear EASS protocol with $(G^{(1)},G^{(2)},F)$.
In this notation, 
the first matrix $G^{(1)}$ identifies the initial state,
the second matrix $G^{(2)}$ identifies the direction of the randomization for secrecy,
and the third matrix $F$ identifies the direction of the message imbedding.
That is, 
the linear EASS protocol with $(G^{(1)},G^{(2)},F)$ is different from 
the linear FEASS protocol with $((G^{(1)},G^{(2)}),F)$ because
the latter is characterized as follows;
The initial state is the maximally entangled state between 
$\cD_D \otimes \cD_E$ and the matrix $(G^{(1)},G^{(2)})$ determines 
 the direction of the randomization for secrecy.

\begin{Protocol}[H]                  
\caption{Linear EASS protocol with $(G^{(1)},G^{(2)},F)$}         
\Label{protocol2} 
\begin{algorithmic}
\STEPONE
\textbf{Preparation}:
We set the initial state $\rho_{DE}$ on ${\cal D}_D \otimes {\cal D}_E[0,G^{(1)}]$
to be $ |\Phi[0,G^{(1)}]\rangle$.
\STEPTWO
\textbf{Share generation}:
The dealer prepares a uniform random variable $U_{D,2} \in \FF_q^{\sy_2} $.
For $m \in \cM$, the dealer applies 
$\mathbf{W}_{[\snn]}( Fm+G^{(2)} U_{D,2})$ on ${\cal D}_D$. That is,
the encoding operation $\Gamma[m]$ \red{on ${\cal D}_D$} is defined as
\begin{align}
\Gamma[m](\rho):= \sum_{u_{D,2} \in \FF_q^{\sy_2}} 
\frac{1}{q^{\sy_2}}
\mathbf{W}_{[\snn]}( Fm+G^{(2)} u_{D,2}) \rho 
\mathbf{W}_{[\snn]}^\dagger( Fm+G^{(2)} u_{D,2}).
\end{align}
The shares are given as parts of the state 
$ \Gamma[m](|\Phi[0,G^{(1)}]\rangle \langle \Phi[0,G^{(1)}]|)$.
\STEPTHREE
\textbf{Decoding}:
For a subset $\cA \in \fA$,
the end-user takes partial trace on $\cD_{E}[\cA^c]$, and 
makes the measurement  given by the POVM
$\big\{\mathbf{W}_{\cA}(z) (
\Tr_{\cD[\cA^c],\cD_{E}[\cA^c]}
|\Phi[0,G^{(1)}]\rangle \langle \Phi[0,G^{(1)}]|)
\mathbf{W}_{\cA}^\dagger(z)\big\}_{z \in \FF_q^{2|\cA|}}$.
Based on the obtained outcome, the end-user recovers $m$.
\end{algorithmic}
\end{Protocol}

\red{In particular, 
when $\sy_2=0$, the protocol does not have the matrix $G^{(2)}$
and the protocol does not have the random variable $U_{D,2}$.
Such a protocol is called the randomless linear EASS protocol with $(G^{(1)},F)$.
Similar to Lemma \ref{L2U}, we have the following theorem for linear EASS protocols, 
Protocol \ref{protocol2}.}

\begin{theo}\Label{TH1}
Given a $2\snn\times \sy_1$ self-column-orthogonal matrix $G^{(1)}$,
a $2\snn\times \sy_2$ matrix $G^{(2)}$,
and
a $2\snn\times \sx$ matrix $F$, 
the following conditions 
are equivalent.
\begin{description}
\item[(J1)]
The linear EASS protocol with $(G^{(1)},G^{(2)},F)$
is $(\ACC,\REJ)$-secure.
\item[(J2)]
The linear FEASS protocol with $((G^{(1)},G^{(2)}),F)$
is $(\ACC,\REJ)$-secure.
\item[(J3)]
The linear CSS protocol with $((G^{(1)},G^{(2)}),F)$
is $(\bar{\ACC},\bar{\REJ})$-secure.
\item[(J4)]
The matrix $((G^{(1)},G^{(2)}),F)$ is an $(\bar{\ACC},\bar{\REJ})$-MMSP.
\end{description}
\end{theo}

The above theorem shows the equivalence between a spacial class of 
linear FEASS protocols and linear EASS protocols.
In addition, the equivalence between (J1) and (J2), 
we have a statement similar to Proposition \ref{ZNO} as follows
while its proof is given in Section \ref{S6-2}.

\begin{prop}\Label{ZNO2}
In Protocol \ref{protocol2},
even when STEP 3 is replaced by another decoder,
the decoder can be simulated by the decoder given in STEP 3.
That is, once STEPs 1 and 2 are given in Protocol \ref{protocol2},
without loss of generality, we can assume that our decoder is given as STEP 3.
\end{prop}

\red{When $\sy_1$ is ${\bar{\snn}}/2=\snn$,
$\cD_{E}[0,G^{(1)}]$ is a one-dimensional system
and the summand for $x$ does not appear in 
\eqref{MUR}.
Hence, 
the state $|\Phi[0,G^{(1)}]\rangle$ is a product state, and  
can be considered as a state on $\cD_D$. 
Since the state on $\cD_E$ in Protocol \ref{protocol2} 
is fixed to $|0,0\rangle$, 
Protocol \ref{protocol2} is essentially the same as Protocol \ref{protocol1CQ}, 
a linear CQSS protocol.
Therefore, as corollaries of Theorem \ref{TH1} and Proposition \ref{ZNO2}, 
we obtain Theorem \ref{Cor1} and Proposition \ref{ZNOCQ}.}

Since $G^{(1)}$ is self-column-orthogonal,
the group
$N:=\{ G^{(1)} y \}_{y \in \FF_q^{(\snn)}}$
satisfies the self-orthogonal condition 
$N\subset N^{\perp}:=
\{  v\in \FF_q^{2\snn}| (v, v')=0 ,\quad \forall v' \in N  \}$.
In particular, the dimension of $N$ is $\snn$,
$N= N^{\perp}$.
Hence, the state $|\Phi[0,G^{(1)}]\rangle$ is given as the stabilizer state of 
$N$.
That is, the encoded state is given as 
the application of $ \Gamma[m]$ to
the stabilizer state of $N$.

Therefore, 
\red{as a generalization of $(\srr,\stt,\snn)$-CQMMSP,}
we define the following special case of $(\ACC,\REJ)$-MMSPs.

\begin{defi}[$(\srr,\stt,\snn)$-EAMMSP]
We choose $\ACC = \{ \cA \subset [\snn] \mid |\cA| \geq \srr \}$
and $\REJ = \{ \cB \subset [\snn] \mid |\cB| \leq \stt \}$.
Given a $2\snn \times \sy_1$ self-column-orthogonal matrix $G^{(1)}$,
a $2\snn\times \sy_2$ matrix $G^{(2)}$,
and a $2\snn\times \sx$ matrix $F$, 
the matrix $(G^{(1)},G^{(2)},F)$ is called an $(\srr,\stt,\snn)$-EAMMSP
when the matrix $((G^{(1)},G^{(2)}),F)$ is an $(\bar{\ACC},\bar{\REJ})$-MMSP.
\end{defi}

\begin{theo}\Label{TH3}
For any positive integers $\srr,\stt,\snn,\sy_1$ with 
$\snn \ge \srr > \stt \ge \sy_1/2>0$
and any prime $p$,
there exist a positive integer $s$,
a $2\snn \times \sy_1$ self-column-orthogonal matrix $G^{(1)}$,
a $2\snn\times (2\stt-\sy_1)$ matrix $G^{(2)}$,
and a $2\snn\times (2\srr-2\stt)$ matrix $F$ on $\FF_{q}$ with $q=p^s$
such that 
the matrix $(G^{(1)},G^{(2)},F)$ is an $(\srr,\stt,\snn)$-EAMMSP.
\end{theo}
Theorem \ref{TH3} is shown in Appendix \ref{A2}.
Combining Theorem \ref{TH3}, we obtain the following corollary.

\begin{coro}\Label{Coro6}
When 
$\snn \ge \srr > \stt >0$,
there exists an $(\srr,\stt,\snn)$-secure EASS protocol with rate $2(\srr-\stt)/\snn$.
\end{coro}

In the case of classical case, 
the optimal rate of $(\srr,\stt,\snn)$-secure SS protocol
is $(\srr-\stt)/\snn$ \cite{Ogata,Okada,Paillier}.
Hence, the rate of Corollary \ref{Coro6} is twice of the classical case.
However, the rate of CQSS cannot exceed $1$ due to the condition $\stt \ge \snn/2$.
This constraint always holds beyond the condition in Corollary \ref{Coro6} 
because CQSS does not have shared entanglement.
Since EASS has shared entanglement, the rate of CQSS exceeds $1$ by removing 
the condition $\stt \ge \snn/2$,
which can be considered as an advantage of EASS over CQSS.
\red{In addition, 
in the case of threshold type, i.e., the case with $\srr=\stt+1$,
a CQSS protocol requires the condition $\stt \ge \snn/2$ in Corollary \ref{Coro6CQ}.
In contrast, an EASS protocol works even with 
$\snn/2 \ge \srr=\stt+1$, which is another advantage of use of preshared entanglement with user
over CQSS protocols.}

\begin{remark}
We compare the above EASS protocols, Protocol \ref{protocol2},
with the combination of QQSS protocol and dense coding.
As discussed in Corollary \ref{Coro9},
QQSS protocol has the rate $(\srr-\max(\stt, \snn-\srr))/\snn$
under the condition $\snn \ge \srr > (\snn+1) /2$.
Combining it with dense coding,  
the obtained protocol has the rate $2(\srr-\max(\stt, \snn-\srr))/\snn$.
In contrast,
as mentioned in Corollary \ref{Coro6},
the EASS protocol has the rate $2(\srr-\stt)/\snn$
under the condition $\snn \ge \srr > \stt >0$.
Since $2(\srr-\stt)/\snn - 2(\srr-\max(\stt, \snn-\srr))/\snn=
\max(0, \snn-\srr-\stt)/\snn\ge 0$,
the EASS protocol has a strictly better performance 
than the simple combination of QQSS protocol and dense coding
when $\snn-\srr-\stt>0$.
\end{remark}

\red{Further,} as a special case of EASS protocol, we define the EA version of MDS codes as follows.
This concept is useful for discussion the QQ version of MDS codes.

\begin{defi}[$(\snn, \sx)$-EAMDS code]
We consider the case with $\sy_2=0$.
Assume that $ G^{(1)}$ is a $\snn\times \sy_1$ self-column-orthogonal matrix and 
$F$ is a $\snn\times \sx$ matrix. 
We say that the randomless linear EASS protocol with $(G^{(1)}, F)$ is a 
$(\snn, \lceil \frac{\sy_1+\sx}{2} \rceil)$-EAMDS code
when it is $\fA$-correct with $\fA= \{ \cA \subset [\snn] \mid |\cA| \geq 
\lceil \frac{\sy_1+\sx}{2} \rceil\}$.
\end{defi}

By considering the case when $\fB$ is empty set,
Theorem \ref{TH1} implies the following lemma.
\begin{lemm}\Label{Coro17}
Assume that $ G^{(1)}$ is a $\snn\times \sy_1$ self-column-orthogonal matrix and 
$F$ is a $\snn\times \sx$ matrix. 
The randomless linear EASS protocol with $(G^{(1)}, F)$ is a 
$(\snn, \lceil \frac{\sy_1+\sx}{2} \rceil)$-EAMDS code
if and only if the linear CSS protocol with $(G^{(1)}, F)$ is
$\bar{\fA}$-correct with 
$\fA= \{ \cA \subset [\snn] \mid |\cA| \geq 
\lceil \frac{\sy_1+\sx}{2} \rceil\}$.
\end{lemm}

\subsection{Proofs of Theorem \ref{TH1} and Proposition \ref{ZNO2} 
}\Label{S6-2}
We show only Theorem \ref{TH1} by using the idea by references \cite{AMTW,SP}.
Lemma \ref{L2U} guarantees the equivalence among (J2), (J3), and (J4).
In the following, we show the equivalence between (J1) and (J2), and Proposition \ref{ZNO2}.

As the preparation, we notice the relation;
\begin{align}
&\sum_{y \in \FF_q^{\sy_1}}\frac{1}{q^{\sy_1}} 
\mathbf{W}_{[\snn]}(G^{(1)} y) |\Phi\rangle \langle \Phi|
\mathbf{W}_{[\snn]}^\dagger(G^{(1)} y) \nonumber \\
=&\sum_{y \in \FF_q^{\sy_1}}\frac{1}{q^{\sy_1}} |\Phi[y,G^{(1)}]\rangle\langle \Phi[y,G^{(1)}]|
\Label{MTX}\\
\stackrel{(a)}{=}&\mathbf{W}_{\cA}(P_{\cA} H^{(1)} y) \otimes \mathbf{W}(H^{(1)} y)
|\Phi[0,G^{(1)}]\rangle\langle \Phi[0,G^{(1)}]|
(\mathbf{W}_{\cA}(P_{\cA} H^{(1)} y) \otimes \mathbf{W}(H^{(1)} y))^\dagger
\Label{MTX3},
\end{align}
where $(a)$ follows from \eqref{CMD}.
Taking the partial trace on $\cD[\cA^c]$ and $\cD_{E}[\cA^c]$, for $z \in \FF_q^{2|\cA|}$,
we have
\begin{align}
\bar{\Pi}_z
:=&
\mathbf{W}_{\cA}( z) \Big(
\sum_{y \in \FF_q^{\sy_1}}\frac{1}{q^{\sy_1}} 
\mathbf{W}_{\cA}(P_{\cA} G^{(1)} y) 
\Big(\Tr_{\cD[\cA^c],\cD_{E}[\cA^c]} |\Phi\rangle \langle \Phi| \Big)
\mathbf{W}_{\cA}^\dagger(P_{\cA} G^{(1)} y) \Big)
\mathbf{W}_{\cA}( z)^\dagger  \Label{MTX4}\\
=&
\mathbf{W}_{\cA}(z)
\Tr_{\cD[\cA^c],\cD_{E}[\cA^c]}
\Big(\sum_{y \in \FF_q^{\sy_1}}\frac{1}{q^{\sy_1}} 
\mathbf{W}_{[\snn]}(G^{(1)} y) |\Phi\rangle \langle \Phi|
\mathbf{W}_{[\snn]}^\dagger( G^{(1)} y) \Big)
\mathbf{W}_{\cA}( z)^\dagger
\nonumber \\
=&
\hat{\Pi}_{z}:=\mathbf{W}_{\cA}( z)
\Big(\sum_{y \in \FF_q^{\sy_1}}\frac{1}{q^{\sy_1}}
\Tr_{\cD[\cA^c],\cD_{E}[\cA^c]}
 |\Phi[y,G^{(1)}]\rangle\langle \Phi[y,G^{(1)}]|\Big)
\mathbf{W}_{\cA}( z)^\dagger .
\Label{MTX2}
\end{align}

To analyze the linear EASS protocol with $(G^{(1)},G^{(2)},F)$, i.e., 
Protocol \ref{protocol2}, we consider 
the modified linear EASS protocol with $(G^{(1)},G^{(2)},F)$, which is defined as 
Protocol \ref{protocol3}.

\begin{Protocol}[H]                  
\caption{Modified linear EASS protocol with $(G^{(1)},G^{(2)},F)$}         
\Label{protocol3}      
\begin{algorithmic}
\STEPONE
\textbf{Preparation}:
We set the initial state $\rho_{DE}$ on ${\cal D}_D \otimes {\cal D}_E$
to be $\sum_{y \in \FF_q^{\sy_1}}\frac{1}{q^{\sy_1}} |\Phi[y,G^{(1)}]\rangle\langle \Phi[y,G^{(1)}]|$.
\STEPTWO
\textbf{Share generation}:
The dealer prepares a uniform random variable $U_D \in \FF_q^{\sy} $.
For $m \in \cM$, the dealer applies 
$\mathbf{W}_{[\snn]}( Fm+G^{(2)} U_{D,2})$ on ${\cal D}_D$. That is,
the encoding operation $\Gamma[m]$ \red{on ${\cal D}_D$} is defined as
\begin{align}
\Gamma[m](\rho):= \sum_{u_D \in \FF_q^{\sy}} 
\frac{1}{q^{\sy}}
\mathbf{W}_{[\snn]}( Fm+G^{(2)} U_{D,2}) \rho 
\mathbf{W}_{[\snn]}^\dagger( Fm+G^{(2)} U_{D,2}).
\end{align}
The shares are given as parts of the state 
\red{$\Gamma[m](
\sum_{y \in \FF_q^{\sy_1}}\frac{1}{q^{\sy_1}} |\Phi[y,G^{(1)}]\rangle\langle \Phi[y,G^{(1)}]|)$}.
\STEPTHREE
\textbf{Decoding}:
For a subset $\cA \in \fA$,
the end-user takes partial trace on $\cD_{E}[\cA^c]$, and 
makes measurement on the basis
$\{\mathbf{W}_{\cA}(\red{z}) |\phi\rangle^{\otimes |\cA|}\}_{\red{z} \in \FF_q^{2|\cA|}}$.
Based on the obtained outcome, the end-user recovers $m$.
\end{algorithmic}
\end{Protocol}

The relation \eqref{MTX} guarantees 
that their final state in STEP2 are the same, which implies that
Protocol \ref{protocol1} has the same performance as Protocol \ref{protocol3}.

Protocol \ref{protocol2}
and 
Protocol \ref{protocol3} are converted to each other 
in decoding process as follows.
For a subset $\cA \in \fA$, we modify the decoder of Protocol \ref{protocol2} as follows.
First, the end-user randomly generates $Y \in \FF_q^{\sy_1} $ subject to the uniform distribution.
The end-user applies the unitary 
$\mathbf{W}_{\cA}(P_{\cA} H^{(1)} Y)
\otimes \mathbf{W}(H^{(1)} Y)
 $ on 
$(\otimes_{j \in \cA}\cD_j)\otimes \cD_{E}$.
Then, the end-user applies the measurement $\{\hat{\Pi}_{z} \}_{z \in \FF_q^{2|\cA|}}$
defined in \eqref{MTX2}.
This measurement has the same output statistics as the measurement given in STEP3 of Protocol \ref{protocol2}.

Due to \eqref{MTX3},
the resultant state by the above unitary application
on
$(\otimes_{j \in \cA}\cD_j)\otimes \cD_{E}$ is the same state as the final state of STEP 2 of 
Protocol \ref{protocol3}.
Since the final state of STEP 2 of Protocol \ref{protocol3} is invariant for 
$ \mathbf{W}_{\cA}(P_{\cA} G^{(1)} y) $,
the measurement in STEP 3 of Protocol \ref{protocol3} can be replaced by the POVM
$\{\bar{\Pi}_z \}_{z \in \FF_q^{2|\cA|}}$ defined in \eqref{MTX4}.
The relation \eqref{MTX2} guarantees that the POVM is the same as the POVM given in 
STEP 3 of Protocol \ref{protocol2}.
Therefore, the decoder of Protocol \ref{protocol2}
has the same output statistics as the decoder of Protocol \ref{protocol3}.
That is , the correctness and the secrecy of 
Protocol \ref{protocol1} is equivalent to those of Protocol \ref{protocol2}.
Hence, we obtain the equivalence between (J1) and (J2).

Next, we proceed to the proof of Proposition \ref{ZNO2}.
Any decoding measurement $\{\Pi_\omega\}_{\omega}$ in Protocol \ref{protocol2} is given as a POVM on
the space $\cD_D \otimes \cD_E[0,G^{(1)}]$.
For a subset $\cA \in \fA$, we modify this decoding measurement as follows.
First, the end-user randomly generates $Y \in \FF_q^{\sy_1} $ subject to the uniform distribution.
The end-user applies the unitary 
$\mathbf{W}_{\cA}(P_{\cA} H^{(1)} Y)
\otimes \mathbf{W}(H^{(1)} Y)
 $ on 
$(\otimes_{j \in \cA}\cD_j)\otimes \cD_{E}$.
Then, the end-user applies the measurement 
$\{
\mathbf{W}_{\cA}(-P_{\cA} H^{(1)} Y)
\Pi_\omega
\mathbf{W}_{\cA}(P_{\cA} H^{(1)} Y) \}_{\omega}$ to 
$\cD_D \otimes \cD_E[Y,G^{(1)}]$.
This modified measurement has the same output statistics as the original measurement $\{\Pi_\omega\}_{\omega}$
in Protocol \ref{protocol2}.
Also, we define the POVM
$\{\tilde{\Pi}_\omega\}_{\omega}$ on 
$\cD_D \otimes \cD_E :=\oplus_{y \in \FF_q^{\sy_1}} \cD_D \otimes \cD_E[y,G^{(1)}]$
as $\tilde{\Pi}_\omega
:=\sum_{y \in \FF_q^{\sy_1}}
\mathbf{W}_{\cA}(- P_{\cA} H^{(1)} y)
\Pi_\omega \mathbf{W}_{\cA}(P_{\cA} H^{(1)} y)$.
Since the state on $\cD_D \otimes \cD_E$ is the same as 
the same state as the final state of STEP 2 of Protocol \ref{protocol3}, which is the same as
the final state of STEP 2 of Protocol \ref{protocol1}.
Due to Proposition \ref{ZNO}, 
the output of this measurement can be simulated by the decoder given in STEP 3 of 
Protocol \ref{protocol1}.
The above proof of Theorem \ref{TH1}
guarantees that 
the decoder given in STEP 3 of 
Protocol \ref{protocol1} can be simulated by 
the decoder given in STEP 3 of 
Protocol \ref{protocol2}.
Hence, we obtain Proposition \ref{ZNO2}.

\section{Proofs of Theorem \ref{TH4QQ} and Lemma \ref{Cor73QQ}
and decoder for linear QQSS protocol}\Label{SS-7}
This section presents the proofs for our results of QQSS protocols stated in
Section~\ref{SS-5} by showing notable relations between dense coding and noiseless quantum state transmission.

\subsection{Dense coding and noiseless quantum state transmission}\Label{S7}
For a preparation for discussion on QQSS protocols,
we investigate the relation between dense coding and noiseless quantum state transmission.
When noiseless transmission with $d$-dimensional quantum system is allowed,
the noiseless classical message 
transmission with size $d^2$ is possible 
with shared entangled state.
It is called dense coding.
Here, we investigate its converse statement.

\red{We choose $\cH_A$ as $\cH^{\otimes \snn'}$.
We prepare the system ${\cal H}_{R}$ as the same dimensional system as ${\cal H}_{A}$.
Here, to clarify that the operator 
$\overline{\mathbf{W}}_{[\snn']}( x) $ is applied on $\cH_A$, 
we denote it by 
$\overline{\mathbf{W}}_{A,[\snn']}( x)$.
This usage of the subscript will be applied in the latter parts.}
For a given a POVM $\Pi= \{\Pi_{x}\}_{x\in \FF_q^{2 \snn'}}$
on the joint system $\cH_R \otimes \cH_B$,
we define the TP-CP map \red{$\overline{\Gamma}[\Pi]$ 
from the system $\cH_B$ to the system $\cH_A$} 
as follows.
We prepare the maximally entangled state 
$ |\phi\rangle \langle \phi|^{\otimes \snn'}$ on 
${\cal H}_{R}\otimes {\cal H}_{A}$.
Then, $\overline{\Gamma}[\Pi]$ is defined as
\begin{align}
\overline{\Gamma}[\Pi](\rho):=
\sum_{x} \overline{\mathbf{W}}_{A,[\snn']}^\dagger( x)
\Tr_{B,R}( \rho \otimes |\phi\rangle \langle \phi|^{\otimes \snn'} \Pi_{x} )
\overline{\mathbf{W}}_{A,[\snn']}( x)
\Label{XMT}
\end{align}
for a density $\rho$ on the system $\cH_B$.

\begin{lemm}\Label{L5}
Let ${\cal H}_A$ and ${\cal H}_R$ be the systems equivalent to
$\cH^{\otimes \snn'}$ and ${\cal H}_B$ be an arbitrary quantum system.
Let $\Lambda$ be a TP-CP map from
${\cal H}_A$ to ${\cal H}_B$.
Assume that a POVM $\Pi= \{\Pi_{x}\}_{x\in \FF_q^{2 \snn'}}$ 
on the joint system $\cH_R \otimes \cH_B$
satisfies that
\begin{align}
\Tr \Pi_{x'} \Lambda (\mathbf{W}_{A,[\snn']}( x) 
|\phi\rangle \langle \phi|^{\otimes \snn'}
\mathbf{W}_{A,[\snn']}^\dagger( x))=\delta_{x,x'}.\Label{ACL}
\end{align}
Then, 
$\overline{\Gamma} [\Pi] \circ \Lambda$ is the identity channel.
\end{lemm}

\begin{proof}
We define the adjoint map $\Lambda^*$ from the Hermitian matrices on 
$\cH_B$ to those on $\cH_A$ as
\begin{align}
\Tr \Lambda (X) Y =\Tr X \Lambda^*(Y). 
\end{align}
Then, we have
\begin{align}
&\overline{\Gamma} [\Pi] \circ \Lambda(\rho) \nonumber \\
=&\sum_{x} \overline{\mathbf{W}}_{A,[\snn']}^\dagger( x)
\Tr_{B,R}( \Lambda(\rho) \otimes |\phi\rangle \langle \phi|^{\otimes \snn'} 
\Pi_{x} )
\overline{\mathbf{W}}_{A,[\snn']}( x)  \nonumber \\
=& \sum_{x} \overline{\mathbf{W}}_{A,[\snn']}^\dagger( x)
\Tr_{B,R}( \rho \otimes |\phi\rangle \langle \phi|^{\otimes \snn'} 
\Lambda^*(\Pi_{x}) )
\overline{\mathbf{W}}_{A,[\snn']}( x) .\Label{ATE}
\end{align}
Due to the condition \eqref{ACL}, 
$\{ \Lambda^*( \Pi_{x})\}$ equals the POVM $
\{\mathbf{W}_{A,[\snn']}( x) |\phi\rangle \langle \phi|^{\otimes \snn'})
\mathbf{W}_{A,[\snn']}^\dagger( x)\}$.
Hence, the process in \eqref{ATE} can be considered as the process of quantum teleportation.
Thus,
$\overline{\Gamma} [\Pi] \circ \Lambda$ is the identity channel.
\end{proof}

\begin{lemm}\Label{L6}
Let ${\cal H}_A$ and ${\cal H}_R$ be the systems equivalent to 
$\cH^{\otimes \snn'}$ and ${\cal H}_B$ be an arbitrary quantum system.
Let $\Lambda$ be a TP-CP map from
${\cal H}_A$ to ${\cal H}_B$.
We generate the random variable $X\in \FF_q^{2 \snn'}$ 
subject to the uniform distribution.
Using the variable $X$,
we generate the state $ 
\Lambda (\mathbf{W}_{A,[\snn']}( X) |\phi\rangle \langle \phi|^{\otimes \snn'}
\mathbf{W}_{A,[\snn']}^\dagger( X))$.
We consider the mutual information between $X$ and the joint system $BR$.
That is, we consider the state 
$\rho_{BRX}:= \sum_{x \in \FF_q^{2 \snn'}}\frac{1}{q^{2\snn'}}
|x\rangle \langle x| \otimes  
\Lambda (\mathbf{W}_{A,[\snn']}( x) |\phi\rangle \langle \phi|^{\otimes \snn'}
\mathbf{W}_{A,[\snn']}^\dagger( x) )$.
Also, we consider another state $\sigma_{BR}:=
\Lambda ( |\phi\rangle \langle \phi|^{\otimes \snn'})$.
Then, we have the following relation.
\begin{align}
I(X;BR)[\rho_{BRX}]=I(R;B)[\sigma_{RB}].
\end{align}
Therefore, the joint system $BR$ has no information for the random variable $X$ in the dense coding scheme
if and only if the output system $B$ of the channel $\Lambda$ has no information for the inputs state on $\cH_A$.
\end{lemm}
\begin{proof}
For $x=(a,b)$, we define $\bar{x}=(a,-b)$.
Then, we have
\begin{align}
&I(X;BR)[\rho_{BRX}] \nonumber \\
=&
 \sum_{x \in \FF_q^{2 \snn'}}\frac{1}{q^{2\snn'}}
D\bigg( \Lambda \Big(\mathbf{W}_{A,[\snn']}( x) 
|\phi\rangle \langle \phi|^{\otimes \snn'}
\mathbf{W}_{A,[\snn']}^\dagger( x) \Big)
  \nonumber \\
&\hspace{6ex} \bigg\|\Lambda \Big(
 \sum_{x' \in \FF_q^{2 \snn'}}\frac{1}{q^{2\snn'}}
\mathbf{W}_{A,[\snn']}( x') |\phi\rangle \langle \phi|^{\otimes \snn'}
\mathbf{W}_{A,[\snn']}^\dagger( x') \Big)
\bigg)
\nonumber \\
=&
 \sum_{x \in \FF_q^{2 \snn'}}\frac{1}{q^{2\snn'}}
D\bigg( \Lambda \Big(\mathbf{W}_{R,[\snn']}( \bar{x}) 
|\phi\rangle \langle \phi|^{\otimes \snn'}
\mathbf{W}_{R,[\snn']}^\dagger( \bar{x}) \Big)
  \nonumber \\
&\hspace{6ex}
\bigg\| 
\Lambda (\rho_{mix,A})\otimes \rho_{mix,R}
\bigg)
\nonumber \\
=&
 \sum_{x \in \FF_q^{2 \snn'}}\frac{1}{q^{2\snn'}}
D\Big( \Lambda ( |\phi\rangle \langle \phi|^{\otimes \snn'})
  \nonumber \\
&\hspace{6ex}\Big\| 
\Lambda (\rho_{mix,A})\otimes 
\mathbf{W}_{R,[\snn']}^\dagger( \bar{x})
\rho_{mix,R} \mathbf{W}_{R,[\snn']}( \bar{x})
\Big) \nonumber \\
=&
D\Big( \Lambda ( |\phi\rangle \langle \phi|^{\otimes \snn'})
\Big\| 
\Lambda (\rho_{mix,A})\otimes 
\rho_{mix,R} 
\Big)
=I(R;B)[\sigma_{RB}].
\end{align}
\end{proof}

\subsection{Proofs of Theorem \ref{TH4QQ} and Lemma \ref{Cor73QQ}
and decoder for linear QQSS protocol}\Label{S8-2}
Now, we show Theorem \ref{TH4QQ} and Lemma \ref{Cor73QQ}
by using the contents of Section \ref{S7}.
Considering linear EASS protocols,
we restate Theorem \ref{TH4QQ} as follows.
\begin{theo}\Label{TH4}
Given a $2\snn\times (\snn-\sx)$ self-column-orthogonal matrix $G^{(1)}$
and
a $2\snn\times \sy_2$ matrix $G^{(2)}$,
we choose a $2\snn\times 2 \sx$ matrix $F$ column-orthogonal to $G^{(1)}$.
Then, the following conditions for $G^{(1)},G^{(2)},F$
are equivalent.
\begin{description}
\item[(J1)]
The linear QQSS protocol with $(G^{(1)},G^{(2)},F)$
is $(\ACC,\REJ)$-secure. That is, there exists a suitable TP-CP map $\overline{\Gamma}$ to recover the original state
in STEP 2.
\item[(J2)]
The linear CSS protocol with $((G^{(1)},G^{(2)}),F)$
is $(\bar{\ACC},\bar{\REJ})$-secure.
\item[(J3)]
The matrix $((G^{(1)},G^{(2)}),F)$ is an $(\bar{\ACC},\bar{\REJ})$-MMSP.
\item[(J4)]
The linear EASS protocol with $(G^{(1)},G^{(2)},F)$
is $(\ACC,\REJ)$-secure.
\end{description}
In particular, when Condition (J4) holds,
the decoder of the linear QQSS protocol with $(G^{(1)},G^{(2)},F)$ is given as 
$\overline{\Gamma}[\Pi]$ defined in \eqref{XMT},
where the POVM $\Pi= \{\Pi_{m}\}_{m\in \FF_q^{2 \sx}}$ is the decoder of 
the linear EASS protocol with $(G^{(1)},G^{(2)},F)$.
\end{theo}

In the following proof, we give a construction of decoder for Condition (J1).

\begin{proof}
Since \red{Theorem \ref{TH1}} guarantees the equivalence among the conditions (J2), (J3), and (J4), 
we show only the equivalence between  the conditions (J1) and (J4).

First, we show the direction (J1)$\Rightarrow$(J4).
We assume Condition (J1).
We combine the linear QQSS protocol with $(G^{(1)},G^{(2)},F)$
and dense coding \cite{BW}.
Then, the obtained protocol is 
the linear EASS protocol with $(G^{(1)},G^{(2)},F)$.
The correctness and secrecy of
the linear QQSS protocol with $(G^{(1)},G^{(2)},F)$ for $(\ACC,\REJ)$
imply those of 
the linear EASS protocol with $(G^{(1)},G^{(2)},F)$ for $(\ACC,\REJ)$.
Hence, Condition (J4) holds.

Next, we show the direction (J4)$\Rightarrow$(J1).
We assume Condition (J4).
Hence, the linear EASS protocol with $(G^{(1)},G^{(2)},F)$ 
satisfies the correctness with respect to $\cA \in \fA$.
We choose 
$\cH_A$ as $\cD_D[0,G^{(1)}]$
and 
$\cH_B$ as $\otimes_{j \in \cA}\cD_j$.
Hence, we choose $\snn'$ to be $\sx$.
Since $F$ is column orthogonal to $G^{(1)}$, 
the action of the unitaries $\{\bW_{[\snn]}(F m)\}_{m \in \FF_q^{2\sx}}$
preserves the subspace $\cD_D[0,G^{(1)}]$.
We choose $\Lambda$ as
\begin{align}
\Lambda(\rho)  :=
\Tr_{\cA^c} \sum_{u_{D,2} \in \FF_q^{\sy_2}} 
\frac{1}{q^{\sy_2}}
\mathbf{W}_{[\snn]}(G^{(2)} u_{D,2}) \rho 
\mathbf{W}_{[\snn]}^\dagger( G^{(2)} u_{D,2}).
\end{align}
Due to the correctness of 
the linear EASS protocol with $(G^{(1)},G^{(2)},F)$ 
for $\cA \in \fA$,
there is a POVM $\Pi= \{\Pi_{m}\}_{m\in \FF_q^{2 \sx}}$ 
on the joint system $\cH_R \otimes \cH_B$
such that
\begin{align}
\Tr \Pi_{m'} \Lambda (\mathbf{W}_{[\snn]}(F m) 
|\phi\rangle \langle \phi|^{\otimes \sx}
\mathbf{W}_{[\sx]}^\dagger(F m))=\delta_{m,m'}.\Label{ACL2}
\end{align}
We choose $\overline{\Gamma}[\Pi]$ defined in \eqref{XMT}.
Then, Lemma \ref{L5} guarantees that
$\overline{\Gamma}[\Pi]\circ \Lambda$ is the identity channel.
Hence, we obtain
 the correctness of the linear QQSS protocol with $(G^{(1)},G^{(2)},F)$
with respect to $\cA \in \fA$ when 
$\overline{\Gamma}[\Pi]$ is chosen as the decoder.

In addition, 
the linear EASS protocol with $(G^{(1)},G^{(2)},F)$ 
satisfies the secrecy with respect to $\cB \in \fB$.
We choose 
$\cH_A$ as $\cD_D[0,G^{(1)}]$ and 
$\cH_B$ as $\otimes_{j \in \cB}\cD_j$.
We apply Lemma \ref{L6} in the same way as the above.
In this application, 
$I(R;B)[\sigma_{RB}]$ expresses the information obtained by the players in $\cB$ 
under the above linear QQSS protocol with $(G^{(1)},G^{(2)},F)$,
and 
$I(X;BR)[\rho_{RBX}]$ expresses the information obtained by the players in $\cB$ 
under the linear EASS protocol with $(G^{(1)},G^{(2)},F)$.
Hence,
the linear EASS protocol with $(G^{(1)},G^{(2)},F)$ 
with respect to $\cB \in \fB$ implies
the secrecy of 
the linear QQSS protocol with $(G^{(1)},G^{(2)},F)$
with respect to $\cB \in \fB$. 
Hence, Condition (J1) holds.
\end{proof}

Combining Lemma \ref{Coro17} and the same idea as Theorem \ref{TH4},
we obtain the following corollary.
This corollary is a recasted statement of 
Lemma \ref{Cor73QQ} with adding the new condition (F3). 

\begin{coro}\Label{Cor73}
Given a $2\snn \times 2(\snn-\srr)$ self-column-orthogonal matrix $G^{(1)}$
and $2\snn\times 2(2\srr-\snn)$ matrix $F$ column-orthogonal to $G^{(1)}$,
the following conditions are equivalent.
\begin{description}
\item[(F1)]
The randomless linear QQSS protocol with $(G^{(1)},F)$ 
is an $(\snn, \srr)$-QQMDS code.
\item[(F2)]
The matrix $(G^{(1)},F)$ accepts
$\bar{\ACC}$ with $\ACC= \{ \cA \subset [\snn] \mid |\cA| \geq \srr \}$.
\item[(F3)]
The randomless linear EASS protocol with $(G^{(1)},F)$
is an $(\snn, \srr)$-EAMDS code.
\end{description}
\end{coro}

\begin{proof}
The equivalence between (F2) and (F3) follows from Lemma \ref{Coro17}.
The equivalence between (F1) and (F3) follows from 
the discussion only for the correctness in the proof of Theorem \ref{TH4}.
\end{proof}

\section{Quantum SPIR protocols with preshared entanglement with user}\Label{S6-3}
This section introduces EASPIR protocols and presents our results for EASPIR protocols, which implies our results for CQSPIR protocols.

\subsection{Formulation}
Modifying the CQSPIR setting by allowing prior entanglement between 
the user and the servers, we introduce 
a quantum SPIR protocol with preshared entanglement with user.
Since this problem setting employs entanglement assistance, 
this protocol is called an entanglement-assisted SPIR (EASPIR) protocol.
As illustrated in Fig. \ref{fig:SPIR} (b), an EASPIR protocol is defined as Protocol \ref{Flow4}.

\begin{Protocol}[H]                  
\caption{EASPIR protocol}         
\Label{Flow4}      
\begin{algorithmic}
\STEPONE
\textbf{Preparation}:
The user has a quantum system $\cD_{U}$ before the protocol
and 
the state of the quantum system $\cD_1' \otimes \cdots \otimes \cD_{\vN}'\otimes \cD_U$ 
is initialized as the initial state $\rho_{\mathrm{prev}}$.
The random seed 
$\mathsf{Enc}_{\mathrm{SR}} (U_S) = R = (R_1,\ldots, R_n)^T 
\in \cR = \cR_1\times \cdots \times \cR_{\snn}$ is defined as the same way as 
Protocol \ref{Flow3}.
\STEPTWO
\textbf{User's encoding}:
This step is done in the same way as Protocol \ref{Flow3}.
\STEPTHREE
\textbf{Servers' encoding}:
This step is done in the same way as Protocol \ref{Flow3}.
\STEPFOUR
\textbf{Decoding}:
For a subset $\cA \in \fA$,
the user decodes the message from the received state from servers $\cA$
by a decoder, which is defined as a POVM
$\mathsf{Dec}(K,Q^{(K)},\cA) \coloneqq \{ {Y}_{K,Q^{(K)},\cA}(w) \mid  w\in [\smm] \}$ on 
$\cD[\cA]\otimes \cD_U$
depending on the variables $K$ and $Q^{(K)}$. 
The user outputs the measurement outcome $W$ as the retrieval result.
\end{algorithmic}
\end{Protocol}

\begin{figure}
\begin{center}
  \includegraphics[width=0.8\linewidth]{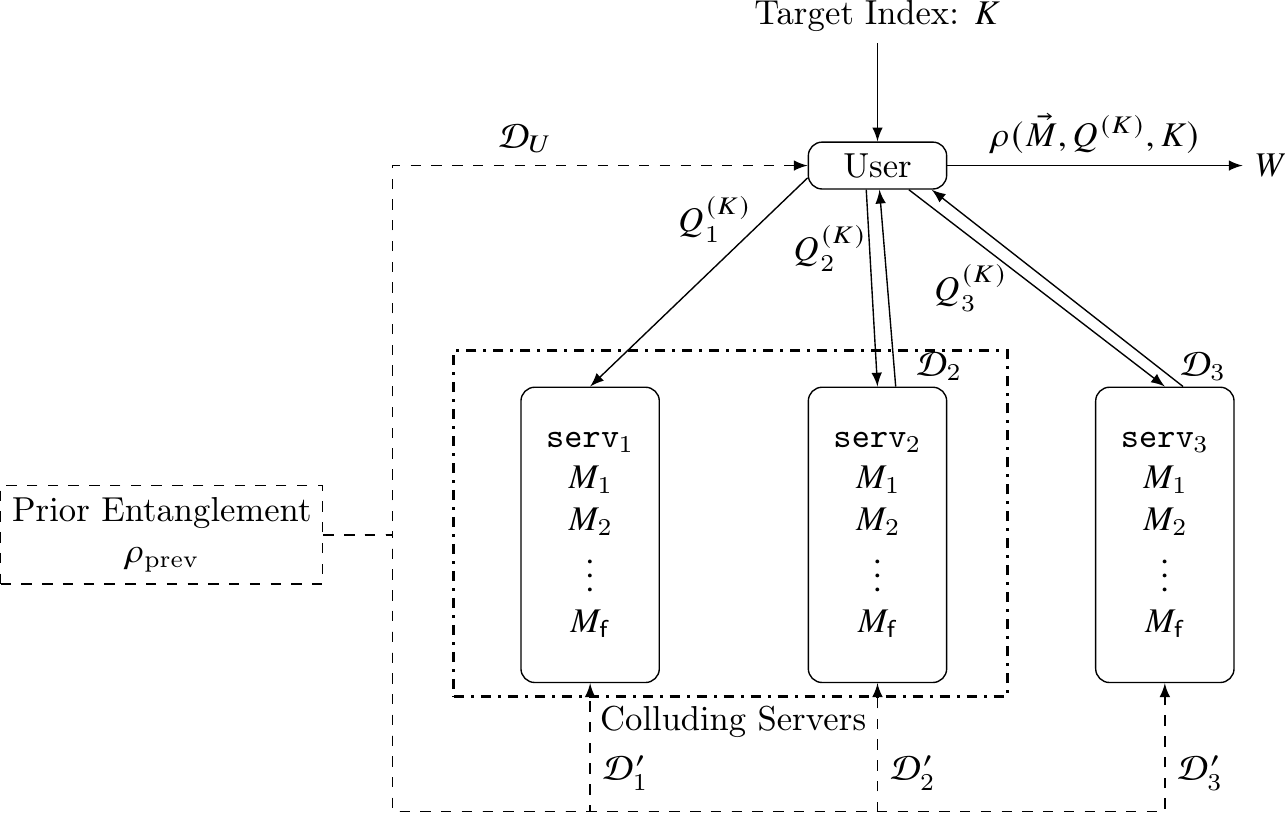}
  \end{center}
\caption{Entanglement-assisted (EA) SPIR protocols where Sever 1 and Server 2 collude and Server 2 and Server 3 respond to the user.
}   \label{fig:EASPIR}
\end{figure}

The $(\ACC,\REJ)$-security for 
an EASPIR protocol is defined in the same way as Definition \ref{Def5} as follows.

\begin{defi}\Label{Def5-B}
For an access structure $(\ACC, \REJ)$  on $[\snn]$,
	an EASPIR protocol defined as Protocol \ref{Flow4} 
	is called $(\ACC,\REJ)$-secure
	if 
	the following conditions are satisfied.
\begin{itemize}
\item \textbf{Correctness}:
For any $\cA \in \fA$, $k \in [\sff]$, and $\vec{m}=(m_1,\ldots, m_{\sff})^T \in [\smm]^\sff$,
the relation 
    \begin{align*}
\Tr \rho( \vec{m},q,k) (Y_{k,q,\cA}(m_k)\otimes I_{\cA^c} )=1
    \end{align*}
holds when $q$ is any possible query $Q^{(K)}$.

\item \textbf{User Secrecy}:
The distribution of $(Q_j^{(k)})_{j \in \cB}$ does not depend on $k \in [\sff]$ for any $\cB \in \fB$.

\item \textbf{Server Secrecy}:
We fix $K=k$, $M_k=m_k$, and $Q^{(K)}=q$. 
Then, the state $\rho( (m_1,\ldots, m_{\sff})^T,q,k)$ does not depend on $(m_j)_{j \neq k} \subset \cM^{\sff-1}$.
\end{itemize}
\end{defi}

In particular, when the system $\cD_1' \otimes \cdots \otimes \cD_{\vN}'$ has the same dimension as $\cD_U$
and the encoded states are maximally entangled states on
$\cD_1' \otimes \cdots \otimes \cD_{\vN}'$ and $\cD_U$, 
the EASPIR protocol is called fully EASPIR (FEASPIR) protocol. 
Indeed, it is difficult to consider quantum-quantum SPIR (QQPIR)
in a similar way as QQSS.
Remark \ref{RTTY} in Section \ref{S8} explains its reason.


\subsection{Linear protocols with preshared entanglement with user}
Given a $2\snn\times \sy$ matrix $G$ and
a $2\snn \times \sff \sx $ matrix $Q^{(K)}$ with index $K$,
we define
an FEASPIR protocol as follows.
The set $\cM$ is given as $\FF_q^{\sx}$.
We set $\cD_j$, $\cD_j'$ and $\cD_{U,j}$ as ${\cal H}$ for $j=1, \ldots,\snn$, and define
$\cD[\cA]$, ${\cal D}_U$, and
${\cal D}_U[\cA]$ as
$\otimes_{j \in \cA}\cD_j$,
$\otimes_{j=1}^{\snn}\cD_{U,j}$, and 
$\otimes_{j \in \cA}\cD_{U,j}$, respectively, for any subset $\cA \subset [\snn]$.
Hence, ${\cal D}_U$ is ${\cal H}^{\otimes \snn}$.

\begin{table}[t]
\caption{Symbols for FEASPIR and EASPIR protocols}
\Label{symbol2}
\begin{center}
\begin{tabular}{|c|c|c|}
\hline
symbol & meaning & definition   \\
\hline
$\cD_{U,j}$ & User's $j$-th system &  $\cH$ \\
\hline
 ${\cal D}_U$&User's whole system &   $\otimes_{j=1}^{\snn}\cD_{U,j}=\cH^{\otimes n}$ \\
\hline
${\cal D}_U[\cA]$ & User's system with subset $\cA$ &$\otimes_{j\in \cA}\cD_{U,j}$   \\
\hline
\end{tabular}
\end{center}
\end{table}

Then, we define the linear FEASPIR protocol with $G,Q^{(K)}$ as Protocol \ref{protocol4}.
\begin{Protocol}[H]                  
\caption{Linear FEASPIR protocol with $G,Q^{(K)}$}         
\Label{protocol4}      
\begin{algorithmic}
\STEPONE
\textbf{Preparation}:
We set the initial state $\rho_{\mathrm{prev}}$ on 
$\cD_1 \otimes \cdots \otimes \cD_{\snn} \otimes {\cal D}_U$
to be $|\Phi\rangle$.
Let $U_S$ be a random variable subject to the uniform distribution on $\FF_q^{\sy}$, which is called
the {\em random seed} for servers.
By using the random seed $U_S$,
the shared randomness $R = (R_1,\ldots, R_{2\snn})^T $ is generated as
$R_j :=G_j U_S$ for $j=1, \ldots, 2\snn$.
The randomness $R$ is distributed so that $j$-th server contains $R_j$ and $R_{\snn +j}$
for $j=1, \ldots, \snn$.
\STEPTWO
\textbf{User's encoding}:
The user 
	randomly encodes the index $K$ to classical queries 
	$Q^{(K)}:=(Q_1^{(K)},\ldots,Q_{\red{2}\snn}^{(K)})^T$, \red{which 
is an ${\bar{\snn}} \times \sff \sx $ random matrix for $k=1, \ldots, \sff$.
The user sends $Q_j^{(K)},Q_{\snn+j}^{(K)}$ to the $j$-th server $\mathtt{serv}_j$.}
\STEPTHREE
\textbf{Servers' encoding}:
The $j$-th server $\mathtt{serv}_j$ applies unitary 
$\mathsf{W}( Q_j^{(K)} \vec{m}+R_j , Q_{\snn+j}^{(K)} \vec{m}+R_{\snn+j})$ on ${\cal D}_j$. 
\STEPFOUR
\textbf{Decoding}:
For a subset $\cA \in \fA$,
the user takes partial trace on $\cD_{U}[\cA^c]$, and makes measurement on the basis
$\{\mathbf{W}_{\cA}(y) |\phi\rangle^{\otimes |\cA|}\}_{y \in \FF_q^{2|\cA|}}$.
Based on the obtained outcome, 
the user outputs the measurement outcome $m$ as the retrieval result.
\end{algorithmic}
\end{Protocol}

\begin{lemm}\Label{L2}
The following conditions for the matrix $G$ and the query $Q^{(K)}$ are equivalent.
\begin{description}
\item[(K1)]
The linear FEASPIR protocol with $G,Q^{(K)}$
is $(\ACC,\REJ)$-secure.
\item[(K2)]
The linear CSPIR protocol with $G,Q^{(K)}$
is $(\bar{\ACC},\bar{\REJ})$-secure.
\end{description}
\end{lemm}
\begin{proof}
 Given a subset $\cA \subset [\snn]$ and $x \in \FF_q^{{\bar{\snn}}}$, we have
 \begin{align}
& \Tr_{(\cA,U)^c} \mathbf{W}_{[\snn]}( x) (|\Phi\rangle \langle \Phi|)
 \mathbf{W}_{[\snn]}^\dagger( x)  \nonumber \\
 =&
\mathbf{W}_{\cA}(P_{\bar{\cA}} x)  |\phi\rangle \langle \phi|
 \mathbf{W}_{\cA}^\dagger(P_{\bar{\cA}} x) \otimes \rho_{mix}^{\otimes |
 \red{\cA^c}|}.
 \end{align}
Hence, the above state can be identified with the 
classical information $ P_{\bar{\cA}} x$.
That is, the analysis on the correctness and the secrecy 
of the linear FEASPIR protocol with $G,Q^{(K)}$
follows from those 
of the linear CSPIR protocol with $G,Q^{(K)}$.
\end{proof}

Next, we assume that the $2\snn\times \sy$
matrix $G$ is written as $(G^{(1)},G^{(2)})$ with 
a $2\snn\times \sy_1$ self-column-orthogonal matrix $G^{(1)}$
and
a $2\snn\times \sy_2$ matrix $G^{(2)}$.
Then, to discuss the relation with CQSPIR protocols,
we modify Protocol \ref{protocol4} as follows.
The initial state $|\Phi\rangle$ is replaced by $ |\Phi[0,G^{(1)}]\rangle$.
The random seed $U_S$ is written as $(U_{S,1},U_{S,2})$ with
$U_{S,1} \in \FF_q^{\by_1}$ and $U_{S,2} \in \FF_q^{\by_2}$ so that
$G U_S= G^{(1)} U_{S,1}+G^{(2)} U_{S,2}$.
The applied unitary 
$\mathsf{W}( Q_j^{(K)} \vec{m}+G_j U_S , Q_{\snn+j}^{(K)} \vec{m}+G_{\snn+j} U_S)$ 
on ${\cal D}_j$ is replaced by 
$\mathsf{W}( Q_j^{(K)} \vec{m}+G_j^{(2)} U_{S,2} , Q_{\snn+j}^{(K)} \vec{m}+G_{\snn+j}^{(2)} U_{S,2})$. 
In this protocol, the end-user's space $\cD_E$ is given as $\cD_E[0,G^{(1)}]$.
This protocol is formally written as Protocol \ref{protocol5}, and is called
the linear EASPIR protocol with $G^{(1)},G^{(2)},Q^{(K)}$.
When $\sy_2$ is zero, i.e., $G_j^{(2)}$ is not given, 
the linear EASPIR protocol with $G^{(1)},\emptyset,Q^{(K)}$ 
does not require shared randomness among servers, and is called
the randomless linear EASPIR protocol with $G^{(1)},Q^{(K)}$.

\begin{Protocol}[H]                  
\caption{Linear EASPIR protocol with $G^{(1)},G^{(2)},Q^{(K)}$}        
\Label{protocol5}      
\begin{algorithmic}
\STEPONE
\textbf{Preparation}:
We set ${\cal D}_U$ to be $\cH^{\otimes \snn} $.
We set the initial state $\rho_{DU}$ on 
$\cD_1 \otimes \cdots \otimes \cD_{\snn} \otimes {\cal D}_U$
to be $|\Phi[0,G^{(1)}]\rangle$.
Let $U_{S,2}$ be a random variable subject to the uniform distribution on $\FF_q^{\sy_2}$.
The shared randomness $R = (R_1,\ldots, R_{2\snn})^T $ is generated as
$R_j :=G_j^{(2)} U_{S,2}$ for $j=1, \ldots, 2\snn$.
The randomness $R$ is distributed so that $j$-th server contains $R_j$ and $R_{\snn +j}$
for $j=1, \ldots, \snn$.
\STEPTWO
\textbf{User's encoding}:
The user 
	randomly encodes the index $K$ to classical queries 
	$Q^{(K)}:=(Q_1^{(K)},\ldots,Q_{\red{2}\snn}^{(K)})^T$, \red{which 
is an ${\bar{\snn}} \times \sff \sx $ random matrix for $k=1, \ldots, \sff$.
The user sends $Q_j^{(K)},Q_{\snn+j}^{(K)}$ to the $j$-th server $\mathtt{serv}_j$.}
\STEPTHREE
\textbf{Servers' encoding}:
The $j$-th server $\mathtt{serv}_j$ applies unitary 
$\mathsf{W}( Q_j^{(K)} \vec{m}+R_j , Q_{\snn+j}^{(K)} \vec{m}+R_{\snn+j})$ on ${\cal D}_j$. 
\STEPFOUR
\textbf{Decoding}:
For a subset $\cA \in \fA$,
the user takes partial trace on $\cD_{U}[\cA^c]$, and 
makes the measurement  given by the POVM
$\{\mathbf{W}_{\cA}(y) (
\Tr_{\cD[\cA^c],\cD_{U}[\cA^c]}
|\Phi[0,G^{(1)}]\rangle \langle \Phi[0,G^{(1)}]|)
\mathbf{W}_{\cA}^\dagger(y)\}_{y \in \FF_q^{2|\cA|}}$.
Based on the obtained outcome, 
the user outputs the measurement outcome $m$ as the retrieval result.
\end{algorithmic}
\end{Protocol}

\begin{theo}\Label{TH2}
Given a $2 \snn\times \sy_1$ self-column-orthogonal matrix $G^{(1)}$,
a $2 \snn \times \sy_2$ matrix $G^{(2)}$, and query $Q^{(K)}$,
the following conditions are equivalent.
\begin{description}
\item[(L1)]
The linear EASPIR protocol with $G^{(1)},G^{(2)},Q^{(K)}$
is $(\ACC,\REJ)$-secure.
\item[(L2)]
The linear FEASPIR protocol with $(G^{(1)},G^{(2)}),Q^{(K)}$
is $(\ACC,\REJ)$-secure.
\item[(L3)]
The linear CSPIR protocol with $(G^{(1)},G^{(2)}),Q^{(K)}$
is $(\bar{\ACC},\bar{\REJ})$-secure.
\end{description}
\end{theo}

Theorem \ref{TH2} can be shown in the same way as Theorem \ref{TH1}.
That is, 
the equivalence between (L2) and (L3) follows from Lemma \ref{L2}, and 
the equivalence between (L1) and (L2) can be shown in the same way as Theorem \ref{TH1}.

In the same way as Proposition \ref{ZNO2},
we have the following proposition.
\begin{prop}\Label{ZNOSPIR}
In Protocol \ref{protocol5},
even when STEP 4 is replaced by another decoder,
the decoder can be simulated by the decoder given in STEP 3.
That is, once STEPs 1, 2, and 3 are given in Protocol \ref{protocol5CQ},
without loss of generality, we can assume that our decoder is given as STEP 4.
\end{prop}

In addition, the linear FEASPIR protocol with $G,Q^{(K)}$ 
is called
the standard linear FEASPIR protocol with $(G,F)$ 
when the query $Q^{(K)}$ is given as \eqref{CAP1} by using $F$.
Under the same condition,
the linear EASPIR protocol with $G^{(1)},G^{(2)},Q^{(K)}$
is called the standard linear EASPIR protocol with $(G^{(1)},G^{(2)},F)$.
In particular, 
when $\sy_2=0$, the protocol does not have the matrix $G^{(2)}$, i.e., 
protocol does not have the random variable $U_{S,2}$.
Such a protocol is called the randomless standard linear EASPIR protocol with $(G^{(1)},F)$.

Therefore, we obtain the following corollary by combining 
Theorem \ref{TH2} and Proposition \ref{P2}.

\begin{coro}\Label{Cor2}
Given a $2 \snn \times \sy_1$ self-column-orthogonal matrix $G^{(1)}$,
a $2 \snn \times \sy_2$ matrix $G^{(2)}$, and 
a $2 \snn \times \sx$ matrix $F$, 
the following conditions are equivalent.
\begin{description}
\item[(M1)]
The standard linear EASPIR protocol with $(G^{(1)},G^{(2)},F)$
is $(\ACC,\REJ)$-secure.
\item[(M2)]
The standard linear FEASPIR protocol with $((G^{(1)},G^{(2)}),F)$
is $(\ACC,\REJ)$-secure.
\item[(M3)]
The standard linear CSPIR protocol with $((G^{(1)},G^{(2)}),F)$
is $(\bar{\ACC},\bar{\REJ})$-secure.
\item[(M4)]
The matrix $((G^{(1)},G^{(2)}),F)$ is an $(\bar{\ACC},\bar{\REJ})$-MMSP.
\end{description}
\end{coro}

When $\sy_1$ is ${\bar{\snn}}/2=\snn$,
as discussed in Section \ref{S6-1-3}, 
$\cD_{E}[0,G^{(1)}]$ is a one-dimensional system. Hence, 
the state $|\Phi[0,G^{(1)}]\rangle$ 
is a product state.
Therefore, 
Protocol \ref{protocol5} essentially
coincides with the linear CQSPIR protocol with $G^{(1)},G^{(2)},Q^{(K)}$.
Since a (standard) linear CQSPIR protocol with $G^{(1)},G^{(2)},Q^{(K)}$
($(G^{(1)},G^{(2)},F)$) is a special case of 
a (standard) linear EASPIR protocol with $G^{(1)},G^{(2)},Q^{(K)}$ ($(G^{(1)},G^{(2)},F)$),
the relations among Conditions (L1), (L3), (M1), (M3), and (M4)
yields Theorem \ref{Cor3CQ}.
In the same way, Proposition \ref{ZNOSPIR}
implies Proposition \ref{ZNOSPIRCQ}.

Combining Corollary \ref{Cor2} and Theorem \ref{TH3}, we obtain the following corollary as a generalization of Corollary \ref{Coro7CQ}

\begin{coro}\Label{Coro7}
When 
$\snn \ge \srr > \stt >0$,
there exists an $(\srr,\stt,\snn)$-secure EASPIR protocol with rate $2(\srr-\stt)/\snn$.
\end{coro}

Remember that the rate of CQSPIR cannot exceed $1$ due to the condition $\stt \ge \snn/2$.
This constraint always holds beyond the condition in Corollary \ref{Coro7} 
because CQSPIR does not have shared entanglement.
Since EASPIR has shared entanglement, the rate of CQSPIR exceeds $1$ by removing 
the condition $\stt \ge \snn/2$,
which can be considered as an advantage of EASPIR over CQSPIR.
Further, we have the following lemma, which implies Lemma \ref{LL5CQ}.

\begin{lemm}\Label{LL5}
When we apply the conversion given in Theorem \ref{theo:SPIRtoNSS}
to the standard linear FEASPIR protocol with $(G,F)$,
the resultant EASS protocol is the linear FEASS protocol with $(G,F)$.
\end{lemm}

\begin{proof}
First, we calculate the share of the resultant EASS protocol.
Combining 
Protocol \ref{Flow5} and
Protocol \ref{protocol5},
we find that it is calculated as
\begin{align}
& \otimes_{j=1}^{\snn} 
\mathsf{W}( Q_j^{(1)} \vec{m}+G_{j} U_S , Q_{\snn+j}^{(1)} \vec{m}+G_{\snn+j} U_S)
  |\Phi\rangle \nonumber \\
\stackrel{(a)}{=}& \otimes_{j=1}^{\snn} 
\mathsf{W}( F_j E_1 \vec{m}+G_{j} U_S , F_{\snn+j} E_1 \vec{m}+G_{\snn+j} U_S)  |\Phi\rangle \nonumber \\
=& \otimes_{j=1}^{\snn} 
\mathsf{W}( F_j m_1+G_j U_S , F_{\snn+j} m_1+G_{\snn+j} U_S )  |\Phi\rangle \nonumber \\
=&
\mathsf{W}_{[\snn]}( F m_1+G U_S)  |\Phi\rangle,\Label{XWA}
\end{align}
where $(a)$ follows from \eqref{CAP1}.
The RHS of \eqref{XWA} is  the same as the share of  the linear FEASS protocol with $(G,F)$.

The decoding of the resultant EASS protocol is 
the standard linear EASPIR protocol with $(G,F)$, which is given as follows.
For a subset $\cA \in \fA$,
the user makes measurement on the basis
$\{\mathbf{W}_{\cA}(y) |\phi\rangle^{|\cA|}\}_{y \in \FF_q^{2|\cA|}}$.
Based on the obtained outcome, 
the user outputs the measurement outcome $m$ as the retrieval result.
It is the same as the decoder of the linear FEASS protocol with $(G,F)$ presented as 
Protocol \ref{protocol1}. 
\end{proof}

In fact, the same property holds for the standard linear EASPIR protocol with 
$((G^{(1)},G^{(2)}),F)$, which includes a standard linear CQSPIR protocol as a special case.

\section{Conversion from EASPIR protocol to EASS protocol}\Label{ZMTA}
\subsection{Conversion from general EASPIR protocol}
As stated in \cite{SH2022}, CSPIR protocol can be converted to CSS protocol.
In this subsection, we present how EASPIR protocol is converted to EASS protocol.
Protocol \ref{Flow5} shows the converted protocol from an EASPIR protocol. 
Protocol \ref{Flow5} contains the conversion from an CQSPIR protocol to CQSS protocol
as the special case when $\cD_{E}=\cD_U$ is one-dimensional.

\begin{Protocol}[H]                  
\caption{EASS protocol converted from EASPIR protocol}
\Label{Flow5}      
\begin{algorithmic}
\STEPONE
\textbf{Preparation}:
The dealer chooses the quantum systems $\cD_{D}$ to be $\cD_1' \otimes \cdots \otimes \cD_{\snn}' $,
and the end-user chooses the quantum system $\cD_{E}$ to be $\cD_U$, respectively.
They share the state $\rho_{\mathrm{prev}}$
as the state $\rho_{DE}$ on the joint quantum system $\cD_{D}\otimes \cD_{E}$
before the protocol.
\STEPTWO
\textbf{Share generation}:
The dealer set $K$ to be $1$.
Depending on the message $m \in \cM$,
the dealer applies the TP-CP map
$\otimes_{j=1}^{\snn}\mathsf{Enc}_{\mathrm{serv}_j} (\vec{m},Q_j^{(1)}, R_j) $
to the system $\cD_{D}=\cD_1' \otimes \cdots \otimes \cD_{\snn}' $, where
$\vec{m}=(m,0,\ldots, 0)$.
The dealer obtains the system $\cD_1 \otimes \cdots \otimes \cD_{\snn} $.
The resultant joint state on $\cD_1 \otimes \cdots \otimes \cD_{\snn} \otimes \cD_E$ is
$\rho( \vec{M},Q^{(1)},1) $.
The dealer sends the $j$-th share system $\cD_j$ to the $j$-th player.
\STEPTHREE
\textbf{Decoding}:
For a subset $\cA \in \fA$,
the end-user decodes the message from the received state from players $\cA$
by the decoder defined as the POVM
$\mathsf{Dec}(1,Q^{(1)},\cA) \coloneqq \{ {Y}_{1,Q^{(1)},\cA}(w) \mid  w\in [\smm] \}$ on 
$\cD[\cA] \otimes \cD_E$.
\end{algorithmic}
\end{Protocol}

\begin{theo} \Label{theo:SPIRtoNSS}
When the original EASPIR protocol is $(\REC,\COL)$-secure,
the converted EASS protocol via Protocol \ref{Flow5}
is $(\REC,\COL)$-secure.
\end{theo}

\begin{proof}
The correctness of the converted EASS protocol follows from 
the correctness of the original EASPIR protocol.

The secrecy of the converted EASS protocol follows from 
User secrecy and Server secrecy of the original EASPIR protocol
 in the following way.
Since 
User secrecy of the original EASPIR protocol guarantees that 
$\prod_{j \in \cB}Q_j^{(1)}$ cannot be distinguished from 
$\prod_{j \in \cB}Q_j^{(k)}$ for $k \neq 1$, we have
$\otimes_{j \in \cB} \mathsf{Enc}_{\mathrm{serv}_j} (\vec{m},Q_j^{(1)}, R_j)
=\otimes_{j \in \cB} \mathsf{Enc}_{\mathrm{serv}_j} (\vec{m},Q_j^{(k)}, R_j)$.
Thus,
\begin{align}
&\Tr_{(\cB,U)^c}\rho( \vec{m},Q^{(1)},1)  \nonumber \\
=& \otimes_{j \in \cB} \mathsf{Enc}_{\mathrm{serv}_j} (\vec{m},Q_j^{(1)}, R_j)
 (\Tr_{(\cB,U)^c}\rho_{\mathrm{prev}})  \nonumber \\
=&\otimes_{j \in \cB} \mathsf{Enc}_{\mathrm{serv}_j} (\vec{m},Q_j^{(k)}, R_j)
 (\Tr_{(\cB,U)^c}\rho_{\mathrm{prev}})  \nonumber \\
 =& \Tr_{(\cB,U)^c}\rho( \vec{m},Q^{(k)},k) .
 \end{align}
Server secrecy of the original EASPIR protocol guarantees that
the state $\Tr_{(\cB,U)^c}\rho( \vec{m},Q^{(k)},k)$ is independent of $m_1$,
which implies the secrecy of the converted EASS protocol.
\end{proof}

\subsection{Conversion from linear EASPIR protocol}
We consider how the linear EASPIR protocol with $G^{(1)},G^{(2)},Q^K$
is converted to a linear EASS protocol.  
When the conversion protocol, Protocol \ref{Flow5} is applied to Protocol \ref{protocol4}, 
we have the following protocol.

\begin{Protocol}[H]                  
\caption{Linear EASS protocol converted from linear EASPIR protocol with $G^{(1)},G^{(2)},Q^K$}         
\Label{protocol-conv} 
\begin{algorithmic}
\STEPONE
\textbf{Preparation}:
We set the initial state $\rho_{DE}$ on ${\cal D}_D \otimes {\cal D}_E[0,G^{(1)}]$
to be $ |\Phi[0,G^{(1)}]\rangle$.
\STEPTWO
\textbf{Share generation}:
The dealer prepares a uniform random variable $U_{D,2} \in \FF_q^{\sy_2} $.
For $m \in \cM$, the dealer applies 
$\mathbf{W}_{[\snn]}( Q_j^{(1)} (m,0,\ldots,0)+G_j^{(2)} U_{D,2})$ on ${\cal D}_j$. 
That is,
the encoding operation $\Gamma[m]$ \red{on ${\cal D}_D$} is defined as
\begin{align}
\Gamma[m](\rho):= \sum_{u_{D,2} \in \FF_q^{\sy_2}} 
\frac{1}{q^{\sy_2}}
\mathbf{W}_{[\snn]}( Q^{(1)} (m,0,\ldots,0)+G^{(2)} u_{D,2}) \rho 
\mathbf{W}_{[\snn]}^\dagger( Q^{(1)} (m,0,\ldots,0) +G^{(2)} u_{D,2}).
\end{align}
The shares are given as parts of the state 
$ \Gamma[m](|\Phi[0,G^{(1)}]\rangle \langle \Phi[0,G^{(1)}]|)$.
\STEPTHREE
\textbf{Decoding}:
For a subset $\cA \in \fA$,
the end-user takes partial trace on $\cD_{E}[\cA^c]$, and 
makes the measurement  given by the POVM
$\big\{\mathbf{W}_{\cA}(z) (
\Tr_{\cD[\cA^c],\cD_{E}[\cA^c]}
|\Phi[0,G^{(1)}]\rangle \langle \Phi[0,G^{(1)}]|)
\mathbf{W}_{\cA}^\dagger(z)\big\}_{z \in \FF_q^{2|\cA|}}$.
Based on the obtained outcome, the end-user recovers $m$.
\end{algorithmic}
\end{Protocol}

Next, we consider the case when the query $Q^K$ has the standard form \eqref{CAP1}.
In this case, 
the uniform random number $U_Q$ in \eqref{CAP1}
is rewritten as $(U_{Q,1},U_{Q,2})$ by using the uniform random numbers 
$U_{Q,1}$ and $U_{Q,2})$ on $\FF_q^{\sy_1}$ and $\FF_q^{\sy_2}$.
Hence, 
$Q^{(1)} (m,0,\ldots,0)+G^{(2)} U_{D,2}$ is rewritten as
\begin{align}
& Q^{(1)} (m,0,\ldots,0)+G^{(2)} U_{D,2}
=Fm +(G^{(1)},G^{(2)}) U +G^{(2)} U_{D,2} \nonumber\\
=& Fm +G^{(1)} U_{Q,1}+G^{(2)} (U_{Q,2}+U_{D,2}). \Label{MCUZ}
\end{align}
Since $\bW_{[\snn]}(G^{(1)} U_{Q,1})$ does not change the state 
$ |\Phi[0,G^{(1)}]\rangle$,
the application of \eqref{MCUZ} is equivalent to 
the application of $Fm +G^{(2)} (U_{Q,2}+U_{D,2})$, which is 
Step 2 of Protocol \ref{protocol2}.
Therefore, we find that 
the standard linear EASPIR protocol with $(G^{(1)},G^{(2)},F)$
is converted to the linear EASS protocol with $(G^{(1)},G^{(2)},F)$ via 
the conversion protocol, Protocol \ref{Flow5}.
That is, the standard linear EASPIR protocols with $(G^{(1)},G^{(2)},F)$
have one-to-one correspondence with 
linear EASS protocols

When we restrict our protocols to standard linear CQSPIR protocols,
we have the following lemma.

\begin{lemm}\Label{LL5CQ}
When we apply the conversion given in Theorem \ref{theo:SPIRtoNSS}
to the standard linear CQSPIR protocol with $(G^{(1)},G^{(2)},F)$,
the resultant CQSS protocol is the linear CQSS protocol with $(G^{(1)},G^{(2)},F)$.
\end{lemm}

When the CQSPIR protocol is a standard linear CQSPIR protocol,
the converted CQSS protocol is characterized by the same matrices.
That is, the reverse conversion is possible in this case.

\section{Example of unified construction of protocols}\Label{Sec-Ex}

\begin{table*}[t]
\begin{center}
\caption{Characterizations for matrices used for respective protocols}
\begin{tabular}{|c|c|c|c|}
    \hline
        & $G^{(1)}$ & $G^{(2)}$ & $F$\\
    \hline
linear standard FEASPIR   & 
\multicolumn{2}{c|}{\multirow{2}{*}{one $2\snn \times \sy$ matrix}}
 &  \multirow{2}{*}{$2\snn \times \sx$} \\
linear FEASS (general form) & \multicolumn{2}{c|}{} &\\
    \hline
linear standard EASPIR &  $2\snn \times \sy_1$     & \multirow{2}{*}{$2\snn \times \sy_2$}  
&  \multirow{2}{*}{$2\snn \times \sx$} \\
linear EASS   (general form)  &self-column-orthogonal  & & \\
    \hline
$(\srr,\stt,\snn)$-secure linear standard EASPIR  
&  $2\snn \times \sy_1$     & \multirow{2}{*}{$2\snn \times (2\stt-\sy_1)$}  
&  \multirow{2}{*}{$2\snn \times (2\srr-2\stt)$} \\
$(\srr,\stt,\snn)$-secure linear EASS & self-column-orthogonal  & & \\
    \hline
linear standard CQSPIR (general form)   &  $2\snn \times \snn$     
& \multirow{2}{*}{$2\snn \times \sy_2$}  &  \multirow{2}{*}{$2\snn \times \sx$}     \\
linear CQSS & self-column-orthogonal  & & \\
\hline
$(\srr,\stt,\snn)$-secure linear standard CQSPIR    
&  $2\snn \times \snn$     & \multirow{2}{*}{$2\snn \times (2\stt-\snn)$}  
&  \multirow{2}{*}{$2\snn \times (2\srr-2\stt)$} \\
$(\srr,\stt,\snn)$-secure linear CQSS & self-column-orthogonal  & & \\
    \hline
\multirow{2}{*}{linear QQSS  (general form)}  &  $2\snn \times (\snn -\sx)$ 
&  \multirow{2}{*}{$2\snn \times \sy_2$} & $2\snn \times 2\sx$ \\
& self-column-orthogonal & & column-orthogonal to $G^{(1)}$ \\
    \hline
\multirow{2}{*}{$(\srr,\stt,\snn)$-secure linear QQSS}   &  $2\snn \times (\snn -\srr+\stt)$ 
&  \multirow{2}{*}{$2\snn \times (\stt +\srr-\snn)$} & $2\snn \times 2(\srr-\stt)$ \\
& self-column-orthogonal & & column-orthogonal to $G^{(1)}$ \\
    \hline
\end{tabular}
\end{center}
\end{table*}

In this section, we give an example of MMSP which can be converted to 
randomless linear CQSS, randomless linear QQSS, randomless standard linear CQSPIR
randomless linear EASS, randomless linear QQSS, 
and randomless standard linear EASPIR protocols.
Since the construction of the respective protocols are given in the above sections,
we construct only MMSPs in this section.

\begin{exam}
We choose 
\begin{align}
G^{(1)}: =
\begin{pmatrix}
1 & 0  \\
1 & 0\\
2 & 2 \\
0 & 1  \\
0 & 1  \\
0 & 2  
\end{pmatrix} \in  \FF_3^{6\times 2},\quad
F: =
\begin{pmatrix}
2 & 0  \\
1 & 0\\
1 & 2 \\
1 & 0  \\
0 & 2 \\
1 & 2  
\end{pmatrix} \in  \FF_3^{6\times 2}.
\end{align}
The column vectors of $(G^{(1)},F)$ are linearly independent.
Then, 
$G^{(1)}$ is self-column-orthogonal, and
$F$ is column-orthogonal to $G^{(1)}$.
Hence, according to Protocol \ref{protocol1CQ},
the randomless linear CQSS protocol with $(G^{(1)},F)$ can be constructed.
Also, according to Protocol \ref{protocol6},
the randomless linear QQSS protocol with $(G^{(1)},F)$ can be constructed.
In addition,
according to Protocol \ref{protocol5CQ} with \eqref{CAP1},
the randomless standard linear CQSPIR protocol with $(G^{(1)},F)$ can be constructed.
Also,
 the randomless linear EASS protocol with $(G^{(1)},F)$
 and 
 the randomless standard linear EASPIR protocol with $(G^{(1)},F)$
can be constructed
according to Protocols \ref{protocol2} and \ref{protocol5}, respectively.

We choose $\fA = \{ \{1,2\} , \{2,3\},\{1,2,3\} \}$ and
$\fB = \{ \emptyset, \{1\}, \{2\}, \{3\} \}$.
Since 
\begin{align}
P_{\overline{\{1,2\}}} (G^{(1)}, F)
=
\begin{pmatrix}
1 & 0 & 2 & 0  \\
1 & 0 & 1 & 0\\
0 & 1 & 1 & 0 \\
0 & 1 & 0 & 2
\end{pmatrix},~
P_{\overline{\{2,3\}}} (G^{(1)}, F)
=
\begin{pmatrix}
1 & 0 & 1 & 0  \\
2 & 2 & 1 & 2\\
0 & 1 & 0 & 2 \\
0 & 2 & 1 & 2
\end{pmatrix}
\end{align}
are invertible,
the MMSP $(G^{(1)}, F)$ accepts $\fA$.
Since the matrices
\begin{align}
P_{\overline{\{1\}}} G^{(1)}
=
P_{\overline{\{2\}}} G^{(1)}
=
\begin{pmatrix}
1 & 0   \\
0 & 1 
\end{pmatrix},\quad
P_{\overline{\{3\}}} G^{(1)}
=
\begin{pmatrix}
2 & 2   \\
0 & 2
\end{pmatrix},
\end{align}
are invertible,
the MMSP $(G^{(1)}, F)$ rejects $\fB$.
Hence, $(G^{(1)}, F)$ is an $(\overline{\fA},\overline{\fB})$-MMSP.
Thus, 
the randomless linear CQSS protocol with $(G^{(1)},F)$ and
the randomless standard linear CQSPIR protocol with $(G^{(1)},F)$ are
$(\overline{\fA},\overline{\fB})$-secure due to 
Theorem~\ref{Cor1}
and Theorem~\ref{Cor3CQ}, respectively.
Also,
the randomless linear EASS protocol with $(G^{(1)},F)$ and
the randomless standard linear EASPIR protocol with $(G^{(1)},F)$ are
$(\overline{\fA},\overline{\fB})$-secure due to 
Theorem~\ref{TH1}
and Theorem~\ref{TH2}, respectively.
In addition, due to Theorem~\ref{TH4QQ}, 
the randomless linear QQSS protocol with $(G^{(1)},F)$
is $(\overline{\fA},\overline{\fB})$-secure, and its decoder 
is given as $\overline{\Gamma}[\Pi]$ defined in \eqref{XMT},
where the POVM $\Pi= \{\Pi_{m}\}_{m\in \FF_q^{2 \sx}}$ is the decoder of 
the randomless linear EASS protocol with $(G^{(1)},F)$.
\end{exam}

\begin{exam}
Next, we choose $G^{(1),*}$ and $F^*$ as
\begin{align}
G^{(1),*}: =
\begin{pmatrix}
1 & 0  & 0\\
1 & 0 & 0\\
2 & 2 & 2\\
0 & 1 &0 \\
0 & 1  &2\\
0 & 2  &2
\end{pmatrix} \in  \FF_3^{6\times 3},\quad
F^{*}: =
\begin{pmatrix}
2 & 0  \\
1 & 1\\
1 & 2 \\
0 & 0  \\
0 & 2 \\
0 & 2  
\end{pmatrix} \in  \FF_3^{6\times 2}.
\end{align}
The column vectors of $(G^{(1),*},F^*)$ are linearly independent.
Then, $G^{(1),*}$ is self-column-orthogonal.
Hence, according to Protocol \ref{protocol1CQ},
the randomless linear CQSS protocol with $(G^{(1),*},F^*)$ can be constructed.
Also, according to Protocol \ref{protocol5CQ} with \eqref{CAP1},
the randomless standard linear CQSPIR protocol with $(G^{(1),*},F^*)$ can be constructed.

We choose $\fA^* = \{ \{1,2,3\} \}$ and
$\fB^* = \{ \emptyset, \{1\}, \{2\}, \{3\} ,\{1,3\}\}$.
Since
\begin{align}
P_{\overline{\{1,3\}}} G^{(1),*}
=
\begin{pmatrix}
1 & 0 & 0 \\
2 & 2 & 2\\
0 & 1 & 0 \\
0 & 2 & 2 
\end{pmatrix},\quad
P_{\overline{\{1,3\}}} F^{*}
=
\begin{pmatrix}
2 & 0 \\
1 & 2 \\
0 & 0 \\
0 &  2 
\end{pmatrix},
\end{align}
the column vectors of $P_{\overline{\{1,3\}}} F^*$
is written as linear sums of 
the column vectors of $P_{\overline{\{1,3\}}} G^{(1),*}$.
Also,
since the rank of $P_{\overline{\{2\}}} G^{(1),*}$ is $2$,
the column vectors of $P_{\overline{\{2\}}} F^*$
is written as linear sums of 
the column vectors of $P_{\overline{\{2\}}} G^{(1),*}$.
Hence, the MMSP $(G^{(1),*}, F^*)$ rejects $\fB^*$.
Since the column vectors of $(G^{(1),*}, F^*)$ are linearly independent,
the MMSP $(G^{(1),*}, F^*)$ accepts $\fA^*$.
Hence, $(G^{(1),*}, F^*)$ is an $(\overline{\fA^*},\overline{\fB^*})$-MMSP.
Thus, 
the randomless linear CQSS protocol with $(G^{(1),*},F^*)$ and
the randomless standard linear CQSPIR protocol with $(G^{(1),*},F^*)$ are
$(\overline{\fA^*},\overline{\fB^*})$-secure
due to 
Theorem~\ref{Cor1} and Theorem~\ref{Cor3CQ}, respectively.
\end{exam}

\begin{exam}
In this example, we give an example of MMSP which can be converted to 
randomless linear EASS and randomless standard linear EASPIR protocols that overperform 
linear CQSS and linear CQSPIR protocols, respectively.

Let $p$ be a prime. 
We define the $2p \times 2$ matrix $G^{**}=(g_{j,k})$ and
the $2p \times 2$ matrix $F^{**}=(f_{j,k})$ over $\FF_p$ as 
$g_{j,1}=1 $,
$g_{j+p,1}=0 $,
$g_{j,2}=0 $,
$g_{j+p,2}=1 $,
$f_{j,1}=j-1 $,
$f_{j+p,1}=0 $,
$f_{j,2}=0 $,
$f_{j+p,2}=j-1 $ for $j=1, \ldots, p$.
The column vectors of $(G^{**},F^{**})$ are linearly independent.
Then, $G^{(1)}$ is self-column-orthogonal.

When $p >2$, $(G^{**},F^{**})$ cannot be used for
a randomless linear CQSS protocol,
a randomless linear QQSS protocol, nor
a randomless standard linear CQSPIR protocol.
But, 
the randomless linear EASS protocol with $(G^{**},F^{**})$ and
the randomless standard linear EASPIR protocol with $(G^{**},F^{**})$
can be constructed according to Protocols \ref{protocol2} and \ref{protocol5}, respectively.

We choose $\ACC^{**}= \{ \cA \subset [p] \mid |\cA| \geq 2 \}$
and $\REJ^{**}= \{ \cB \subset [p] \mid |\cB| \leq 1 \}$.
For $i<j \in [p]$, we have
\begin{align}
P_{\overline{\{i\}}} (G^{**} F^{**})=
\left(
\begin{array}{cccc}
1 & 0 & i-1&0 \\
0 & 1 & 0 & i-1
\end{array}
\right), \quad
P_{\overline{\{i,j\}}} (G^{**} F^{**})=
\left(
\begin{array}{cccc}
1 & 0 & i-1&0 \\
1 & 0 & j-1&0 \\
0 & 1 & 0 & i-1\\
0 & 1 & 0 & j-1
\end{array}
\right).
\end{align}
The first relation shows that the MMSP $(G^{**}, F^{**})$ rejects $\fB^{**}$,
and 
the second relation shows that
the MMSP $(G^{**}, F^{**})$ accepts $\fA^{**}$.

Thus, the randomless linear EASS protocol with $(G^{**},F^{**})$ and
the randomless standard linear EASPIR protocol with $(G^{**},F^{**})$ are
$(\overline{\fA}^{**},\overline{\fB}^{**})$-secure due to 
Theorem~\ref{TH1} and Theorem~\ref{TH2}, respectively.
The randomless linear EASS protocol with $(G^{**},F^{**})$ and
the randomless standard linear EASPIR protocol with $(G^{**},F^{**})$ 
have the rate $2$.
\end{exam}

\begin{figure}[htbp]
\begin{center}
  \includegraphics[width=0.8\linewidth]{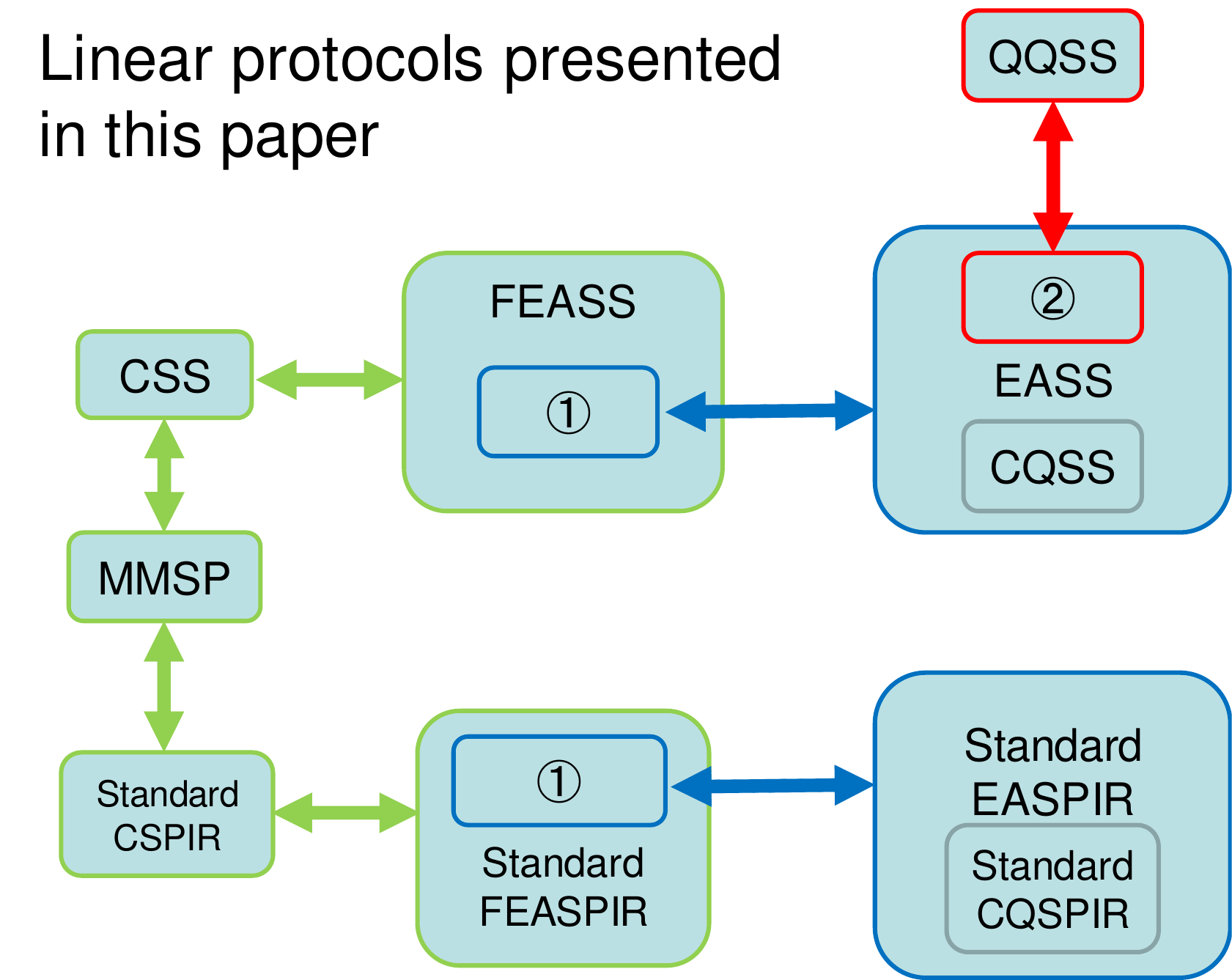}
  \end{center}
\caption{One-to-one relations among various linear protocols.
In this figure, the word ``linear'' is omitted.
Arrows of each color show 
a one-to-one relation among several protocols.
$\ctext{1}$ shows the restriction that $G=(G^{(1)},G^{(2)})$ and $G^{(1)}$ is self-column-orthogonal.
$\ctext{2}$ shows the restriction that $G^{(2)}=\emptyset$ and $F$ is 
column-orthogonal to $G^{(1)}$.}
\Label{FF1C}
\end{figure}   

\section{Conclusion} \Label{sec:conclusion}
We have characterized CQSS and QQSS protocols and CQSPIR protocols
under general access structure by using MMSP with symplectic structure.
These characterizations yield 
ramp type of CQSS and QQSS protocols and CQSPIR protocols with general qualified set,
which were not studied sufficiently until this paper.
Also, these characterizations yield interesting constructions QQMDS codes.
However, the derivation of these characterizations cannot be derived from 
simple application of similar relations in the classical case.
To overcome this problem,
we have introduced EASS and EASPIR protocols.
Since these two types of protocols can be converted to classical protocols,
we have easily derived 
their relation with general access structure and MMSP
while these analyses require column-orthogonality for MMSP.
Fortunately, 
CQSS and QQSS protocols and CQSPIR protocols can be considered as
special cases of EASS and EASPIR protocols, respectively.
That is,
the relation among these settings is summarized as Fig. \ref{FF1C}.
In addition, 
we have shown the existence of desired types of MMSP in Appendix, which implies
the existence of 
CQSS, QQSS, and CQSPIR protocols parameterized by two threshold parameters $\stt$ and $\srr$.

For this discussion, as subclasses of EASS and EASPIR protocols,
we have newly introduced \red{linear} EASS and \red{linear} EASPIR protocols
and the symplectification for an access structure.
In particular, we have focused on linear FEASS and FEASPIR protocols
because they are directly linked to 
linear classical protocols as Lemmas \ref{L2U} and \ref{L2}
thanks to the orthogonality of generalized Bell basis.
Such a simple structure has never appeared in 
CQSS and QQSS protocols and CQSPIR protocols.
Under the self-column-orthogonality for the matrix $G^{(1)}$,
linear FEASS and FEASPIR protocols are converted to
linear EASS and EASPIR protocols as Theorems \ref{TH1} and \ref{TH2}.
Since the classical linear protocols are linked to MMSP as Proposition \ref{prop1} and Lemma \ref{P2}, 
linear EASS and EASPIR protocols are linked to MMSP via the above relations.
Since CQSS and CQSPIR protocols are special classes of EASS and EASPIR protocols,
CQSS and CQSPIR protocols are characterized by using MMSP in this way.

However, the relation with QQSS is more complicated.
To establish the relation between QQSS and EASS protocols, we have introduced new relation between
dense coding and quantum state transmission.
It was known that noiseless quantum state transmission implies 
dense coding with zero error.
However, no existing study clarified whether dense coding protocol with zero error 
yields noiseless quantum state transmission.
In this paper, we have constructed a concrete protocol for 
noiseless quantum state transmission from dense coding protocol with zero error.
That is, we constructed a decoder for quantum state transmission with zero error
from a decoder for dense coding protocol with zero error as Lemma \ref{L5}.
Also, we have derived the equivalence relation between the mutual information 
between dense coding and quantum state transmission as Lemma \ref{L6}.
Using these relations, 
we have made the conversion between QQSS and EASS protocols as Theorem \ref{TH4}.
Also, we have pointed out that a special class of QQSS protocols yields QQMDS codes, which are often called
quantum MDS codes.
In addition, as Remark \ref{R5}, we have sown that 
any stabilizer code can be characterized as the performance of 
QQSS protocols in our method.
{\bf Overall, our main contribution can be summarized as 
revealing the relation between EASS and EASPIR protocols and
the symplectification for an access structure, which is 
a hidden simple structure behind CQSS, QQSS, and CQSPIR protocols.}

Although we have constructed various types of MMSP with 
column-orthogonality,
these constructions are based on algebraic extension similar to \cite{SH20,CH17}.
In contrast, existing studies \cite{HG} discussed how small size of field can realize QQMDS codes under certain condition.
Therefore, it is an interesting future study to find efficient constructions of 
various types of MMSP with 
column-orthogonality depending on two threshold parameters $\stt$ and $\srr$.
This is because these constructions are essential for constructing 
our linear protocols.
In addition, the existing study \cite[Section IV-B]{HVA} discussed SPIR with quantum noisy 
multiple access channel.
Since noisy setting is realistic, it is another interesting study to extend our results to 
the setting with quantum noisy channels.

\section*{Acknowledgments}
MH is supported in part by the National Natural Science Foundation of China (Grant No. 62171212),
Guangdong Provincial Key Laboratory (Grant No. 2019B121203002),
and 
a JSPS Grant-in-Aids for Scientific Research (A) No.17H01280.
SS is supported by JSPS Grant-in-Aid for JSPS Fellows No. JP20J11484 and Lotte Foundation Scholarship.

\appendices


\section{Preparation for proofs of Theorems \ref{TH3} and \ref{TH6}}\Label{A1}
This appendix prepares several lemmas to be used in our proofs of 
Theorems \ref{TH3CQ}, \ref{TH3}, and \ref{TH6}.
For this aim, we prepare the following lemma.
\begin{lemm}\Label{LL11}
We consider a $(\sdd+1) \times (\sdd+1)$ matrix
$A=(a_{i,j})_{1\le i\le \sdd+1, 1\le j\le \sdd+1 }$ over a finite field $\FF_q[e]$ to satisfy the following conditions.
(i) The $\sdd \times \sdd$ matrix $(a_{i,j})_{1\le i\le \sdd, 1\le j\le \sdd }$ 
is an invertible matrix.
(ii) The components $a_{i,j}$ except for $a_{\sdd+1,\sdd+1 }$ belong to $\FF_q$.
Then, the $(\sdd+1) \times (\sdd+1)$ matrix
$A$ is invertible.
\end{lemm}

\begin{proof}
We show the desired statement by contradiction.
We assume that $A$ is not an invertible matrix.
We denote the $\sdd+1$ column vectors of $A$ 
by $a^1, \ldots, a^{\sdd+1}$.
Since the $\sdd$ column vectors of 
$a^1, \ldots, a^{\sdd}$
are linearly independent, 
there exist $\sdd$ elements
$\beta_1, \ldots, \beta_{\sdd}$ of $\FF_q[e] $ such that
\begin{align}
\sum_{j=1}^{\sdd} \beta_j a_{i,j}
=a_{i,\sdd+1} \Label{LPA}
\end{align}
for $i=1, \ldots, \sdd+1$.
The finite field $\FF_q[e] $ is a vector space over the finite field $\FF_q$
generated by $\alpha^0=1, \alpha^1, \ldots, \alpha^n$ with a certain positive integer $n \ge 1$, where $\alpha:= a_{\sdd+1,\sdd+1}$.
We choose elements $\beta_{j,i} \in \FF_q$ such that
$\beta_j=\sum_{i=0}^{n} \beta_{j,i} \alpha^i$.
Then, \eqref{LPA} with $i=\sdd+1 $ is rewritten as
\begin{align}
\sum_{i=0}^{n} (\sum_{j=1}^{\sdd}  \beta_{j,i}  a_{\sdd+1,j}) \alpha^i
=\sum_{j=1}^{\sdd} \sum_{i=0}^{n} \beta_{j,i} \alpha^i a_{\sdd+1,j}
=a_{\sdd+1,\sdd+1}=\alpha.
\end{align}
Considering the case with $i=1$, we have
\begin{align}
\sum_{j=1}^{\sdd}  \beta_{j,1}  a_{\sdd+1,j}=1,
\end{align}
which implies that
the vector $(\beta_{j,1})_{j=1}^{\sdd}$ is a non-zero vector.
We denote the $\sdd+1$ column vectors of $A$ only with the initial $\sdd$ components
by $b^1, \ldots, b^{\sdd+1} \in \FF_q^{\sdd}$.
Since $b^1, \ldots, b^{\sdd}$ are linearly independent,
the vector $\sum_{j=1}^{\sdd}  \beta_{j,1}  b^j$ 
is a non-zero vector.

Next, we rewrite \eqref{LPA} with $i=1, \ldots, \sdd $ as
\begin{align}
\sum_{i=0}^{n} (\sum_{j=1}^{\sdd}  \beta_{j,i}  b^j )\alpha^i
=\sum_{j=1}^{\sdd} \sum_{i=0}^{n} \beta_{j,i} \alpha^i b^j
=\sum_{j=1}^{\sdd} \beta_j b^j
=b^{\sdd+1} \Label{LPA2}.
\end{align}
We focus on the case with $i=1$, which implies
\begin{align}
\sum_{j=1}^{\sdd}  \beta_{j,1}  b^j
=0.
\end{align}
However, the LHS is a non-zero vector,
we obtain contradiction.
Hence, we obtain the desired statement.
\end{proof}

\begin{lemm}\Label{LL12}
Assume that a $(\sdd +\sff) \times \sdd$ matrix
$D$ over a finite field $\FF_{q'}$ is a $(\sdd +\sff , \sdd)$-MDS code.
We consider a $(\sdd +\sff) \times \sgg$ matrix $F=(f_{i,j})_{1\le i \le \sdd +\sff,
1\le j \le \sgg}$
over the finite field $\FF_{q'}[e_1, \ldots, e_{\sdd}]$, 
We assume the following conditions.
(i) $\sff \ge \sgg$.
(ii) 
The component $f_{i,j}$ belongs to $\FF_{q'}$ when $i+j\le \sdd+\sgg$.
(iii)
The component $f_{i,j}$ belongs to 
$
\FF_{q'}[e_1, \ldots, e_{i+j- \sdd-\sgg}]
\setminus \FF_{q'}[e_1, \ldots, e_{i+j- \sdd-\sgg-1}]$
when $i+j> \sdd+\sgg$.
Then, the $(\sdd +\sff) \times (\sdd+\sgg)$ matrix 
$G:=(D,F)$ is  a $(\sdd +\sff , \sdd+\sgg)$-MDS code.
\end{lemm}

\begin{proof}
We denote the matrix $G$ as $(g_{i,j})_{1\le i \le \sdd +\sff, 1 \le j \le \sdd +\sgg}$.
We choose a strictly increasing function $\pi$ from
$[\sdd+\sgg]$ to $[\sdd+\sff]$.
We define the subset $\cA_\pi$ as
$\{\pi(1), \ldots, \pi(\sdd+\sgg)\}$.
Hence, it is sufficient to show that 
the $(\sdd +\sgg) \times (\sdd+\sgg)$ matrix 
$P_{\cA_\pi}(D,F)$ is invertible for any map $\pi$.
To show this statement, we show that
the matrix 
$G_k:=(g_{\pi(i),j})_{1\le i \le \sdd +k, 1 \le j \le \sdd +k}$
is invertible for $k=0,1, \ldots, \sgg$
by the induction for $k$.
The case with $k=0$ holds because 
$D$ is a $(\sdd +\sff , \sdd)$-MDS code.

Now, we assume that $G_k$ is invertible for $k=t-1$.
The component
$g_{\pi(t),t }=f_{\pi(t)-\sdd,t }$ is an element of 
$\FF_{q'}[e_1, \ldots, e_{\pi(t)+t- 2\sdd-\sgg}]
\setminus \FF_{q'}[e_1, \ldots, e_{\pi(t)+t- 2 \sdd-\sgg-1}]$.
Also, other components of $G_t$ belong to 
$\FF_{q'}[e_1, \ldots, e_{\pi(t)+t- 2 \sdd-\sgg-1}]$.
Hence, Lemma \ref{LL11} guarantees that $G_t$ is invertible.
\end{proof}

Now, we recall Proposition 4 of
\cite[Appendix D]{SH20}, which is a generalization of 
Appendix of \cite{CH17}.

\begin{prop}[\protect{\cite[Proposition 4]{SH20}}]\Label{PP8}
Given positive integers $l<r$ and a prime $p$,
we choose $q$ such that 
$\FF_q$ is an algebraic extension $\FF_p[e_1, \ldots, e_{k+r-l-2}]$.
We choose $\alpha_{i,j}$ as an element of 
$\FF_p[e_1, \ldots, e_{i+j-2}]\setminus \FF_p[e_1, \ldots, e_{i+j-3}]$.
We define $l$ vectors $v^1, \ldots, v^l \in \FF_{p}^r$ as 
\begin{align}
v_{i}^j:=
\left\{
\begin{array}{ll}
\delta_{i,j} & \hbox{ when } i \le l \\
\alpha_{i-l,j} & \hbox{ when } i \ge l+1.
\end{array}
\right.
\end{align}
In addition, we assume that $\alpha_{1,1}=1$.
Then, the matrix $(v^1, \ldots, v^{l})$ is an $(r,l)$-MDS code.
\end{prop}

\begin{lemm}\Label{L78}
Given positive integers $l<k<r$ and a prime $p$,
we choose $q$ such that 
$\FF_{q}$ is an algebraic extension $\FF_p[e_1, \ldots, e_{k+r-l-2}]$.
We choose $\alpha_{i,j}$ as an element of 
$\FF_p[e_1, \ldots, e_{i+j-2}]\setminus \FF_p[e_1, \ldots, e_{i+j-3}]$.
We define $k$ vectors $v^1, \ldots, v^k \in \FF_{p}^r$ as 
\begin{align}
v_{i}^j:=
\left\{
\begin{array}{ll}
\delta_{i,j} & \hbox{ when } i \le l \\
\alpha_{i-l,j} & \hbox{ when } i \ge l+1.
\end{array}
\right.
\end{align}
In addition, we assume that  $\alpha_{1,1}=1$.
Then, the matrix $(v^1, \ldots, v^{k})$ is an $(r,k)$-MDS code.
\end{lemm}

\begin{proof}
We apply Lemma \ref{LL12} to the case when
$D=(v^1, \ldots, v^l)$,
$F=(v^{l+1}, \ldots, v^{k})$,
$\sdd=l $,
$\sff=r-l$, $\sgg=k-l $,
and 
$\FF_{q'}$ is $\FF_p[e_1, \ldots, e_{l+k-3}]$.
Then, we find that 
the matrix $(v^1, \ldots, v^{k})$ is an $(r,k)$-MDS code.
\end{proof}

\begin{lemm}\Label{AMT}
We consider a finite field $\FF_p$.
We choose
a $(\sbb-\saa)\times \saa$ matrix $A_1=(a_{i,j})_{1\le i \le \sbb-\saa, 1\le j\le \saa}$ and 
a $(\sbb-\saa)\times \saa$ matrix $A_2=(a_{i,j})_{\sbb-\saa+1\le i \le 2(\sbb-\saa), 1\le j\le \saa}$
 over $\FF_p[e_1, \ldots, e_{2\sbb-\saa-1}]$ to satisfy the following conditions.
The component $a_{i,j}$ is an element of 
$\FF_p[e_1, \ldots, e_{i+j-1}]\setminus \FF_p[e_1, \ldots, e_{i+j-2}]$
for $1\le i \le 2(\sbb-\saa)$ and $1 \le j \le \saa$.
We choose a $\saa\times \saa$ matrix $A_3=(a_{i,j})_{2(\sbb-\saa)+1\le i \le 2\sbb-\saa, 1\le j\le \saa}$ 
 over $\FF_p[e_1, \ldots, e_{2\sbb-1}]$
 to satisfy the conditions;
The component $a_{i,j}$ is an element of 
$\FF_p[e_1, \ldots, e_{i+j-1}]\setminus \FF_p[e_1, \ldots, e_{i+j-2}]$
for $2(\sbb-\saa)+1\le i \le 2\sbb-\saa$ and $1 \le j \le \saa$.
The relation 
\begin{align}
&a_{2(\sbb-\saa)+j,i}+\sum_{k=1}^{\sbb-\saa} a_{k,j} a_{\sbb-\saa +k,i}
 \nonumber \\
 =&
a_{2(\sbb-\saa)+i,j}+\sum_{k=1}^{\sbb-\saa} a_{k,j} a_{\sbb-\saa +k,i}
\Label{AMP}
\end{align}
holds for $1 \le i \le \saa$ and $1 \le j \le \saa$.
Notice that the above choice is always possible.
Then, we define 
a $2\sbb \times \saa$ matrix $A$ and 
a $2\sbb \times (2\sbb-\saa)$ matrix $B$ as follows
\begin{align}
A=
\left( 
\begin{array}{c}
A_2 \\
A_3 \\
A_1 \\
I
\end{array}
\right), \quad
B=
\left( 
\begin{array}{ccc}
I & 0 & 0\\
-A_1^T & A_2^T & A_3^T \\
0 & I & 0 \\
0 & 0 & I 
\end{array}
\right).
\end{align}
In addition, we assume that $a_{1,1}=1$.
Then, we have the following conditions.
\begin{description}
\item[(N1)]
The relation $A^T J A =0$ holds.
\item[(N2)]
The relation $A^T J B =0$ holds.
\item[(N3)]
The $2\sbb \times \saa$ matrix $A$ is a 
$(2\sbb , \saa)$-MDS code.
\item[(N4)]
The $2\sbb \times (2\sbb-\saa)$ matrix $B$ is a 
$(2\sbb , (2\sbb-\saa))$-MDS code.
\end{description}
\end{lemm}

\red{In fact,
since the set $\FF_p[e_1, \ldots, e_{i+j-1}]\setminus \FF_p[e_1, \ldots, e_{i+j-2}]$
is not empty, it is possible to choose 
the component $a_{i,j}$ in the above way.}

\begin{proof}
The condition (N1) follows from the condition \eqref{AMP}.
The condition (N2) follows from the definitions of $A$ and $B$.
The condition (N3) holds if and only if
the $2\sbb \times \saa$ matrix 
$
\left( 
\begin{array}{c}
I \\
A_1 \\
A_2 \\
A_3 
\end{array}
\right)$ is a 
$(2\sbb , \saa)$-MDS code.
The latter condition follows from Proposition \ref{PP8}.
Hence, we obtain the condition (N3).
The condition (N4) holds if and only if
the $2\sbb \times (2\sbb-\saa)$ matrix 
$
\left( 
\begin{array}{ccc}
I & 0 & 0\\
0 & I & 0 \\
0 & 0 & I \\
-A_1^T & A_2^T & A_3^T 
\end{array}
\right)$ is a 
$(2\sbb , (2\sbb-\saa))$-MDS code.
The latter condition follows from Proposition \ref{PP8}.
Hence, we obtain the condition (N4).
\end{proof}

\begin{lemm}\Label{AMX}
We choose matrices $A$ and $B$ in the same way as Lemma \ref{AMT}.
We choose a prime power $q$ such that 
$\FF_{q}=\FF_p[e_1, \ldots, e_{2\sbb-1}]$.
We choose
a $(\sbb-\saa)\times \scc$ matrix $C_1=(c_{i,j})_{1\le i \le \sbb-\saa, 1\le j\le \saa}$ and 
a $(\sbb-\saa)\times \scc$ matrix $C_2=(c_{i,j})_{\sbb-\saa+1\le i \le 2(\sbb-\saa), 1\le j\le \saa}$
 over $\FF_{q}[e_1', \ldots, e_{2\sbb-\saa-1}']$ to satisfy the following conditions.
The component $c_{i,j}$ is an element of 
$\FF_{q}[e_1', \ldots, e_{i+j-1}']\setminus \FF_{q}[e_1', \ldots, e_{i+j-2}']$
for $1\le i \le 2(\sbb-\saa)$ and $1 \le j \le \scc$.
We choose a $\saa\times \scc$ matrix $C_3=(c_{i,j})_{2(\sbb-\saa)+1\le i \le 2\sbb-\saa, 1\le j\le \scc}$ to satisfy the conditions;
The component $c_{i,j}$ is an element of 
$\FF_{q}[e_1', \ldots, e_{i+j-1}']\setminus \FF_{q}[e_1', \ldots, e_{i+j-2}']$
for $2(\sbb-\saa)+1\le i \le 2\sbb-\saa$ and $1 \le j \le \scc$.

Then, we define 
a $2\sbb \times \scc$ matrix $C$ as follows
\begin{align}
C=
\left( 
\begin{array}{c}
C_2 \\
C_3 \\
C_1 \\
0
\end{array}
\right).
\end{align}
Then, we have the following conditions.
\begin{description}
\item[(N5)]
The $2\sbb \times (\saa+\scc)$ matrix $(A,C)$ is a 
$(2\sbb , (\saa+\scc))$-MDS code.
\item[(N6)]
Let $C^{(s)}$ be the matrix composed of the first $s$ column vectors of $C$.
The $2\sbb \times (\saa+s)$ matrix $(A,C^{(s)})$ is a 
$(2\sbb , (\saa+s))$-MDS code.
\item[(N7)]
The $2\sbb \times (2\sbb-\saa+\scc)$ matrix $(B,C)$ is a 
$(2\sbb , (2\sbb-\saa+\scc))$-MDS code.
\end{description}
\end{lemm}
\begin{proof}

The condition (N5) holds if and only if
the $2\sbb \times (\saa+\scc)$ matrix 
$
\left( 
\begin{array}{cc}
I & 0\\
A_1& C_1 \\
A_2& C_2 \\
A_3 & C_3
\end{array}
\right)$ is a 
$(2\sbb , \saa)$-MDS code.
We apply Lemma \ref{LL12} to the case with
$D=\left( 
\begin{array}{c}
I \\
A_1 \\
A_2 \\
A_3
\end{array}
\right)$,
$F=\left( 
\begin{array}{c}
0\\
C_1\\
C_2\\
C_3
\end{array}
\right)$, 
$l= \saa$,
$k=\saa+\scc$,
$r=2\sbb$,
and 
$\FF_{q'}=\FF_{q}[e_1',\ldots, e_{\scc-1}']$.
Then, due to the condition (N3) of Lemma \ref{AMT},
we obtain the condition (N5).

To show the condition (N6), 
we apply Lemma \ref{LL12} to the case with
$D=\left( 
\begin{array}{c}
I \\
A_1 \\
A_2 \\
A_3
\end{array}
\right)$,
$F=\left( 
\begin{array}{c}
0\\
C^{(s)}_1\\
C^{(s)}_2\\
C^{(s)}_3
\end{array}
\right)$, 
$l= \saa$,
$k=\saa+s$,
$r=2\sbb$,
and 
$\FF_{q'}=\FF_{q}[e_1',\ldots, e_{s-1}']$.
Hence, we obtain the condition (N6).

The condition (N7) holds if and only if
the $2\sbb \times (\saa+\scc)$ matrix 
$\left( 
\begin{array}{cccc}
I & 0 & 0 & C_2\\
0 & I & 0 & C_1\\
0 & 0 & I & 0 \\
-A_1^T & A_2^T & A_3^T & C_3
\end{array}
\right)$ is a 
$(2\sbb , (2\sbb-\saa))$-MDS code.
We apply Lemma \ref{LL12} to the case with
$D=\left( 
\begin{array}{ccc}
I & 0 & 0 \\
0 & I & 0 \\
0 & 0 & I \\
-A_1^T & A_2^T & A_3^T 
\end{array}
\right)$,
$F=\left( 
\begin{array}{c}
C_2\\
C_1\\
0 \\
C_3
\end{array}
\right)$, 
$l= 2\sbb-\saa$,
$k=2\sbb-\saa+\scc$,
$r=2\sbb$,
and 
$\FF_{q'}=\FF_{q}[e_1',\ldots, e_{2(\sbb-\saa)+\scc-1}']$.
Then, due to the condition (N4) of Lemma \ref{AMT},
we obtain the latter condition.
Hence, we obtain the condition (N7).
\end{proof}

\section{Proof of Theorem \ref{TH3}}\Label{A2}
To show Theorem \ref{TH3}, 
we apply Lemma \ref{AMX}
with $\sbb=\snn$, $\saa=\sy_1$, $\scc=2\srr$ that is given in Appendix \ref{A1}.
Then, we choose the matrices $G^{(1)},G^{(2)},F$ as 
$A$,
the matrix composed of the first $2\stt-\sy_1$ column vectors of $C$,
the matrix composed of the remaining $2\srr-(2\stt-\sy_1)$ column vectors of $C$,
respectively.
Due to the condition \eqref{AMP}, the matrix $G^{(1)}$ is self-column-orthogonal. 

The condition (N6) with $s=2\stt-\sy_1$ 
guarantees that $(G^{(1)},G^{(2)})$ is a $(2\snn,2\stt)$-MDS code.
This property guarantees the rejection condition with $\bar{\fB}$ with the choice $\REJ = \{ \cB \subset [\snn] \mid |\cB| \leq \stt \}$.
Also, the condition (N5)
guarantees that $( (G^{(1)},G^{(2)}),F)$ is a $(2\snn,2\srr)$-MDS code.
This property guarantees the acceptance condition with $\bar{\fA}$ with the choice 
$\fA = \{ \cA \subset [\snn] \mid |\cA| \geq \srr \}$.
Therefore, the proof of Theorem \ref{TH3} is completed.

\section{Proof of Theorem \ref{TH3CQ}}\Label{A2B}
To show Theorem \ref{TH3CQ}, 
we apply Lemma \ref{AMX}
with $\sbb=\snn$, $\saa=\snn$, $\scc=2\srr$ that is given in Appendix \ref{A1}.
Then, we choose the matrices $G^{(1)},G^{(2)},F$ as 
$A$,
the matrix composed of the first $\sy_2:=[2\stt-\snn]_+$ column vectors of $C$,
the matrix composed of the remaining $2\srr-\sy_2$ column vectors of $C$,
respectively.
Due to the condition \eqref{AMP}, the matrix $G^{(1)}$ is self-column-orthogonal. 

The condition (N6) with $s=\sy_2$ 
guarantees that $(G^{(1)},G^{(2)})$ is a $(2\snn,\snn+ \sy_2)$-MDS code.
This property guarantees the rejection condition with $\bar{\fB}$ with the choice $\REJ = \{ \cB \subset [\snn] \mid |\cB| \leq \stt \}$ because 
$\snn+\sy_2 \ge 2\stt$.
Also, the condition (N5)
guarantees that $( (G^{(1)},G^{(2)}),F)$ is a $(2\snn,2\srr)$-MDS code.
This property guarantees the acceptance condition with $\bar{\fA}$ with the choice 
$\fA = \{ \cA \subset [\snn] \mid |\cA| \geq \srr \}$.
Therefore, the proof of Theorem \ref{TH3CQ} is completed.

\section{Proof of Theorem \ref{TH6}}\Label{A3}
Since 
Condition (E3) implies Condition (E2), we show the directions
(E2) $\Rightarrow$ (E1) and (E1) $\Rightarrow$ (E3).

To show the direction
(E2) $\Rightarrow$ (E1),
we assume Condition (E2).
Due to Theorem \ref{TH4QQ},
there exists a linear QQSS protocol that satisfies the correctness with $\ACC = \{ \cA \subset [\snn] \mid |\cA| \geq \srr \}$.
If Condition (E1) does not holds, i.e.,  $\srr < (\snn+1)/2$,
there are two disjoint subsets $\cA_1\cA_2 \in \ACC$.
Hence, the players in $\cA_1$ and the players in $\cA_2$ 
 can recover the original state, which contradicts the no-cloning theorem \cite{WZ,Dieks}.
Hence, Condition (E2) implies Condition (E1). 
 
To show the direction
(E1) $\Rightarrow$ (E3),
we assume Condition (E1), and
choose $\stt':= \max(\stt, \snn-\srr)$.
Then, we choose a positive integer $s$, 
a $2\snn \times (\snn-\srr+\stt')$ self-column-orthogonal matrix $G^{(1)}$,
a $2\snn\times (\stt'+\srr-\snn)$ matrix $G^{(2)}$,
and a $2\snn\times 2 (\srr-\stt')$ matrix $F$ column-orthogonal 
to the matrix $G^{(1)}$
on $\FF_{q}$ with $q=p^s$
such that 
the matrix $(G^{(1)},G^{(2)},F)$ is an $(\srr,\stt',\snn)$-QQMMSP.
Since $\stt'\ge \stt$, this statement implies Condition (E3).

For this aim, we apply Lemma \ref{AMX}
to the case with
$\sbb= \snn$, $\saa= \snn-\srr+\stt'$,
$\scc=\stt'+\srr-\snn$ that is given in Appendix \ref{A1}.
We choose the $2\snn \times (\snn-\srr+\stt')$ matrix $G^{(1)}$, 
the $2\snn \times (\stt'-\srr+\snn)$ matrix $G^{(2)}$,
and the $2\snn \times 2 (\stt'-\stt')$ matrix $F$,
as $A$, $C$, and $ 
\left( 
\begin{array}{cc}
I & 0 \\
-A_1^T & A_2^T \\
0 & I \\
0 & 0 
\end{array}
\right)$,
respectively.
The matrix $G^{(1)}$ is self-column-orthogonal due to the condition (N1).
The matrix $F$ is column-orthogonal to $G^{(1)}$ due to the condition (N2).

The matrix $(G^{(1)},G^{(2)})$ is $(2\snn,2\stt')$-MDS code due to the condition (N5).
Since 
the $\snn+ \srr-\stt'$ column vectors of $(G^{(1)},F)$
forms the orthogonal space to the $\snn+ \srr-\stt'$ column vectors of $G^{(1)}$,
the linear space spanned by 
the $\snn+ \srr-\stt'$ column vectors of $(G^{(1)},F)$
equals 
the linear space spanned by 
the $\snn+ \srr-\stt'$ column vectors of $B$.
Hence,
the matrix $(G^{(1)},G^{(2)},F)$ is $(2\snn,2\srr)$-MDS code 
if and only if 
The matrix $(B,G^{(2)})$ is $(2\snn,2\srr)$-MDS code, which is guaranteed by  
the condition (N6).
Therefore, 
the matrix $(G^{(1)},G^{(2)},F)$ is an $(\srr,\stt',\snn)$-QQMMSP, which 
implies Condition (E3).

\end{document}

%% file: symbols.tex
\usepackage[caption = false]{subfig}

\usepackage{tikz}
\usetikzlibrary{shapes,snakes}
\usetikzlibrary{tikzmark}
\usetikzlibrary{fit,backgrounds}
\usetikzlibrary{matrix,decorations.pathreplacing,angles,quotes}
\usetikzlibrary{patterns}
\usetikzlibrary{calc}
\usetikzlibrary{positioning}
\tikzstyle{block} = [rectangle, draw, 
    minimum width=4em, text centered, rounded corners, minimum height=1em]
\tikzstyle{rec} = [rectangle, draw]
\tikzstyle{line} = [draw, -latex]
\usetikzlibrary{intersections}

\newcommand{\ctext}[1]{\raise0.2ex\hbox{\textcircled{\scriptsize{#1}}}}

\usepackage{color}
\newcommand{\myred}[1]{{\leavevmode\color{red}#1}}

\DeclareFontFamily{OMX}{MnSymbolE}{}
\DeclareFontShape{OMX}{MnSymbolE}{m}{n}{
    <-6>  MnSymbolE5
   <6-7>  MnSymbolE6
   <7-8>  MnSymbolE7
   <8-9>  MnSymbolE8
   <9-10> MnSymbolE9
  <10-12> MnSymbolE10
  <12->   MnSymbolE12}{}
\DeclareSymbolFont{mnlargesymbols}{OMX}{MnSymbolE}{m}{n}
\SetSymbolFont{mnlargesymbols}{bold}{OMX}{MnSymbolE}{b}{n}
\DeclareMathDelimiter{\llangle}{\mathopen}{mnlargesymbols}{'164}{mnlargesymbols}{'164}
\DeclareMathDelimiter{\rrangle}{\mathclose}{mnlargesymbols}{'171}{mnlargesymbols}{'171}

\newtheorem{lemm}{Lemma}
\newtheorem{defi}{Definition}
\newtheorem{exam}{Example}
\newtheorem{theo}{Theorem}

\newtheorem*{prot*}{Protocol}
\newtheorem{prop}{Proposition}

\newtheorem{coro}{Corollary}

\theoremstyle{definition}
\newtheorem{remark}{Remark}
\newtheorem*{remark*}{Remark}

\def\im{\mathop{\rm Im}}

\newcommand{\cA}{\mathcal{A}}
\newcommand{\cB}{\mathcal{B}}
\newcommand{\cC}{\mathcal{C}}
\newcommand{\cD}{\mathcal{D}}
\newcommand{\cE}{\mathcal{E}}

\newcommand{\cH}{\mathcal{H}}

\newcommand{\cM}{\mathcal{M}}

\newcommand{\cR}{\mathcal{R}}

\DeclareMathOperator{\Tr}{Tr}

\DeclareMathOperator{\tr}{tr}

\DeclareMathOperator*{\rank}{\mathrm{rank}}

\newcommand\sA{\mathsf{A}}

\newcommand\bF{\mathbf{F}}

\newcommand\bW{\mathbf{W}}

\newcommand\by{\mathbf{y}}

\newcommand\saa{\mathsf{a}}
\newcommand\sbb{\mathsf{b}}
\newcommand\scc{\mathsf{c}}
\newcommand\sdd{\mathsf{d}}

\newcommand\sff{\mathsf{f}}
\newcommand\sgg{\mathsf{g}}

\newcommand\smm{\mathsf{m}}
\newcommand\snn{\mathsf{n}}

\newcommand\srr{\mathsf{r}}

\newcommand\stt{\mathsf{t}}

\newcommand\sx{\mathsf{x}}
\newcommand\sy{\mathsf{y}}
\newcommand\sz{\mathsf{z}}

\newcommand\qprot{\Lambda}

\newcommand\PNSS{\Phi^{\smm}_{\mathrm{QNSS}}}

\newcommand\fA{\mathfrak{A}}
\newcommand\fB{\mathfrak{B}}
\newcommand\fC{\mathfrak{C}}

\newcommand\REC{\mathfrak{A}}
\newcommand\COL{\mathfrak{B}}
\newcommand\ACC{\mathfrak{A}}
\newcommand\REJ{\mathfrak{B}}

\newcommand\FF{\mathbb{F}}
\newcommand\Fq{\mathbb{F}_q}

\newcommand\vD{\mathsf{d}}
\newcommand\vN{\mathsf{n}}

\newcommand\vM{\mathsf{m}}

\DeclareMathOperator*{\supp}{\mathrm{supp}}